\documentclass[11pt]{article}

\usepackage[margin=1.25in]{geometry}
\usepackage[pdftex]{graphicx}
\usepackage[update,prepend]{epstopdf}
\usepackage[font=small,format=hang,justification=justified,labelfont=bf]{caption}
\usepackage[round,semicolon]{natbib}
\setcitestyle{notesep={; },round,aysep={},yysep={;}}
\usepackage{amsmath}
\usepackage{amsthm}
\usepackage{amsfonts}
\usepackage{amssymb}
\usepackage{dsfont}
\usepackage{bold-extra}
\usepackage{titlesec}
\usepackage[title]{appendix}
\usepackage{cases}
\usepackage{enumitem}
\usepackage{verbatim}
\usepackage{color}
\usepackage{mathtools}

\DeclarePairedDelimiter\floor{\lfloor}{\rfloor}
\allowdisplaybreaks[4]

\usepackage{multicol,multirow,array}
\usepackage{booktabs}
\usepackage{wrapfig}
\usepackage{setspace}
\usepackage{fancybox}
\usepackage{fancyvrb}
\usepackage{framed}
\usepackage{calc}
\usepackage[pdfstartview={XYZ null null 0.85}]{hyperref}
\usepackage{url}
\newcommand*{\MyLaw}{\mathrm{law}}
\newcommand*{\eqlawU}{\ensuremath{\mathop{\overset{\MyLaw}{=}}}} 
\newcommand*{\eqlaw}{\mathop{\overset{\MyLaw}{\resizebox{\widthof{\eqlawU}}{\heightof{=}}{=}}}}

\newcommand\independent{\protect\mathpalette{\protect\independenT}{\perp}}
\def\independenT#1#2{\mathrel{\rlap{$#1#2$}\mkern2mu{#1#2}}}

\DeclareMathOperator*{\plim}{lim\vphantom{p}}

\DeclareMathOperator{\E}{\mathbb{E}}
\DeclareMathOperator{\Prob}{\mathbb{P}}
\DeclareMathOperator{\Ind}{\mathds{1}}

\newcommand{\RR}{\mathbb{R}}

\makeatletter
\newsavebox\tboxa
\newsavebox\tboxb
\newlength\tdima
\newcommand*{\oversymb}{\mathpalette\@oversymb}
\newcommand*{\@oversymb}[2]{
    \sbox{\tboxa}{$\m@th#1\mathrm{#2}$}
    \setbox\tboxb\null
    \ht\tboxb\ht\tboxa
    \dp\tboxb\dp\tboxa
    \wd\tboxb\wd\tboxa
    \sbox{\tboxa}{$\m@th#1{#2}$}
    \setlength\tdima{\the\wd\tboxa}
    \addtolength\tdima{-\the\wd\tboxb}
    \sbox{\tboxb}{$\m@th#1\hskip\tdima\overline{\xusebox{\tboxb}}$}
    \rlap{\usebox\tboxb}{\usebox\tboxa}}
\newcommand*{\xusebox}[1]{\mathord{{\usebox{#1}}}}
\makeatother

\newtheorem{theorem}{Theorem}[section]
\newtheorem{proposition}[theorem]{Proposition}
\newtheorem{lemma}[theorem]{Lemma}

\newtheorem{remark}[theorem]{Remark}
\let\oldremark\remark
\renewcommand{\remark}{\oldremark\normalfont}
\numberwithin{equation}{section}

\date{}
\begin{document}

\title{A mixed Monte Carlo and PDE variance reduction method\\ for foreign exchange options under the Heston-CIR model}

\author{\scshape{andrei cozma}\thanks{\footnotesize\scshape{Mathematical Institute, University of Oxford, Oxford, OX2 6GG, UK}\newline
\hspace*{1.8em}andrei.cozma@maths.ox.ac.uk, christoph.reisinger@maths.ox.ac.uk} \and \scshape{christoph reisinger\footnotemark[1]}}

\maketitle

\begin{abstract}\noindent
In this paper, the valuation of European and path-dependent options in foreign exchange (FX) markets is considered when the currency exchange rate evolves according to the Heston model combined with the Cox-Ingersoll-Ross dynamics for the stochastic domestic and foreign short interest rates. The mixed Monte Carlo/PDE method requires that we simulate only the paths of the squared volatility and the two interest rates, while an ``inner'' Black-Scholes-type expectation is evaluated by means of a PDE. This can lead to a substantial variance reduction and complexity improvements under certain circumstances depending on the contract and the model parameters. In this work, we establish the uniform boundedness of moments of the exchange rate process and its approximation, and prove strong convergence in $L^p$ ($p\geq1$) of the latter. Then, we carry out a variance reduction analysis and obtain accurate approximations for quantities of interest. All theoretical contributions can be extended to multi-factor short rates in a straightforward manner. Finally, we illustrate the efficiency of the method for the four-factor Heston-CIR model through a detailed quantitative assessment.

\vspace{1em}\noindent
\end{abstract}

\section{Introduction}\label{sec:intro}

In FX markets, option pricing with stochastic volatility and stochastic interest rates has seen a large amount of interest in the last few years \citep{Grzelak:2011,Haastrecht:2011,Ahlip:2013}, leading to the extension of the \citet{Heston:1993} two-factor stochastic volatility model and the Sch\"obel-Zhu \citep{Schobel:1999} model to currency derivatives. Although appealing due to its simplicity, assuming constant interest rates is inappropriate for long-dated FX products, and the effect of interest rate volatility can even outweigh that of FX rate volatility for long maturities, a fact confirmed by empirical results \citep{Haastrecht:2009}. Here, the spot FX rate is defined as the number of units of domestic currency per one unit of foreign currency.

In this paper, we consider the four-factor Heston-CIR model proposed and examined in \citet{Ahlip:2013} when the volatility and the exchange rate dynamics are correlated, whereas the domestic and foreign interest rates are pairwise independent and also independent of the exchange rate and the volatility. Our motivation comes from the fact that the square root (CIR) process \citep{Cox:1985} for the variance and interest rates is widely used in the industry due to its desirable properties, such as mean-reversion and non-negativity. Under these restrictive assumptions on the independence of the Brownian drivers, the authors argue that the model is affine and derive a semi-analytical formula for the European call option price. The importance of a non-zero correlation between the FX rate and the interest rate(s) is recognized in \citet{Hunter:2005}. However, any other non-zero correlations give rise to a non-affine model, in which case we lose analytical tractability.

Therefore, we see ourselves forced to turn to numerical algorithms, and the mixed Monte Carlo/PDE solver \citep{Loeper:2009} is a good alternative to the classical Monte Carlo and finite difference methods. Monte Carlo simulation methods \citep{Glasserman:2003} can handle path-dependent features easily and scale linearly with the dimension, however a considerable number of simulations is typically required for a good accuracy. Conversely, finite difference methods incorporate early exercise features easily and provide a fast convergence for low-dimensional problems (up to three dimensions), but become intractable as the dimensionality increases. This makes Monte Carlo methods more attractive for the four-factor Heston-CIR model, and the relatively slow convergence $\mathcal{O}\big(M^{-0.5}\big)$ in the number of simulations $M$ raises the question of finding an efficient variance reduction technique.

The idea behind the mixed Monte Carlo/PDE method is to write the option values as nested conditional expectations. Then, the innermost expectation is evaluated analytically, in case of European-style options, or using finite differences, and the outer expectation by simulation. The earliest related published work belongs to \citet{Hull:1987}. The authors consider a European call option under a two-dimensional stochastic volatility model with uncorrelated asset price and volatility and prove that, conditional on the integral of the variance process, the asset price is lognormally distributed and hence the option price can be expressed as a Black-Scholes price. \citet{Willard:1997} extends the analysis of \citet{Hull:1987} to path-independent options under stochastic volatility and instantaneously correlated factors, and uses the smoothness of the ``conditional price'' to calculate price sensitivities (the Greeks). The author also employs quasi-Monte Carlo (low-discrepancy) methods to further reduce the variance of price estimates, but with no effect on the discretization bias. The mixed Monte Carlo/PDE method develops this conditioning technique -- known as conditional Monte Carlo -- and allows the combined use of Monte Carlo and finite difference methods for the valuation of path-dependent contracts.

The utility of the method can be easily recognised when pricing European-style options under the Heston-CIR model. Conditioning on the entire paths of the squared volatility and the domestic and foreign interest rates, the dynamics of the exchange rate are governed by a geometric Brownian motion with time-dependent drift and diffusion coefficients. Combined with the existence of a closed-form solution for the conditional option price, the algorithm reduces the variance in Monte Carlo simulations -- by eliminating a source of noise -- and the dimension of the problem, from four to three. The mixed Monte Carlo/PDE method has been the object of several numerical studies in the past few years, e.g., \citet{Lipp:2013}, \citet{Ang:2013}, \citet{McGhee:2014}, \citet{Dang:2015}, and various extensions to the original idea have been considered. For instance, \citet{Dang:2015} use the Hull-White type dynamics of the interest rates and condition on the variance path to find a closed-form solution for the conditional price of a European option. A few of these references examine convergence properties of the algorithm by heuristic arguments, and none of them take into account the error arising from the discretization of the variance and interest rate processes.

To the best of our knowledge, the convergence of the mixed Monte Carlo/PDE method has not yet been established, and even the literature on Monte Carlo methods under stochastic volatility is scarce. \citet{Higham:2005} considered an Euler simulation of the Heston model with a reflection fix and proved strong convergence of the stopped approximation process, as well as for a European put and an up-and-out call, by using the boundedness of payoffs. \citet{Cozma:2015} extended these results to derivatives with unbounded payoffs, stochastic-local volatility, stochastic interest rates, and exotic payoffs. Their discretization scheme coincides with the one considered in this paper at the discrete time points, but necessarily differs in the continuous-time interpolation. Some results can be utilized from the earlier work, although stronger conditions on the model parameters were needed there and therefore the results here can also be seen as an improvement for the standard (i.e., non-mixed) scheme. In particular, the new results are always guaranteed for zero or negative correlation in the Heston model. The main conceptual difference, however, is the new continuous-time interpolation using conditional drifts, which is crucial for deriving the conditional PDEs and leads to different technical challenges. Also, for path-dependent options, neither the original scheme nor the analysis applies to the present work.

In this paper, we study convergence properties of the mixed Monte Carlo/PDE method with the full truncation Euler (FTE) discretization for the squared volatility and the two interest rates, and demonstrate the efficiency of the method for a European call and an up-and-out put option. We prefer the full truncation scheme \citep{Lord:2010} because it preserves positivity, is easy to implement and is found empirically to produce the smallest bias among all Euler schemes. An interesting alternative to the FTE scheme would be the backward Euler-Maruyama (BEM) scheme \citep{Szpruch:2014}. However, the quanto correction term in the dynamics of the foreign interest rate would lead to technical challenges in the convergence analysis. The major contributions of this paper are as follows.

\begin{itemize}
	\item We establish the uniform boundedness of moments of the four-dimensional process and its approximation, and prove strong convergence in $L^p$ ($p\geq1$) of the discretization scheme. Then, we deduce the convergence of mixed Monte Carlo/PDE estimators for computing option prices and discuss possible extensions to higher-dimensional models.
	\item We carry out a thorough theoretical variance reduction analysis of the mixed Monte Carlo/PDE method and employ standard Monte Carlo -- with the log-Euler discretization -- as the reference method, noting that the analysis applies to general interest rate dynamics. In particular, we investigate how different values of the underlying model parameters affect the variance of mixed estimators.
	\item We perform a series of numerical experiments and demonstrate the convergence of the mixed Monte Carlo/PDE method under the four-factor FX model for two financial derivatives: a European call and an up-and-out put option. In addition, we establish the efficiency of the method by comparison with alternative numerical schemes and examine the sensitivity of the variance reduction factor to changes in the parameters.
\end{itemize}

The remainder of this paper is structured as follows. In Section~\ref{sec:setup}, we introduce the four-factor FX model, define the mixed simulation scheme and describe the pricing algorithm. In Section~\ref{sec:convergence}, we prove strong convergence of the exchange rate approximations and then discuss some extensions. Detailed proofs of some technical results are given in the Appendix. In Section~\ref{sec:variance}, we carry out a variance reduction analysis of the mixed Monte Carlo/PDE method for a European option. Various numerical experiments are presented and discussed in Section~\ref{sec:numerics}. Finally, Section~\ref{sec:conclusion} summarizes the results and outlines possible future work.

\section{Preliminaries}\label{sec:setup}

\subsection{The four-factor model}\label{subsec:model}

We have in mind a model in an FX market, for the spot FX rate $S$, the variance of the FX rate $v$, the domestic short interest rate $r^{d}$ and the foreign short interest rate $r^{f}$. Unless otherwise stated, in this paper, the subscripts and superscripts ``$d$'' and ``$f$'' are used to indicate domestic and foreign, respectively. Consider a filtered probability space $\left(\Omega,\mathcal{F},\{\mathcal{F}_t\}_{t\geq0},\mathbb{Q}\right)$ and suppose that the dynamics of the underlying processes are governed by the system of stochastic differential equations (SDEs) below under the domestic risk-neutral measure $\mathbb{Q}$:
\begin{align}\label{eq2.1}
	\begin{dcases}
	dS_{t} = \big(r^{d}_{t}-r^{f}_{t}\big)S_{t}dt + \sqrt{v_{t}}\hspace{.5pt}S_{t}\hspace{1pt}dW^{s}_{t} \\[4pt]
	dv_{t} \hspace{1pt} = k\big(\theta-v_{t}\big)dt + \xi\sqrt{v_{t}}\,dW^{v}_{t} \\[0pt]
	dr^{d}_{t} = k_{d}\big(\theta_{d}-r^{d}_{t}\big)dt + \xi_{d}\sqrt{r^{d}_{t}}\,dW^{d}_{t} \\[-2pt]
	dr^{f}_{t} = \big(k_{f}\theta_{f}-k_{f}r^{f}_{t}-\rho_{s\hspace{-.7pt}f}\xi_{f}\sqrt{v_{t}r^{f}_{t}}\hspace{1pt}\big)dt + \xi_{f}\sqrt{r^{f}_{t}}dW^{f}_{t}\hspace{-1pt},
	\end{dcases}
\end{align}
where $\{W^{s},W^{v},W^{d},W^{f}\}$ are correlated standard Brownian motions (BMs) under the risk-neutral measure with constant correlation matrix $\Sigma$. We consider a full correlation structure between the four Brownian drivers, which reflects movements in the financial markets more accurately and allows for better calibrations. Then, we decouple $\{W^{s},W^{f},W^{d},W^{v}\}$ and express them as linear combinations of independent Brownian motions $\{W^{1},W^{2},W^{3},W^{4}\}$. Hence, define the two vectors $W = [W^{s},W^{f},W^{d},W^{v}]^{T}$ and $\tilde{W} = [W^{1},W^{2},W^{3},W^{4}]^{T}$. As $\Sigma$ is symmetric positive definite, a standard Cholesky factorisation gives rise to an upper triangular matrix of coefficients $A=\left(a_{ij}\right)_{1\leq i,j\leq4}$\hspace{1pt}, satisfying $\Sigma = AA^{T}$, which is given below.
\begin{equation}\label{eq2.2}
\Sigma = \begin{bmatrix}
												1 & \rho_{s\hspace{-.7pt}f} & \rho_{sd} & \rho_{sv} \\
												\rho_{s\hspace{-.7pt}f} & 1 & \rho_{d\hspace{.01pt}f} & \rho_{v\hspace{-.7pt}f} \\
												\rho_{sd} & \rho_{d\hspace{.01pt}f} & 1 & \rho_{v\hspace{.2pt}d} \\
												\rho_{sv} & \rho_{v\hspace{-.7pt}f}& \rho_{v\hspace{.2pt}d} & 1
								\end{bmatrix}
\quad\text{and}\quad\;
A = \begin{bmatrix}
				  a_{11} & a_{12} & a_{13} & a_{14} \\
					0 & a_{22} & a_{23} & a_{24} \\
					0 & 0 & a_{33} & a_{34} \\
					0 & 0 & 0 & 1
				\end{bmatrix}
\end{equation}
This decomposition implies that we can choose $\tilde{W}$ so that $W=A\tilde{W}$. We can determine the matrix of coefficients by solving a system of ten equations. Assuming $\rho_{v\hspace{.2pt}d}\neq\pm1$, we find:
\begin{align*}
&a_{14}=\rho_{sv},\hspace{2pt} a_{24}=\rho_{v\hspace{-.7pt}f},\hspace{2pt} a_{34}=\rho_{v\hspace{.2pt}d},\hspace{2pt} a_{33}=\big(1-\rho_{v\hspace{.2pt}d}^{2}\big)^{\frac{1}{2}}\hspace{.5pt},\hspace{2pt} a_{13}=\big(\rho_{sd}-\rho_{sv}\rho_{v\hspace{.2pt}d}\big)\big(1-\rho_{v\hspace{.2pt}d}^{2}\big)^{-\frac{1}{2}}\hspace{.5pt}, \\[1pt]
&a_{23}=\big(\rho_{d\hspace{.01pt}f}-\rho_{v\hspace{-.7pt}f}\rho_{v\hspace{.2pt}d}\big)\big(1-\rho_{v\hspace{.2pt}d}^{2}\big)^{-\frac{1}{2}}\hspace{.5pt},\hspace{2pt}	a_{22}=\big(1-\rho_{d\hspace{.01pt}f}^{2}-\rho_{v\hspace{-.7pt}f}^{2}-\rho_{v\hspace{.2pt}d}^{2}+2\rho_{d\hspace{.01pt}f}\rho_{v\hspace{-.7pt}f}\rho_{v\hspace{.2pt}d}\big)^{\frac{1}{2}}\big(1-\rho_{v\hspace{.2pt}d}^{2}\big)^{-\frac{1}{2}}\hspace{.5pt},
\end{align*}
whereas $a_{11}$ and $a_{12}$ can be found in a similar fashion.

The quanto correction term in the drift of the foreign interest rate in \eqref{eq2.1} comes from changing from the foreign to the domestic risk-neutral measure \citep{Clark:2011}. Alternatively, we can also think of \eqref{eq2.1} as a model in an equity market with asset price process $S$, interest rate $r^{d}$ and dividend yield $r^{f}$, in which case the quanto drift adjustment term vanishes.

\subsection{The mixed simulation scheme}\label{subsec:simulation}

First, we discretize the variance and the two interest rate processes using the full truncation Euler (FTE) scheme of \citet{Lord:2010}. Consider the square root process
\begin{equation}\label{eq2.3}
dy_{t} = k_{y}(\theta_{y}-y_{t})dt + \xi_{y}\sqrt{y_{t}}\,dW^{y}_{t}.
\end{equation}
For a time interval $[0,T]$, consider a uniform grid: $\delta t=T\hspace{-.5pt}/N$, $t_{n}=n\hspace{.5pt}\delta t,\; \forall\hspace{.5pt} n\in\{0,1,...,N\}$. We introduce the time discrete auxiliary process
\begin{equation}\label{eq2.4}
\tilde{y}_{t_{n+1}} = \tilde{y}_{t_{n}} + k_{y}(\theta_{y}-\tilde{y}_{t_{n}}^{+})\delta t + \xi_{y}\sqrt{\tilde{y}_{t_{n}}^{+}}\,\delta W^{y}_{t_{n}},
\end{equation}
where $y^{+} = \max\left(0,y\right)$ and $\delta W^{y}_{t_{n}} = W^{y}_{t_{n+1}} - W^{y}_{t_{n}}$, and the time continuous interpolation
\begin{equation}\label{eq2.5}
\tilde{y}_{t} = \tilde{y}_{t_{n}} + k_{y}(\theta_{y}-\tilde{y}_{t_{n}}^{+})(t-t_{n}) + \xi_{y}\sqrt{\tilde{y}_{t_{n}}^{+}}\big(W^{y}_{t}-W^{y}_{t_{n}}\big),\; \forall\hspace{.5pt} t \in [t_{n},t_{n+1}),
\end{equation}
as suggested in \citet{Higham:2005}. Moreover, we define the non-negative processes
\begin{equation}\label{eq2.6}
Y_{t} = \tilde{y}_{t}^{+}
\end{equation}
and
\begin{equation}\label{eq2.7}
\hspace{-.5pt}\oversymb{\hspace{.5pt}Y}_{\hspace{-2.5pt}t} = \tilde{y}_{t_{n}}^{+},
\end{equation}
whenever $t \in [t_{n},t_{n+1})$. Let $\oversymb{V}$ and $\oversymb{r}^{d}$ be the FTE discretizations -- as defined in \eqref{eq2.7} -- of the variance and the domestic interest rate, respectively. Taking into account the presence of the quanto correction term in the drift of the foreign interest rate, we similarly define
\begin{equation}\label{eq2.7.1}
\tilde{r}_{t}^{f} = \tilde{r}_{t_{n}}^{f} + \Big[k_{f}\theta_{f}-k_{f}\big(\tilde{r}_{t_{n}}^{f}\big)^{+}-\rho_{s\hspace{-.7pt}f}\xi_{f}\sqrt{\tilde{v}_{t_{n}}^{+}\big(\tilde{r}^{f}_{t_{n}}\big)^{+}}\hspace{1.5pt}\Big](t-t_{n}) + \xi_{f}\sqrt{\big(\tilde{r}_{t_{n}}^{f}\big)^{+}}\hspace{.5pt}\big(W^{f}_{t}-W^{f}_{t_{n}}\big),
\end{equation}
as well as
\begin{equation}\label{eq2.7.2}
\hat{r}^{f}_{t} = \big(\tilde{r}_{t}^{f}\big)^{+}
\end{equation}
and
\begin{equation}\label{eq2.7.3}
\oversymb{r}^{f}_{t} = \big(\tilde{r}_{t_{n}}^{f}\big)^{+},
\end{equation}
whenever $t \in [t_{n},t_{n+1})$. Next, we define $\hspace{1.5pt}\oversymb{\hspace{-1.5pt}S}$, the continuous-time approximation of $S$, as the solution to the following SDE:
\begin{align}\label{eq2.8}
d\hspace{1.5pt}\oversymb{\hspace{-1.5pt}S}_{t} &= \mu_{t}\hspace{2pt}\oversymb{\hspace{-1.5pt}S}_{t}\hspace{.5pt}dt + a_{11}\sqrt{\oversymb{V}_{\hspace{-2.5pt}t}}\,\hspace{1.5pt}\oversymb{\hspace{-1.5pt}S}_{t}\hspace{1pt}dW^{1}_{t}, \\[3pt]
\mu_{t} &= \oversymb{r}^{d}_{t} - \oversymb{r}^{f}_{t} - \frac{1}{2}\hspace{-1pt}\left(1-a_{11}^{2}\right)\hspace{-.5pt}\oversymb{V}_{\hspace{-2.5pt}t} + \sum_{j=2}^{4}{a_{1\hspace{-0.5pt}j}\hspace{0.5pt}\sqrt{\oversymb{V}_{\hspace{-2.5pt}t}}\hspace{2pt}\frac{\delta W_{t}^{j}}{\delta t}}\hspace{.5pt}, \nonumber
\end{align}
where $\delta W^{j}_{t} = W^{j}_{t_{n+1}} - W^{j}_{t_{n}},\hspace{1pt}\forall\hspace{.5pt} t \in [t_{n},t_{n+1})$, hence $\mu_{t}$ is piecewise constant. For convenience, we introduce the actual and the approximated log-processes, $x=\log S$ and $X=\log\hspace{1.5pt}\oversymb{\hspace{-1.5pt}S}$.

Conditioning on the trajectories of $\big\{W^{j}\hspace{-0.5pt}, \hspace{0.5pt} j=2,3,4\big\}$, i.e., on the complete knowledge of the paths of the variance and interest rates, $\hspace{1.5pt}\oversymb{\hspace{-1.5pt}S}$ evolves like a geometric Brownian motion with time-dependent drift and diffusion coefficients. It follows from It\^o's formula that
\begin{align*}
\hspace{1.5pt}\oversymb{\hspace{-1.5pt}S}_{T} = S_{0}\exp\bigg\{\int_{0}^{T}{\Big(\oversymb{r}^{d}_{u}-\oversymb{r}^{f}_{u}-\frac{1}{2}\hspace{.5pt}\oversymb{V}_{\hspace{-2.5pt}u}\Big) du} + \sum_{j=2}^{4}{a_{1\hspace{-0.5pt}j}\int_{0}^{T}{\hspace{-.2em}\sqrt{\oversymb{V}_{\hspace{-2.5pt}u}}\,\frac{\delta W^{j}_{u}}{\delta t}\,du}} + a_{11}\int_{0}^{T}{\hspace{-.2em}\sqrt{\oversymb{V}_{\hspace{-2.5pt}u}}\,dW^{1}_{u}}\bigg\}\hspace{.5pt}.
\end{align*}
However, we know the conditional probability law of the stochastic integral on the right-hand side to be that of a normal random variable, namely
\begin{equation*}
\int_{0}^{T}{\hspace{-.2em}\sqrt{\oversymb{V}_{\hspace{-2.5pt}u}} \, dW^{1}_{u}} \;\,\eqlaw\;\, \sqrt{\int_{0}^{T}{\oversymb{V}_{\hspace{-2.5pt}u}\,du}}\hspace{1pt}\cdot\hspace{1pt}Z\hspace{.5pt},
\end{equation*}
where $Z \sim \mathcal{N}\left(0,1\right)$, so we can think of $\hspace{1.5pt}\oversymb{\hspace{-1.5pt}S}_{T}$ as a function of $Z$. Therefore, we can express it as
\begin{equation}\label{eq2.9}
\hspace{1.5pt}\oversymb{\hspace{-1.5pt}S}_{T} \,\eqlaw\, S_{0}\exp\left\{\Big(r - q - \frac{1}{2}\sigma^{2}\Big)T + \sigma\sqrt{T}\hspace{1pt}Z\right\},
\end{equation}
where
\begin{equation*}
r \,=\, \frac{1}{T}\int_{0}^{T}{\oversymb{r}^{d}_{u}\hspace{1pt}du} \,=\, \frac{1}{N}\sum_{i=0}^{N-1}{\oversymb{r}^{d}_{t_{i}}}\hspace{1pt},\hspace{1em}
\sigma^{2} \,=\, \frac{a_{11}^{2}}{T}\int_{0}^{T}{\oversymb{V}_{\hspace{-2.5pt}u}\hspace{1pt}du} \,=\, \frac{a_{11}^{2}}{N}\sum_{i=0}^{N-1}{\oversymb{V}_{\hspace{-2.5pt}t_{i}}}\hspace{1pt},
\end{equation*}
and
\begin{align*}
q &\,=\, \frac{1}{T}\int_{0}^{T}{\oversymb{r}^{f}_{u}\hspace{1pt}du} + \frac{1 - a_{11}^2}{2T}\int_{0}^{T}{\oversymb{V}_{\hspace{-2.5pt}u}\hspace{1pt}du} - \frac{1}{T}\sum_{j=2}^{4}a_{1\hspace{-0.5pt}j}\int_{0}^{T}{\hspace{-.2em}\sqrt{\oversymb{V}_{\hspace{-2.5pt}u}}\,\frac{\delta W^{j}_{u}}{\delta t}\,du} \\[2pt]
&\,=\, \frac{1}{N}\sum_{i=0}^{N-1}{\oversymb{r}^{f}_{t_{i}}} + \frac{1 - a_{11}^2}{2N}\sum_{i=0}^{N-1}{\oversymb{V}_{\hspace{-2.5pt}t_{i}}} - \frac{1}{T}\sum_{j=2}^{4}a_{1\hspace{-0.5pt}j}\sum_{i=0}^{N-1}{\sqrt{\oversymb{V}_{\hspace{-2.5pt}t_{i}}}\,\delta W^{j}_{t_{i}}}\hspace{.5pt}.
\end{align*}
Conditional on the trajectories of $\big\{W^{j}\hspace{-0.5pt}, \hspace{0.5pt} j=2,3,4\big\}$, \eqref{eq2.9} has the same law as a terminal asset price that evolves as a geometric Brownian motion with interest rate $r$, continuous dividend yield $q$ and volatility $\sigma$, all constant. Then the arbitrage-free price at time $t=0$ of a European-style option with payoff $f(S_{T})$ is the discounted expectation under the risk-neutral measure, which can be approximated using the mixed Monte Carlo/PDE method by
\begin{equation}\label{eq2.10}
\hspace{1pt}\oversymb{\hspace{-1pt}U\hspace{-.5pt}}\hspace{.5pt} = \E\left[e^{-\int_{0}^{T}{\oversymb{r}^{d}_{u}\hspace{1pt}du}}\hspace{.5pt}f(\hspace{1.5pt}\oversymb{\hspace{-1.5pt}S}_{T})\right] = \E\left[\E\left[e^{-rT}f(\hspace{1.5pt}\oversymb{\hspace{-1.5pt}S}_{T})\big|\,\mathcal{G}_{T}^{f,d,v}\right]\right],
\end{equation}
where $\big\{\mathcal{G}_{t}^{f,d,v},\hspace{1.5pt} 0\hspace{-.5pt}\leq\hspace{-.5pt}t\hspace{-.5pt}\leq\hspace{-.5pt}T\big\}$ is the natural filtration generated by the independent Brownian motions $\big\{W^{2},W^{3},W^{4}\big\}$, i.e., generated by the processes $v$, $r^{d}$, $r^{f}$ as observed until time $T$. The second equality in \eqref{eq2.10} comes from the ``tower property'' of conditional expectations. Unless otherwise stated, all expectations are under $\mathbb{Q}$. Let the approximate conditional option price be the inner expectation in \eqref{eq2.10}, which is analytically tractable for European contracts,
\begin{equation*}
\E\left[e^{-rT}f(\hspace{1.5pt}\oversymb{\hspace{-1.5pt}S}_{T}) \big| \,\mathcal{G}_{T}^{f,d,v}\right] = e^{-rT}\int_{-\infty}^{\infty}{f\big(\hspace{1.5pt}\oversymb{\hspace{-1.5pt}S}_{T}(z)\big)\hspace{.5pt}\phi(z)\hspace{1pt}dz}\hspace{.5pt},
\end{equation*}
where $\phi$ and $\Phi$ are the standard normal PDF and CDF, respectively. The conditional prices of some popular financial instruments are given below,
\begin{align}\label{eq2.11}
	\begin{dcases}
	\text{European options:} &\psi S_{0}\hspace{.5pt}e^{-qT}\Phi(\psi d_{1}) - \psi Ke^{-rT}\Phi(\psi d_{2}) \\[2pt]
	\text{Cash-or-nothing options:} &e^{-rT}\Phi(\psi d_{2}) \\[2pt]
	\text{Asset-or-nothing options:} &S_{0}\hspace{.5pt}e^{-qT}\Phi(\psi d_{1})\hspace{.5pt},
	\end{dcases}
\end{align}
where $\psi = 1$ for a call and $\psi = -1$ for a put, whereas
\begin{equation}\label{eq2.12}
d_{1,2} = \frac{\log(S_{0}/K) + \left(r - q \pm \sigma^{2}/2\right)T}{\sigma \sqrt{T}}\hspace{1pt}.
\end{equation}
The approximate option price, i.e., the outer expectation in \eqref{eq2.10}, is estimated by a Monte Carlo average over a sufficiently large number of discrete trajectories of $\big\{W^{2},W^{3},W^{4}\big\}$.

There are many other derivatives that admit a closed-form solution for the conditional price, like the power option, the chooser option or the forward-start option. However, for most path-dependent derivatives, we need to use a different approach in order to compute the conditional price, in which case we will rely on finite difference methods (see Section \ref{sec:numerics}).

We conclude this section with a discussion on our choice of conditioning. For European option pricing, an analytical formula for the inner expectation is available only when conditioning on all three factors, which results in a dimension reduction of the problem by one. The high-dimensionality of the Heston-CIR model makes this the natural choice. For path-dependent option pricing, due to the quanto correction term in the drift of the foreign rate, $v$ and $r^{f}$ are coupled and hence we cannot condition on $r^{f}$ alone. Moreover, the variance of Monte Carlo estimators due to the short rates is typically much lower than the variance due to the instantaneous squared volatility. Hence, we could alternatively condition on $r^{d}$ and solve a three-dimensional PDE for the inner expectation. On one hand, we could reach the same level of accuracy with fewer Monte Carlo sample paths than when simulating $v$ and $r^{f}$ as well. On the other hand, the computational effort grows linearly with the dimension for Monte Carlo methods and exponentially for finite difference methods. Hence, we believe that conditioning on all three factors is more efficient. However, if the exchange rate and the foreign interest rate dynamics are independent, the quanto correction term vanishes. In this case, conditioning on the two short rates would be an interesting alternative.

\section{Convergence analysis}\label{sec:convergence}

Even though weak convergence is very important in financial mathematics when estimating expectations of payoffs, strong convergence may be required for complex path-dependent derivatives and plays a key role in multilevel Monte Carlo methods \citep{Giles:2008}. In this section, we prove the strong convergence of the approximation scheme defined in \eqref{eq2.8}. Hence, we first examine exponential integrability properties of the square root process and its discretization, and then the finiteness of moments of order higher than one of the exchange rate process and its approximation.

Let $y$ be the square root process defined in \eqref{eq2.3} and let $\hspace{-.5pt}\oversymb{\hspace{.5pt}Y}$ be the piecewise constant FTE interpolant from \eqref{eq2.7}. The exponential integrability of functionals of the two processes was already discussed in Propositions 3.2 and 3.6 in \citet{Cozma:2015.2}. However, we need to adjust their second result for our approximation scheme \eqref{eq2.8} in order to establish the convergence.

\begin{lemma}\label{Lem3.1}
Let $\lambda, \mu \in \mathbb{R}$ be given, $\Delta\equiv\lambda+\frac{1}{2}\hspace{1pt}\mu^{2}$, and define the stochastic process
\begin{equation}\label{eq3.1}
\hspace{1pt}\oversymb{\hspace{-1pt}\Theta\hspace{-1pt}}\hspace{1pt}_{t} \equiv \exp\bigg\{\lambda\int_{0}^{t}{\hspace{-.5pt}\oversymb{\hspace{.5pt}Y}_{\hspace{-2.5pt}u}\hspace{1pt}du} + \mu\int_{0}^{t}{\sqrt{\hspace{-.5pt}\oversymb{\hspace{.5pt}Y}_{\hspace{-2.5pt}u}}\,\frac{\delta W^{y}_{u}}{\delta t}\hspace{1.5pt}du}\bigg\},\hspace{3pt} \forall\hspace{.5pt}t\in[0,T].
\end{equation}
If $\Delta\leq0$ and $T\geq0$, or otherwise, if $\Delta>0$ and $T\leq T^{*}$, then there exists $\eta>0$ such that
\begin{equation}\label{eq3.2}
\sup_{\delta t \in (0,\eta)}\hspace{1.5pt}\sup_{t \in [0,T]}\E\big[\hspace{1.5pt}\oversymb{\hspace{-1pt}\Theta\hspace{-1pt}}\hspace{1pt}_{t}\big] < \infty,
\end{equation}
where $T^{*}$ is given below:
\begin{enumerate}
\item{If $k_{y}\leq\xi_{y}(\mu+\sqrt{0.5\hspace{.5pt}\Delta})$,
\begin{equation}\label{eq3.3.1}
T^{*} = \frac{1}{\xi_{y}(\mu+\sqrt{2\Delta})-k_{y}}\hspace{1pt}.
\end{equation}}
\item{If $k_{y}>\xi_{y}(\mu+\sqrt{0.5\hspace{.5pt}\Delta})$,
\begin{equation}\label{eq3.3.2}
T^{*} = \frac{2(k_{y}-\mu\xi_{y})}{\xi_{y}^{2}\Delta}\hspace{1pt}.
\end{equation}}
\end{enumerate}
\end{lemma}
\begin{proof}
See Appendix \ref{sec:aux3.1}.
\end{proof}

\subsection{Moment bounds}\label{subsec:moments}

For many stochastic volatility models, moments of order higher than one can explode in finite time \citep{Andersen:2007}. This can cause significant problems in practice, for instance when computing the arbitrage-free price of an option whose payoff function has super-linear growth. The same troublesome behaviour can be observed for the Euler-Maruyama approximation of some SDEs with super-linearly growing drift or diffusion coefficients, where moments diverge in finite time \citep{Jentzen:2015}. Next, we prove the boundedness of moments of the exchange rate process and its approximation.

At this point, we assume that $\rho_{v\hspace{.2pt}d}\neq\pm1$ and that $a_{13}$ is non-zero, i.e., $\rho_{sd}\neq\rho_{sv}\rho_{v\hspace{.2pt}d}$.

\begin{proposition}\label{Prop3.2}
For $\alpha\geq1$, define the two quantities
\begin{align}
q_{0}(\alpha) &\equiv \frac{1}{2\alpha^{2}\xi^{2}a_{13}^{2}}\bigg\{\sqrt{\big[2\alpha\rho_{sv}\xi k + \alpha^{2}\xi^{2}(a_{11}^{2}+a_{12}^{2})-\alpha\xi^{2}\big]^{2} + 4\alpha^{2}a_{13}^{2}\xi^{2}k^{2}} \nonumber \\[0pt]
												 &\hspace{11.5em} - \big[2\alpha\rho_{sv}\xi k + \alpha^{2}\xi^{2}(a_{11}^{2}+a_{12}^{2})-\alpha\xi^{2}\big]\bigg\}, \label{eq3.3} \\[2pt]
q_{1}(\alpha) &\equiv q_{0}(\alpha)\Ind_{\rho_{sv}\leq\hspace{1pt}0} \hspace{1pt}+\hspace{1pt} \min\left\{q_{0}(\alpha)\hspace{.5pt},\,\frac{k}{\alpha\rho_{sv}\xi}\right\}\Ind_{\rho_{sv}>\hspace{1pt}0}. \label{eq3.4}
\end{align}
If the following conditions on the model parameters are satisfied,
\begin{equation}\label{eq3.5}
k > \alpha\rho_{sv}\xi + \sqrt{\alpha(\alpha-1)}\hspace{1.5pt}\xi \quad\text{and}\quad \frac{k_{d}^{2}}{2\xi_{d}^{2}} > \frac{\alpha\hspace{1.5pt}q_{1}(\alpha)}{q_{1}(\alpha)-1}\,,
\end{equation}
then there exists $\alpha_{1}>\alpha$ such that for all $\omega \in [1,\alpha_{1})$,
\begin{equation}\label{eq3.6}
\sup_{t \in [0,T]} \E\big[S_{t}^{\omega}\big] < \infty.
\end{equation}
\end{proposition}
\begin{proof}
See Appendix \ref{sec:aux3.2}.
\end{proof}

Next, assume that $a_{13}$ and $a_{14}$ are not simultaneously zero, i.e., $\rho_{sv}^{2}+\rho_{sd}^{2}\neq0$.

\begin{proposition}\label{Prop3.4} For $\alpha\geq1$, define
\begin{align}\label{eq3.22}
q_{2}(\alpha) &\equiv \bigg\{\sqrt{\big[\alpha\rho_{sv}\xi + T\alpha^{2}\xi^{2}(a_{11}^{2}+a_{12}^{2})/4-T\alpha\xi^{2}/4\big]^{2} + T\alpha^{2}\xi^{2}(a_{13}^{2}+a_{14}^{2})k} \nonumber\\[2pt]
&- \big[\alpha\rho_{sv}\xi + T\alpha^{2}\xi^{2}(a_{11}^{2}+a_{12}^{2})/4-T\alpha\xi^{2}/4\big]\bigg\}\hspace{1pt}\frac{2}{T\alpha^{2}\xi^{2}(a_{13}^{2}+a_{14}^{2})}\hspace{1pt}.
\end{align}
If the following conditions on the model parameters are satisfied,
\begin{equation}\label{eq3.23}
k > \alpha\rho_{sv}\xi + \frac{1}{4}\hspace{1pt}\alpha(\alpha-1)T\xi^{2} \hspace{-2pt}\quad\text{and}\quad \frac{2k_{d}}{T\xi_{d}^{2}} > \frac{\alpha\hspace{1.5pt}q_{2}(\alpha)}{q_{2}(\alpha)-1}\,,
\end{equation}
then there exists $\alpha_{2}>\alpha$ such that for all $\omega\in[1,\alpha_{2})$, we can find $\eta_{\omega}>0$ so that
\begin{equation}\label{eq3.24}
\sup_{\delta t \in (0,\eta_{\omega})}\hspace{1pt}\sup_{t \in [0,T]}\E\big[\hspace{1.5pt}\oversymb{\hspace{-1.5pt}S}_{\hspace{-.5pt}t}^{\omega}\big] < \infty.
\end{equation}
\end{proposition}
\begin{proof}
See Appendix \ref{sec:aux3.4}.
\end{proof}

Since the most popular FX and equity contracts grow at most linearly in FX and asset prices, and their valuation requires the computation of the expected discounted payoff under the risk-neutral measure, it is useful to study finiteness of moments under discounting. Let $R$ be the discounted exchange rate process,
\begin{align}\label{eq3.36}
R_{t} = S_{0}\exp\bigg\{\hspace{-2pt}-\hspace{-1pt}\int_{0}^{t}{\Big(r^{f}_{u}+\frac{1}{2}\hspace{1pt}v_{u}\Big)du} + \int_{0}^{t}{\sqrt{v_{u}}\,dW_{u}^{s}}\bigg\},
\end{align}
and let $\hspace{1.5pt}\oversymb{\hspace{-1.5pt}R}$ be its continuous-time approximation,
\begin{align}\label{eq3.37}
\hspace{1.5pt}\oversymb{\hspace{-1.5pt}R}_{t} = S_{0}\exp\bigg\{\hspace{-2pt}-\hspace{-1pt}\int_{0}^{t}{\Big(\oversymb{r}^{f}_{u}+\frac{1}{2}\hspace{1pt}\oversymb{V}_{\hspace{-2.5pt}u}\Big)du} + a_{11}\int_{0}^{t}{\sqrt{\oversymb{V}_{\hspace{-2.5pt}u}}\,dW^{1}_{u}} + \sum_{j=2}^{4}{a_{1\hspace{-0.5pt}j}\int_{0}^{t}{\sqrt{\oversymb{V}_{\hspace{-2.5pt}u}}\,\frac{\delta W^{j}_{u}}{\delta t}\,du}}\bigg\}.
\end{align}

\begin{proposition}\label{Prop3.6}
Let $\alpha\geq1$. If $T<T^{*}$, there exists $\alpha_{1}>\alpha$ such that for all $\omega \in [1,\alpha_{1})$,
\begin{equation}\label{eq3.39.1}
\sup_{t \in [0,T]} \E\big[R_{t}^{\omega}\big] < \infty.
\end{equation}
If $\alpha>1$ and $T\geq T^{*}$, then
\begin{equation}\label{eq3.39.2}
\E\big[R_{T}^{\hspace{.5pt}\alpha}\big] = \infty.
\end{equation}
If $\alpha=1$, then $T^{*}=\infty$, whereas if $\alpha>1$, then $T^{*}$ is given below:
\begin{enumerate}
\item{If $k<\alpha\rho_{sv}\xi-\sqrt{\alpha(\alpha-1)}\hspace{1.5pt}\xi$,
\begin{equation}\label{eq3.38.1}
T^{*} = \frac{1}{\nu(\alpha)}\log\left(\frac{\alpha\rho_{sv}\xi-k+\nu(\alpha)}{\alpha\rho_{sv}\xi-k-\nu(\alpha)}\right),
\end{equation}
where
\begin{equation*}
\nu(\alpha)=\sqrt{(\alpha\rho_{sv}\xi-k)^{2}-\alpha(\alpha-1)\xi^{2}}\hspace{1pt}.
\end{equation*}}
\item{If $k=\alpha\rho_{sv}\xi-\sqrt{\alpha(\alpha-1)}\hspace{1.5pt}\xi$,
\begin{equation}\label{eq3.38.2}
T^{*} = \frac{2}{\alpha\rho_{sv}\xi-k}\hspace{1pt}.
\end{equation}}
\item{If $\alpha\rho_{sv}\xi-\sqrt{\alpha(\alpha-1)}\hspace{1.5pt}\xi<k<\alpha\rho_{sv}\xi+\sqrt{\alpha(\alpha-1)}\hspace{1.5pt}\xi$,
\begin{equation}\label{eq3.38.3}
T^{*} = \frac{2}{\hat{\nu}(\alpha)}\left[\frac{\pi}{2}-\arctan\left(\frac{\alpha\rho_{sv}\xi-k}{\hat{\nu}(\alpha)}\right)\right],
\end{equation}
where
\begin{equation*}
\hat{\nu}(\alpha)=\sqrt{\alpha(\alpha-1)\xi^{2}-(\alpha\rho_{sv}\xi-k)^{2}}\hspace{1pt}.
\end{equation*}}
\item{If $k\geq\alpha\rho_{sv}\xi+\sqrt{\alpha(\alpha-1)}\hspace{1.5pt}\xi$,
\begin{equation}\label{eq3.38.4}
T^{*} = \infty\hspace{.5pt}.
\end{equation}}
\end{enumerate}
\end{proposition}
\begin{proof}
See Appendix \ref{sec:aux3.6}.
\end{proof}

When the domestic and the foreign interest rates are constant, Proposition \ref{Prop3.6} examines moment boundedness in the Heston model. It is an extension of Proposition 3.1 in \citet{Andersen:2007} from bounds on moments of order $\alpha>1$ to bounds on all moments of order $\omega\in[\alpha,\alpha_{1})$, for some $\alpha_{1}>\alpha\geq1$. We use this result to prove the strong convergence of the discretized discounted spot FX rate process.

\begin{proposition}\label{Prop3.7} Let $\alpha\geq1$. If $T<T^{*}$, there exists $\alpha_{2}>\alpha$ and $\eta_{\omega}>0$ such that for all $\omega \in [1,\alpha_{2})$,
\begin{equation}\label{eq3.41}
\sup_{\delta t \in (0,\eta_{\omega})}\hspace{1pt}\sup_{t \in [0,T]}\E\big[\hspace{1.5pt}\oversymb{\hspace{-1.5pt}R}_{\hspace{-.5pt}t}^{\omega}\big] < \infty,
\end{equation}
where $T^{*}$ is given below:
\begin{enumerate}
\item{If $k<\alpha\rho_{sv}\xi+\frac{1}{2}\hspace{1pt}\sqrt{\alpha(\alpha-1)}\hspace{1.5pt}\xi$,
\begin{equation}\label{eq3.40.1}
T^{*} = \frac{1}{\alpha\rho_{sv}\xi+\sqrt{\alpha(\alpha-1)}\hspace{1.5pt}\xi-k}\hspace{1pt}.
\end{equation}}
\item{If $k\geq\alpha\rho_{sv}\xi+\frac{1}{2}\hspace{1pt}\sqrt{\alpha(\alpha-1)}\hspace{1.5pt}\xi$,
\begin{equation}\label{eq3.40.2}
T^{*} = 
\begin{dcases}
\omit\hfill$\infty$\hfill\hspace{-8pt}&, \hspace{4pt}\text{if}\hspace{5pt} \alpha=1 \\[2pt]
\frac{4(k-\alpha\rho_{sv}\xi)}{\alpha(\alpha-1)\xi^{2}}\hspace{-8pt}&, \hspace{4pt}\text{if}\hspace{5pt} \alpha>1.
\end{dcases}
\end{equation}}
\end{enumerate}
\end{proposition}
\begin{proof}
See Appendix \ref{sec:aux3.7}.
\end{proof}

To the best of our knowledge, the boundedness of moments of discretization schemes for the Heston model and extensions thereof had not been established until recently \citep{Cozma:2015} -- a fact that is also mentioned in \citet{Kloeden:2012} -- and Proposition \ref{Prop3.7} is only the second to address this issue. For the Heston model, Proposition \ref{Prop3.7} can be seen as an improvement of Proposition 3.9 in \citet{Cozma:2015} due to the sharper conditions on the critical time.

\subsection{The four-dimensional system}\label{subsec:fxrate}

The strong mean square convergence of the discretized variance and domestic interest rate processes was established in Proposition 3.5 in \citet{Cozma:2015}. First of all, we prove an equivalent result for the foreign interest rate.

\begin{proposition}\label{Prop3.8.0} If $2k_{f}\theta_{f}>\xi_{f}^{2}$, then the process $\oversymb{r}^{f}$ converges strongly in $L^{2}$, i.e.,
\begin{equation}\label{eq3.42.0}
\plim_{\delta t\to0}\hspace{1pt}\sup_{t\in[0,T]}\E\Big[\big|r^{f}_{t}-\oversymb{r}^{f}_{t}\big|^{2}\Big] = 0.
\end{equation}
\end{proposition}
\begin{proof}
See Appendix \ref{sec:aux3.8.0}.
\end{proof}

As will become clear from the proof, the Feller condition $2k_{f}\theta_{f}>\xi_{f}^{2}\hspace{.5pt}$, which ensures that the process $r^{f}$ does not hit zero, allows us to control the potential growth of the absolute difference between the original and the discretized processes that comes from the sublinear correction term in the drift.

Second, we consider the logarithm of the process from \eqref{eq2.8} and examine its convergence properties. The formulae of the log-process, $x$, and its approximation, $X$, are given below.
\begin{align}
x_{t} &= x_{0} + \int_{0}^{t}{\Big(r^{d}_{u}-r^{f}_{u}-\frac{1}{2}\hspace{1pt}v_{u}\Big)du} + \int_{0}^{t}{\sqrt{v_{u}}\,dW_{u}^{s}}, \label{eq3.42}\\[1pt]
X_{t} &= x_{0} + \int_{0}^{t}{\Big(\oversymb{r}^{d}_{u}-\oversymb{r}^{f}_{u}-\frac{1}{2}\hspace{1pt}\oversymb{V}_{\hspace{-2.5pt}u}\Big)du} + a_{11}\int_{0}^{t}{\sqrt{\oversymb{V}_{\hspace{-2.5pt}u}}\,dW^{1}_{u}} + \sum_{j=2}^{4}{a_{1\hspace{-0.5pt}j}\int_{0}^{t}{\sqrt{\oversymb{V}_{\hspace{-2.5pt}u}}\,\frac{\delta W^{j}_{u}}{\delta t}\,du}}\hspace{.5pt}. \label{eq3.43}
\end{align}

\begin{proposition}\label{Prop3.8} If $2k_{f}\theta_{f}>\xi_{f}^{2}$, then the log-process converges uniformly in $L^{2}$, i.e.,
\begin{equation}\label{eq3.44}
\plim_{\delta t \to 0} \E\bigg[\sup_{t \in [0,T]}\big|x_{t}-X_{t}\big|^{2}\bigg] = 0.
\end{equation}
\end{proposition}
\begin{proof}
See Appendix \ref{sec:aux3.8}.
\end{proof}

Third, we prove convergence of the discretized spot FX rate process.

\begin{proposition}\label{Prop3.9} If $2k_{f}\theta_{f}>\xi_{f}^{2}$, then the process $\hspace{1.5pt}\oversymb{\hspace{-1.5pt}S}$ converges uniformly in probability, i.e.,
\begin{equation}\label{eq3.54}
\plim_{\delta t \to 0} \Prob\bigg(\sup_{t \in [0,T]}\big|S_{t}-\hspace{1.5pt}\oversymb{\hspace{-1.5pt}S}_{t}\big| > \epsilon \bigg) = 0\hspace{.5pt},\hspace{2pt} \forall\hspace{.5pt}\epsilon > 0.
\end{equation}
\end{proposition}
\begin{proof}
See Appendix \ref{sec:aux3.9}.
\end{proof}

\begin{theorem}\label{Thm3.10} Let $\alpha\geq1$ and assume that the following conditions are satisfied:
\begin{align}\label{eq3.60}
& k > \alpha\rho_{sv}\xi + \max\bigg\{\sqrt{\alpha(\alpha-1)}\hspace{1pt}\xi\hspace{.5pt},\hspace{1.5pt} \frac{1}{4}\hspace{1pt}\alpha(\alpha-1)T\xi^{2}\bigg\}, \nonumber\\[3pt]
&\hspace{-1em} \frac{k_{d}^{2}}{2\xi_{d}^{2}} > \frac{\alpha\hspace{1.5pt}q_{1}(\alpha)}{q_{1}(\alpha)-1}\hspace{1pt},\hspace{8pt} \frac{2k_{d}}{T\xi_{d}^{2}} > \frac{\alpha\hspace{1.5pt}q_{2}(\alpha)}{q_{2}(\alpha)-1} \hspace{4pt}\text{ and }\hspace{5pt} 2k_{f}\theta_{f}>\xi_{f}^{2}\hspace{1pt}.
\end{align}
Then the process converges strongly in $L^{\alpha}$ in the sense that
\begin{equation}\label{eq3.61}
\plim_{\delta t \to 0}\hspace{1pt}\sup_{t \in [0,T]}\E\Big[\big|S_{t}-\hspace{1.5pt}\oversymb{\hspace{-1.5pt}S}_{t}\big|^{\alpha}\Big] = 0.
\end{equation}
\end{theorem}
\begin{proof}
See Appendix \ref{sec:aux3.10}.
\end{proof}

Since the payoff of a typical FX or equity contract grows at most linearly in the exchange rate or the asset price, we only need to know the strong convergence in $L^{1}$ of the discounted process to deduce the convergence of the time-discretization error to zero. The following theorem can be generalized to the $L^{\alpha}$ case relatively easily, for all $\alpha\geq1$, upon noticing that the critical time $T^{*}$ from \eqref{eq3.38.1} -- \eqref{eq3.38.4} is always greater than the one from \eqref{eq3.40.1}~--~\eqref{eq3.40.2}.
\begin{theorem}\label{Thm3.11} If $2k_{f}\theta_{f}>\xi_{f}^{2}$ and $T<T^{*}$, then the discounted process converges strongly in $L^{1}$, i.e.,
\begin{equation}\label{eq3.79.1}
\plim_{\delta t \to 0}\hspace{1pt}\sup_{t \in [0,T]}\E\Big[\big|R_{t}-\hspace{1.5pt}\oversymb{\hspace{-1.5pt}R}_{t}\big|\Big] = 0,
\end{equation}
where $T^{*}$ is given below:
\begin{equation}\label{eq3.79.2}
T^{*} =
\begin{dcases}
\frac{1}{\rho_{sv}\xi-k}\hspace{-8pt}&, \hspace{4pt}\text{if}\hspace{5pt} k<\rho_{sv}\xi \\[2pt]
\omit\hfill$\infty$\hfill\hspace{-8pt}&, \hspace{4pt}\text{if}\hspace{5pt} k\geq\rho_{sv}\xi.
\end{dcases}
\end{equation}
\end{theorem}
\begin{proof}
The convergence in probability of the discounted process is a consequence of Proposition \ref{Prop3.9}, by taking the domestic interest rate to be zero. The rest of the proof follows the argument of Theorem \ref{Thm3.10} closely and makes use of Propositions \ref{Prop3.6} and \ref{Prop3.7}.
\end{proof}

We can extend the convergence analysis from the four-dimensional Heston-CIR model to multi-factor short rates with CIR dynamics and a term structure, in which case Propositions \ref{Prop3.6} to \ref{Prop3.9}, and hence Theorem \ref{Thm3.11}, still hold, albeit with slightly modified proofs.
\begin{table}[hbt]\renewcommand{\arraystretch}{1.10}\addtolength{\tabcolsep}{0pt}\small
\begin{center}
\caption{The calibrated Heston parameters. Column 2: for USD/EUR market data of 2 January 2004 - 27 September 2005 \citep{Jessen:2013}. Column 3: for EUR/USD market data of 22 August 2006 \citep{Elices:2013}. Column 4: for the S\&P 500 index between 2 January 1990 - 30 September 2003, using VIX data \citep{Ait-Sahalia:2007}. Column 5: for the S\&P 500 index between 2 January 1990 - 30 December 2011, using two out-of-the-money options written on the index \citep{Hurn:2014}.}\label{table:1}
\begin{tabular}{ c c c c c }
  \toprule[.1em]
  Parameter \hspace{.1em}& Jessen \& Poulsen & Elices \& Gim\'enez & A\"it-Sahalia \& Kimmel & Hurn et al. \\
  \hline
	\addlinespace[4pt]
  $k$ &	$2.2200$ & $1.1000$ & $5.1300$ & $1.9775$ \\
  $\theta$ & $0.0120$ & $0.0097$ & $0.0436$ & $0.0376$ \\
	$\xi$ & $0.1830$ & $0.1400$ & $0.5200$ & $0.4568$ \\
	$\rho_{sv}$ & $0.0634$ & $0.1400$ & $\hspace{-.75em}-0.7540$ & $\hspace{-.75em}-0.7591$ \\
	\bottomrule[.1em]
\end{tabular}
\end{center}
\end{table}

The condition $k\geq\rho_{sv}\xi$, also known as the \textit{good correlation regime} \citep{Zeliade:2011}, is almost always satisfied in both FX and equity markets. This is because the speed of mean reversion $k$ is usually larger than the volatility of volatility $\xi$, a fact clearly illustrated in Table \ref{table:1}. And even if this was not the case, a negative correlation between the underlying process and the variance, as is typically the case in equity markets (the so-called \textit{leverage effect}), or a small absolute value of this correlation, as is typically the case in FX markets, would ensure the validity of the condition. Furthermore, the Feller condition $2k_{f}\theta_{f}>\xi_{f}^{2}$ in \eqref{eq3.60} for the foreign interest rate is generally satisfied in practice, a fact clearly illustrated in Table \ref{table:2}. We do not require a Feller condition for the stochastic volatility, and this is not always given in practice.
\begin{table}[ht]\renewcommand{\arraystretch}{1.10}\addtolength{\tabcolsep}{0pt}\small
\begin{center}
\caption{The calibrated Cox-Ingersoll-Ross parameters. Column 2: to the 3-month US Treasury bill yield between January 1964 - December 1998 \citep{Driffill:2003}. Column 3: to the US Treasury bill yield between October 1982 - April 2011 \citep{Erismann:2011}. Column 4: to the Euro ATM caps volatility curve on 17 January 2000 \citep{Brigo:2006}. Column 5: to the Euro OverNight Index Average between 1 January 2008 - 6 October 2008 \citep{Laffers:2009}. Column 6: using historical data for Euro between 1 January 2001 - 1 September 2011 \citep{Amin:2012}.}\label{table:2}
\begin{tabular}{ c c c c c c }
  \toprule[.1em]
  Parameter \hspace{.1em}& Driffill et al. & Erismann & Brigo \& Mercurio \hspace{.1em}& Laff\'ers &\hspace{1.4em} Amin \\
  \hline
	\addlinespace[4pt]
  $k_{d,f}$ & $0.0684$ & $0.1104$ & $0.3945$ & $0.2820$ &\hspace{1.4em} $0.1990$ \\
  $\theta_{d,f}$ & $0.0161$ & $0.0509$ & $0.2713$ & $0.0411$ &\hspace{1.4em} $0.0497$ \\
	$\xi_{d,f}$ & $0.0177$ & $0.0498$ & $0.0545$ & $0.0058$ &\hspace{1.4em} $0.0354$ \\
	\bottomrule[.1em]
\end{tabular}
\end{center}
\end{table}

\subsection{Option pricing}\label{subsec:options}

We conclude this section with a brief study on the convergence of mixed Monte Carlo/PDE estimators for computing FX option prices. Define the fair price of an option written on $S$:
\begin{equation}\label{eq3.80}
U = \E\Big[e^{-\int_{0}^{T}{r^{d}_{t} dt}}f(\hspace{.25pt}S\hspace{.75pt})\Big],
\end{equation}
and its approximation under \eqref{eq2.8}:
\begin{equation}\label{eq3.81}
\hspace{1pt}\oversymb{\hspace{-1pt}U\hspace{-.5pt}}\hspace{.5pt} = \E\Big[e^{-\int_{0}^{T}{\oversymb{r}^{d}_{t} dt}}f(\hspace{.25pt}\hspace{1.5pt}\oversymb{\hspace{-1.5pt}S}\hspace{1pt})\Big],
\end{equation}
where the payoff function $f$ may depend on the entire path of the process and all expectations in this section are under the domestic risk-neutral measure $\mathbb{Q}$. The following theorem is concerned with the convergence of the time-discretization error to zero under the four-factor FX model, and can be extended to multi-factor CIR short rates in a straightforward manner. The proof employs Propositions \ref{Prop3.6}, \ref{Prop3.7} and \ref{Prop3.9}, and Theorem \ref{Thm3.11}. However, we omit it here because of its similarity to the proof of Theorem 2.1 in \citet{Cozma:2015}, once the aforementioned results are established.

\begin{theorem}\label{Thm3.13} Suppose that $2k_{f}\theta_{f}>\xi_{f}^{2}$. Then the following two statements hold:
\begin{itemize}
\item[(i)] The approximations to the values of the European put, the up-and-out call and any barrier put option defined in \eqref{eq3.81} converge as $\delta t\to0$.
\item[(ii)] If $T\hspace{-1pt}<T^{*}$, with $T^{*}$ from \eqref{eq3.79.2}, the approximations to the values of the European call, Asian options, the down-and-in/out and the up-and-in call option defined in \eqref{eq3.81} converge as $\delta t\to0$.
\end{itemize}
\end{theorem}

For European contracts, we can evaluate the conditional option price, i.e., the innermost expectation in \eqref{eq2.10}, analytically. Hence, consider $M$ simulations of the discrete paths of the variance and the interest rates and, for $1\leq j\leq M$, let $\omega_{j}$ denote the $j$-th sample. Then
\begin{equation}\label{eq3.82}
\frac{1}{M}\sum_{j=1}^{M}{\hspace{1pt}\E\left[e^{-\int_{0}^{T}{\oversymb{r}^{d}_{t}dt}}f(\hspace{1.5pt}\oversymb{\hspace{-1.5pt}S}_{T})\hspace{1pt}\big|\hspace{2pt}\mathcal{G}_{T}^{f,d,v}, \omega=\omega_{j}\hspace{.5pt}\right]}
= \frac{1}{M}\sum_{j=1}^{M}{\hspace{1pt}\E\left[e^{-\int_{0}^{T}{\oversymb{r}^{d}_{t}\scalebox{0.65}{$(j)$}\hspace{.5pt}dt}}f(\hspace{1.5pt}\oversymb{\hspace{-1.5pt}S}_{T}\scalebox{0.85}{$(j)$})\hspace{1pt}\big|\hspace{2pt}\mathcal{G}_{T}^{f,d,v}\right]}
\end{equation}
is the mixed Monte Carlo/PDE estimator of the European option price at time $t=0$. The global error can be split into two parts:
\begin{align*}
\text{Error} &= \E\Big[e^{-\int_{0}^{T}{r^{d}_{t} dt}}f(S_{T})\Big] - \frac{1}{M}\sum_{j=1}^{M}{\hspace{1pt}\E\left[e^{-\int_{0}^{T}{\oversymb{r}^{d}_{t} dt}}f(\hspace{1.5pt}\oversymb{\hspace{-1.5pt}S}_{T})\hspace{1pt}\big|\hspace{2pt}\mathcal{G}_{T}^{f,d,v}, \omega=\omega_{j}\hspace{.5pt}\right]} \\[1pt]
&= \bigg(\E\Big[e^{-\int_{0}^{T}{r^{d}_{t} dt}}f(S_{T})\Big] - \E\Big[e^{-\int_{0}^{T}{\oversymb{r}^{d}_{t} dt}}f(\hspace{1.5pt}\oversymb{\hspace{-1.5pt}S}_{T})\Big]\bigg) \\[1pt]
&+ \bigg(\E\Big[\E\left[e^{-\int_{0}^{T}{\oversymb{r}^{d}_{t} dt}}f(\hspace{1.5pt}\oversymb{\hspace{-1.5pt}S}_{T})\hspace{1pt}\big|\hspace{2pt}\mathcal{G}_{T}^{f,d,v}\right]\Big] - \frac{1}{M}\sum_{j=1}^{M}{\hspace{1pt}\E\left[e^{-\int_{0}^{T}{\oversymb{r}^{d}_{t} dt}}f(\hspace{1.5pt}\oversymb{\hspace{-1.5pt}S}_{T})\hspace{1pt}\big|\hspace{2pt}\mathcal{G}_{T}^{f,d,v}, \omega=\omega_{j}\hspace{.5pt}\right]}\bigg).
\end{align*}
The first term is the \textit{time-discretization error} and the second term is the \textit{statistical error}. The convergence to zero of the former was derived in Theorem \ref{Thm3.13} for European put and call options, result that can be easily extended to other financial derivatives, such as binary options. The convergence to zero of the latter follows from the Central Limit Theorem \citep[see][]{Glasserman:2003} upon noticing the following upper bound on the variance,
\begin{equation*}
\text{Var}\bigg(\hspace{-1pt}\E\!\left[e^{-\int_{0}^{T}{\oversymb{r}^{d}_{t} dt}}f(\hspace{1.5pt}\oversymb{\hspace{-1.5pt}S}_{T})\hspace{1pt}\big|\hspace{2pt}\mathcal{G}_{T}^{f,d,v}\right]\hspace{-1pt}\bigg)
\leq \E\!\left[\E\!\left[e^{-\int_{0}^{T}{\oversymb{r}^{d}_{t} dt}}f(\hspace{1.5pt}\oversymb{\hspace{-1.5pt}S}_{T})\hspace{1pt}\big|\hspace{2pt}\mathcal{G}_{T}^{f,d,v}\right]^{2}\right]
\leq \E\!\Big[e^{-2\int_{0}^{T}{\oversymb{r}^{d}_{t} dt}}f(\hspace{1.5pt}\oversymb{\hspace{-1.5pt}S}_{T})^{2}\Big].
\end{equation*}
Assuming $f$ has at most polynomial growth, we can employ Proposition \ref{Prop3.4} to deduce the finiteness of the variance under some conditions on the model parameters. In particular, if $f$ is Lipschitz, Proposition \ref{Prop3.7} gives sharper sufficient conditions, for $\alpha=2$.

For path-dependent contracts, closed-form solutions are rarely available and we rely on finite differences instead to compute the conditional option price. Conditioning on the $j$-th realization of the variance and interest rates paths, let $u_{\hspace{.5pt}j}^{f,d,v}\big(t,\hspace{1.5pt}\oversymb{\hspace{-1.5pt}S}_{t}\big)$ and $\bar{u}_{\hspace{.5pt}j}^{f,d,v}\big(t,\hspace{1.5pt}\oversymb{\hspace{-1.5pt}S}_{t};P,L\big)$ be the solutions to the conditional PDE and to the associated finite difference scheme, when a uniform mesh with $P$ time steps and $L$ spatial steps is employed, respectively. It can be shown that the conditional PDE has a unique solution \citep[see Section 7.1.2 in][]{Evans:1998} which is, in fact, the conditional option price \citep[see Theorem 7.3.1 in][]{Shreve:2004}. Then
\begin{equation}\label{eq3.83}
\frac{1}{M}\sum_{j=1}^{M}{\bar{u}_{\hspace{.5pt}j}^{f,d,v}\big(0,S_{0};P,L\big)}
\end{equation}
is the mixed Monte Carlo/PDE estimator of the option price at $t=0$. The global error can be split into three parts:
\begin{align*}
\text{Error} &= \E\Big[e^{-\int_{0}^{T}{r^{d}_{t} dt}}f(\hspace{.25pt}S\hspace{.75pt})\Big] - \frac{1}{M}\sum_{j=1}^{M}{\bar{u}_{\hspace{.5pt}j}^{f,d,v}\big(0,S_{0};P,L\big)} \\[1pt]
&= \bigg(\E\Big[e^{-\int_{0}^{T}{r^{d}_{t} dt}}f(\hspace{.25pt}S\hspace{.75pt})\Big] - \E\Big[e^{-\int_{0}^{T}{\oversymb{r}^{d}_{t} dt}}f(\hspace{.25pt}\hspace{1.5pt}\oversymb{\hspace{-1.5pt}S}\hspace{1pt})\Big]\bigg) \\[1pt]
&+ \bigg(\E\Big[\E\left[e^{-\int_{0}^{T}{\oversymb{r}^{d}_{t} dt}}f(\hspace{.25pt}\hspace{1.5pt}\oversymb{\hspace{-1.5pt}S}\hspace{1pt})\hspace{1pt}\big|\hspace{2pt}\mathcal{G}_{T}^{f,d,v}\right]\Big] - \frac{1}{M}\sum_{j=1}^{M}{\hspace{1pt}\E\left[e^{-\int_{0}^{T}{\oversymb{r}^{d}_{t} dt}}f(\hspace{.25pt}\hspace{1.5pt}\oversymb{\hspace{-1.5pt}S}\hspace{1pt})\hspace{1pt}\big|\hspace{2pt}\mathcal{G}_{T}^{f,d,v}, \omega=\omega_{j}\hspace{.5pt}\right]}\bigg) \\[-2pt]
&+ \frac{1}{M}\sum_{j=1}^{M}{\left(u_{\hspace{.5pt}j}^{f,d,v}\big(0,S_{0}\big)-\bar{u}_{\hspace{.5pt}j}^{f,d,v}\big(0,S_{0};P,L\big)\right)}.
\end{align*}
The first term is the \textit{time-discretization error}, the second term is the \textit{statistical error} and the third term is the \textit{finite difference (FD) discretization error}. The convergence to zero of the first term was derived in Theorem \ref{Thm3.13} for Asian and barrier options, and the convergence to zero of the second term is a consequence of the Central Limit Theorem (see the above discussion). However, the convergence of the third term, i.e., of the finite difference scheme, depends on the contract and the particular scheme employed.

\section{Variance reduction analysis}\label{sec:variance}

In this section, we carry out a variance reduction analysis for European option valuation with the mixed Monte Carlo/PDE method under the four-factor FX model. However, the theory extends naturally to general interest rate dynamics. We use standard Monte Carlo with the log-Euler discretization as the reference method and define $\hat{X}$ and $\hat{S}$ to be the time-continuous approximations of $x$, defined in \eqref{eq3.42}, and $S$, defined in \eqref{eq2.1}, respectively. Then
\begin{equation}\label{eq4.1}
\hat{X}_{t} = \hat{X}_{t_{n}} + \Big(\oversymb{r}^{d}_{t_{n}} - \oversymb{r}^{f}_{t_{n}} - \frac{1}{2}\hspace{1pt}\oversymb{V}_{\hspace{-2.5pt}t_{n}}\Big)\big(t-t_{n}\big) + \sqrt{\oversymb{V}_{\hspace{-2.5pt}t_{n}}}\,\Delta W^{s}_{t},
\end{equation}
where $\Delta W^{s}_{t} = W^{s}_{t}-W^{s}_{t_{n}}$ whenever $t\in[t_{n},t_{n+1})$. Integrating \eqref{eq4.1} leads to
\begin{equation}\label{eq4.2}
\hat{S}_{t} = S_{0}\exp\bigg\{\int_{0}^{t}{\Big(\oversymb{r}^{d}_{u}-\oversymb{r}^{f}_{u}-\frac{1}{2}\hspace{1pt}\oversymb{V}_{\hspace{-2.5pt}u}\Big)du} + \int_{0}^{t}{\sqrt{\oversymb{V}_{\hspace{-2.5pt}u}}\,dW^{s}_{u}}\bigg\}.
\end{equation}
We prefer the log-Euler scheme to the standard Euler scheme because it preserves positivity. Moreover, if the processes $v$, $r^{d}$, $r^{f}$ are constant, then the first scheme is exact. Recall that
\begin{equation*}
\hspace{1.5pt}\oversymb{\hspace{-1.5pt}S}_{t} = S_{0}\exp\bigg\{\int_{0}^{t}{\Big(\oversymb{r}^{d}_{u}-\oversymb{r}^{f}_{u}-\frac{1}{2}\hspace{1pt}\oversymb{V}_{\hspace{-2.5pt}u}\Big)du} + a_{11}\int_{0}^{t}{\sqrt{\oversymb{V}_{\hspace{-2.5pt}u}}\,dW^{1}_{u}} + \sum_{j=2}^{4}{a_{1\hspace{-0.5pt}j}\int_{0}^{t}{\sqrt{\oversymb{V}_{\hspace{-2.5pt}u}}\,\frac{\delta W^{j}_{u}}{\delta t}\,du}}\bigg\}.
\end{equation*}
Since $\oversymb{V}$ is piecewise constant, we deduce that $\hat{S}_{t_{n}}=\hspace{1.5pt}\oversymb{\hspace{-1.5pt}S}_{t_{n}}$, $\forall\hspace{1pt} 0\leq n\leq N$. The quantity that we want to estimate is the fair price of a European option with payoff function $f$, i.e.,
\begin{equation}\label{eq4.3}
\Theta = \E\Big[e^{-\int_{0}^{T}{r^{d}_{t} dt}}f(S_{T})\Big].
\end{equation}
Then the corresponding standard and mixed Monte Carlo estimators are
\begin{align}
\Theta_{\text{stdMC}} &= \frac{1}{M}\sum_{j=1}^{M}{\hspace{1pt}e^{-\int_{0}^{T}{\oversymb{r}^{d}_{t}\scalebox{0.65}{$(j)$}\hspace{.5pt}dt}}f(\hat{S}_{T}\scalebox{0.85}{$(j)$})}, \label{eq4.4} \\[1pt]
\Theta_{\text{mixMC}} &= \frac{1}{M}\sum_{j=1}^{M}{\hspace{1pt}\E\left[e^{-\int_{0}^{T}{\oversymb{r}^{d}_{t}\scalebox{0.65}{$(j)$}\hspace{.5pt}dt}}f(\hspace{1.5pt}\oversymb{\hspace{-1.5pt}S}_{T}\scalebox{0.85}{$(j)$})\hspace{1pt}\big|\hspace{2pt}\mathcal{G}_{T}^{f,d,v}\right]}. \label{eq4.5}
\end{align}
Define
\begin{equation}\label{eq4.6}
\text{Var}_\text{stdMC} = \text{Var}\big(\Theta_{\text{stdMC}}\big) = \frac{1}{M}\hspace{1pt}\text{Var}\left(e^{-\int_{0}^{T}{\oversymb{r}^{d}_{t}dt}}f(\hat{S}_{T})\right)
\end{equation}
and
\begin{equation}\label{eq4.7}
\text{Var}_\text{mixMC} = \text{Var}\big(\Theta_{\text{mixMC}}\big) = \frac{1}{M}\hspace{1pt}\text{Var}\left(\E\left[e^{-\int_{0}^{T}{\oversymb{r}^{d}_{t}dt}}f(\hspace{1.5pt}\oversymb{\hspace{-1.5pt}S}_{T})\hspace{1pt}\big|\hspace{2pt}\mathcal{G}_{T}^{f,d,v}\right]\right).
\end{equation}
Let the variance reduction factor \citep[as in][]{Dang:2015} and the standard deviation ratio be
\begin{equation}\label{eq4.8}
\Gamma_{\text{var}} = \frac{\text{Var}_\text{stdMC}}{\text{Var}_\text{mixMC}} \hspace{4pt}\text{ and }\hspace{4pt} \Gamma_{\text{dev}} = \sqrt{\Gamma_{\text{var}}}\hspace{1pt}.
\end{equation}
For convenience, we also define the discount factors, $D=e^{-\int_{0}^{T}{r^{d}_{t}dt}}$ and $\hspace{1.5pt}\oversymb{\hspace{-1.5pt}D\hspace{-.5pt}}\hspace{.5pt}=e^{-\int_{0}^{T}{\oversymb{r}^{d}_{t}dt}}$.

\begin{remark}\label{Rem4.1}
From the ``tower property'' of conditional expectations, since $\hat{S}_{T}=\hspace{1.5pt}\oversymb{\hspace{-1.5pt}S}_{T}$,
\begin{equation*}
\E\big[\Theta_{\text{stdMC}}\big] = \E\big[\hspace{1.5pt}\oversymb{\hspace{-1.5pt}D\hspace{-.5pt}}\hspace{.5pt}f(\hat{S}_{T})\big]
= \E\big[\hspace{1.5pt}\oversymb{\hspace{-1.5pt}D\hspace{-.5pt}}\hspace{.5pt}f(\hspace{1.5pt}\oversymb{\hspace{-1.5pt}S}_{T})\big]
= \E\big[\E\big[\hspace{1.5pt}\oversymb{\hspace{-1.5pt}D\hspace{-.5pt}}\hspace{.5pt}f(\hspace{1.5pt}\oversymb{\hspace{-1.5pt}S}_{T})\hspace{1pt}|\hspace{2pt}\mathcal{G}_{T}^{f,d,v}\big]\big]
= \E\big[\Theta_{\text{mixMC}}\big].
\end{equation*}
Therefore, the standard and the mixed Monte Carlo estimators have the same discretization bias, i.e.,
\begin{equation}\label{eq4.9}
\text{Bias}\big(\Theta_{\text{stdMC}}\big) = \text{Bias}\big(\Theta_{\text{mixMC}}\big).
\end{equation}
\end{remark}

\begin{remark}\label{Rem4.2}
From the law of total variance we know that
\begin{equation} \label{eq4.10}
\text{Var}\big(\hspace{1.5pt}\oversymb{\hspace{-1.5pt}D\hspace{-.5pt}}\hspace{.5pt}f(\hspace{1.5pt}\oversymb{\hspace{-1.5pt}S}_{T})\big)
= \text{Var}\big(\E\big[\hspace{1.5pt}\oversymb{\hspace{-1.5pt}D\hspace{-.5pt}}\hspace{.5pt}f(\hspace{1.5pt}\oversymb{\hspace{-1.5pt}S}_{T})\hspace{1pt}|\hspace{2pt}\mathcal{G}_{T}^{f,d,v}\big]\big) + \E\big[\text{Var}\big(\hspace{1.5pt}\oversymb{\hspace{-1.5pt}D\hspace{-.5pt}}\hspace{.5pt}f(\hspace{1.5pt}\oversymb{\hspace{-1.5pt}S}_{T})\hspace{1pt}|\hspace{2pt}\mathcal{G}_{T}^{f,d,v}\big)\big].
\end{equation}
However, since $\hat{S}_{T}=\hspace{1.5pt}\oversymb{\hspace{-1.5pt}S}_{T}$ and the variance is non-negative, we deduce from \eqref{eq4.6} -- \eqref{eq4.7} that the variance of the standard Monte Carlo estimator is greater than or equal to the variance of the mixed estimator, i.e.,
\begin{equation}\label{eq4.11}
\text{Var}_\text{stdMC} \geq \text{Var}_\text{mixMC}\hspace{.5pt}.
\end{equation}
Assuming a non-trivial payoff function $f$, equality occurs in \eqref{eq4.11} if and only if the second expectation on the right-hand side of \eqref{eq4.10} is zero, i.e., if and only if $\hspace{1.5pt}\oversymb{\hspace{-1.5pt}S}_{T}$ is $\mathcal{G}_{T}^{f,d,v}$-measurable. Since $\Sigma=AA^{T}$ and after some straightforward calculations, we find an equivalent condition:
\begin{align*}
a_{11}=0 \hspace{2pt}\Leftrightarrow\hspace{1pt} 1 &- \rho_{v\hspace{.2pt}d}^{2} - \rho_{v\hspace{-.7pt}f}^{2} - \rho_{d\hspace{.01pt}f}^{2} + 2\rho_{v\hspace{.2pt}d}\rho_{v\hspace{-.7pt}f}\rho_{d\hspace{.01pt}f}
= \rho_{sv}^{2}\big(1-\rho_{d\hspace{.01pt}f}^{2}\big) + \rho_{sd}^{2}\big(1-\rho_{v\hspace{-.7pt}f}^{2}\big) + \rho_{s\hspace{-.7pt}f}^{2}\big(1-\rho_{v\hspace{.2pt}d}^{2}\big) \\[4pt]
&+ 2\rho_{sv}\rho_{sd}\big(\rho_{v\hspace{-.7pt}f}\rho_{d\hspace{.01pt}f}-\rho_{v\hspace{.2pt}d}\big) + 2\rho_{sv}\rho_{s\hspace{-.7pt}f}\big(\rho_{v\hspace{.2pt}d}\rho_{d\hspace{.01pt}f}-\rho_{v\hspace{-.7pt}f}\big) + 2\rho_{sd}\rho_{s\hspace{-.7pt}f}\big(\rho_{v\hspace{.2pt}d}\rho_{v\hspace{-.7pt}f}-\rho_{d\hspace{.01pt}f}\big).
\end{align*}
In particular, if the variance and the two interest rates are pairwise independent, then
\begin{equation*}
a_{11}=0 \hspace{2pt}\Leftrightarrow\hspace{1pt} \rho_{sv}^{2} + \rho_{sd}^{2} + \rho_{s\hspace{-.7pt}f}^{2} = 1.
\end{equation*}
Therefore, apart from this case, combining conditioning with Monte Carlo always reduces the variance of estimates. This is, however, to be expected since we eliminate the additional noise that comes from simulating the Brownian motion $W^{1}$.
\end{remark}

\begin{remark}\label{Rem4.3}
Note that for any $2\leq j\leq4$, using the It\^o isometry, we have
\begin{equation}\label{eq4.11.1}
\text{Var}\bigg(\int_{0}^{T}{\hspace{-1pt}\sqrt{\oversymb{V}_{\hspace{-2.5pt}t}}\hspace{1.5pt}dW^{j}_{t}}\bigg) = \E\bigg[\int_{0}^{T}{\oversymb{V}_{\hspace{-2.5pt}t}\hspace{1pt}dt}\bigg].
\end{equation}
Moreover, using Cauchy's and H\"older's inequalities, Fubini's theorem, and Remark 3.2 and Propositions 3.4 and 3.5 in \citet{Cozma:2015}, one can easily prove that
\begin{equation}\label{eq4.11.2}
\lim_{\delta t\to0}\hspace{1.5pt}\frac{\E\left[\int_{0}^{T}{\oversymb{V}_{\hspace{-2.5pt}t}\hspace{1pt}dt}\right]}{\text{Var}\left(\int_{0}^{T}{\oversymb{V}_{\hspace{-2.5pt}t}\hspace{1pt}dt}\right)} \hspace{1pt}=\hspace{1pt}
\frac{\E\left[\int_{0}^{T}{v_{t}\hspace{.5pt}dt}\right]}{\text{Var}\left(\int_{0}^{T}{v_{t}\hspace{.5pt}dt}\right)}\hspace{1pt}.
\end{equation}
From \citet{Dufresne:2001}, we can compute the first two moments of the integrated square root process explicitly. Hence, we find
\begin{equation}\label{eq4.11.3}
\E\bigg[\int_{0}^{T}{v_{t}\hspace{.5pt}dt}\bigg] = \theta T + \frac{v_{0}}{k} - \frac{\theta}{k} + e^{-kT}\left(\frac{\theta}{k} - \frac{v_{0}}{k}\right)
\end{equation}
and
\begin{align}\label{eq4.11.4}
\text{Var}\bigg(\int_{0}^{T}{v_{t}\hspace{.5pt}dt}\bigg) &= \frac{\xi^{2}}{k^{2}}\left[\theta T + \frac{v_{0}}{k} - \frac{5\theta}{2k} + 2e^{-kT}\left(\frac{\theta}{k} + \theta T - v_{0}T\right) + e^{-2kT}\left(\frac{\theta}{2k} - \frac{v_{0}}{k}\right)\right] \nonumber\\[3pt]
&< \frac{\xi^{2}}{k^{2}}\big(1+e^{-kT}\big)\E\bigg[\int_{0}^{T}{v_{t}\hspace{.5pt}dt}\bigg] + \frac{\xi^{2}}{k^{2}}\left[\frac{\theta}{2k}\hspace{1pt}e^{-kT}\left(4 + 2kT - e^{-kT} - 3e^{kT}\right)\right] \nonumber\\[3pt]
&< \frac{\xi^{2}}{k^{2}}\big(1+e^{-kT}\big)\E\bigg[\int_{0}^{T}{v_{t}\hspace{.5pt}dt}\bigg].
\end{align}
Combining \eqref{eq4.11.1} -- \eqref{eq4.11.4}, we deduce that for sufficiently small values of $\delta t$,
\begin{equation}\label{eq4.11.5}
\frac{\text{Var}\left(\int_{0}^{T}{\hspace{-1pt}\sqrt{\oversymb{V}_{\hspace{-2.5pt}t}}\hspace{1.5pt}dW^{j}_{t}}\right)}{\text{Var}\left(\int_{0}^{T}{\oversymb{V}_{\hspace{-2.5pt}t}\hspace{1pt}dt}\right)} > \frac{k^{2}}{\xi^{2}\left(1+e^{-kT}\right)} > \frac{k^{2}}{2\xi^{2}}\hspace{1pt}.
\end{equation}
The data in Tables \ref{table:1} and \ref{table:2} suggest that the interest rates have little impact on the variance of the mixed Monte Carlo estimator, and also that $k\gg\xi$ in both FX and equity markets. Hence, the stochastic integral on the left-hand side of \eqref{eq4.11.5} -- part of the dividend yield $q$ defined in \eqref{eq2.9} -- contributes to the overall variance of the mixed estimator considerably more than the squared volatility $\sigma^{2}$ defined in \eqref{eq2.9}. Therefore, we expect the minimum variance of the mixed estimator to be attained when all but the first term in $q$ vanish, i.e., when $a_{11}=1$ and $a_{12}=a_{13}=a_{14}=0$.
\end{remark}

In practice, the volatility of interest rates in the CIR model is very small, i.e., $\xi_{d,f}\ll1$, a fact which can be observed in Table \ref{table:2}. Moreover, the volatility of volatility in the Heston model -- calibrated to FX market data -- is also small, i.e., $\xi\ll1$, and is significantly smaller than the rate of mean reversion, i.e., $\xi\ll k$, a fact clearly illustrated in Table \ref{table:1}. Hence, the drift term in the square root model for the variance dominates the diffusion term. Therefore, we assume in the subsequent analysis ``almost deterministic'' dynamics of the variance and interest rates. Let $\hspace{.5pt}\oversymb{\hspace{-.5pt}\gamma}$, $\hspace{.5pt}\oversymb{\hspace{-.5pt}\gamma}^{d}$ and $\hspace{.5pt}\oversymb{\hspace{-.5pt}\gamma}^{f}$ be the FTE discretizations of $v$, $r^{d}$ and $r^{f}$ corresponding to $\xi=\xi_{d}=\xi_{f}=0$, i.e., when the volatility of volatility parameters are equal to zero. Then
\begin{equation*}
\int_{0}^{T}{\hspace{-1.5pt}\oversymb{r}^{d}_{t}dt} \approx \int_{0}^{T}{\hspace{-1.5pt}\oversymb{\hspace{-.5pt}\gamma}^{d}_{t}dt}\hspace{.5pt},\hspace{1pt}
\int_{0}^{T}{\hspace{-1.5pt}\oversymb{r}^{f}_{t}dt} \approx \int_{0}^{T}{\hspace{-1.5pt}\oversymb{\hspace{-.5pt}\gamma}^{f}_{t}dt}\hspace{.5pt},\hspace{1pt}
\int_{0}^{T}{\hspace{-1.5pt}\oversymb{V}_{\hspace{-2.5pt}t}\hspace{.5pt}dt} \approx \int_{0}^{T}{\hspace{-1.5pt}\oversymb{\hspace{-.5pt}\gamma}_{t}\hspace{.5pt}dt}\hspace{.5pt},\hspace{1pt}
\int_{0}^{T}{\hspace{-2.5pt}\sqrt{\oversymb{V}_{\hspace{-2.5pt}t}}\hspace{1.5pt}dW^{s}_{t}} \approx \int_{0}^{T}{\hspace{-2.5pt}\sqrt{\hspace{.5pt}\oversymb{\hspace{-.5pt}\gamma}_{t}}\hspace{1.5pt}dW^{s}_{t}}.
\end{equation*}

\begin{remark}\label{Rem4.4}
Suppose that $W^{s}$ is independent of the Brownian motions $W^{v}$, $W^{d}$ and $W^{f}$, i.e., that $a_{11}=1$ and $a_{12}=a_{13}=a_{14}=0$. Employing \eqref{eq4.7} as well as the above assumption on the dynamics of the variance and the domestic and foreign interest rates,
\begin{equation*}
\text{Var}_\text{mixMC} = \frac{1}{M}\hspace{1pt}\text{Var}\left(e^{-\int_{0}^{T}{\oversymb{r}^{d}_{t}dt}}\E\big[f(\hspace{1.5pt}\oversymb{\hspace{-1.5pt}S}_{T})\hspace{1pt}|\hspace{2pt}\mathcal{G}_{T}^{f,d,v}\big]\right)
\approx \frac{1}{M}\hspace{1pt}e^{-2\int_{0}^{T}{\hspace{.5pt}\oversymb{\hspace{-.5pt}\gamma}^{d}_{t}dt}}\hspace{2pt}\text{Var}\big(\E\big[f(\hspace{1.5pt}\oversymb{\hspace{-1.5pt}S}_{T})\big]\big) = 0.
\end{equation*}
On the other hand, since $\hat{S}_{T}=\hspace{1.5pt}\oversymb{\hspace{-1.5pt}S}_{T}$, we know from \eqref{eq4.6} that
\begin{equation*}
\text{Var}_\text{stdMC} = \frac{1}{M}\hspace{1pt}\text{Var}\left(e^{-\int_{0}^{T}{\oversymb{r}^{d}_{t}dt}}f(\hspace{1.5pt}\oversymb{\hspace{-1.5pt}S}_{T})\right)
\approx \frac{1}{M}\hspace{1pt}e^{-2\int_{0}^{T}{\hspace{.5pt}\oversymb{\hspace{-.5pt}\gamma}^{d}_{t}dt}}\hspace{2pt}\text{Var}\big(f(\hspace{1.5pt}\oversymb{\hspace{-1.5pt}S}_{T})\big).
\end{equation*}
Assuming a non-trivial payoff $f$ as well as a non-zero squared volatility $v$, since the variance of the mixed estimator is close to zero, this results in a substantial variance reduction.
\end{remark}

\begin{remark}\label{Rem4.6}
Let $f$ be the European call option payoff and suppose that the dynamics of the variance and interest rates are ``almost deterministic'' (i.e., their own volatility parameters are small compared to their mean-reversion speed). Then we can approximate the discounted payoff as follows:
\begin{align}\label{eq4.14}
P &= e^{-\int_{0}^{T}{\oversymb{r}^{d}_{t}dt}}\bigg(S_{0}\exp\bigg\{\int_{0}^{T}{\Big(\oversymb{r}^{d}_{t}-\oversymb{r}^{f}_{t}-\frac{1}{2}\hspace{1pt}\oversymb{V}_{\hspace{-2.5pt}t}\Big)dt} + \int_{0}^{T}{\sqrt{\oversymb{V}_{\hspace{-2.5pt}t}}\,dW^{s}_{t}}\bigg\}-K\bigg)^{+} \nonumber\\[2pt]
&\approx e^{-\int_{0}^{T}{\hspace{.5pt}\oversymb{\hspace{-.5pt}\gamma}^{d}_{t}dt}}\bigg(S_{0}\exp\bigg\{\int_{0}^{T}{\Big(\hspace{.5pt}\oversymb{\hspace{-.5pt}\gamma}^{d}_{t}-\hspace{.5pt}\oversymb{\hspace{-.5pt}\gamma}^{f}_{t}-\frac{1}{2}\hspace{1.5pt}\oversymb{\hspace{-.5pt}\gamma}_{t}\Big)dt} + \int_{0}^{T}{\hspace{-1pt}\sqrt{\hspace{.5pt}\oversymb{\hspace{-.5pt}\gamma}_{t}}\,dW^{s}_{t}}\bigg\}-K\bigg)^{+}.
\end{align}
For now, assume that $a_{11}>0$. For convenience, define the quantities
\begin{equation}\label{eq4.21}
\varrho = \sqrt{1-a_{11}^{2}}\hspace{1pt},\hspace{4pt}
\tilde{\sigma} = \sqrt{\int_{0}^{T}{\hspace{.5pt}\oversymb{\hspace{-.5pt}\gamma}_{t}\hspace{.5pt}dt}}\hspace{1pt},\hspace{4pt}
\tilde{D} = e^{-\int_{0}^{T}{\hspace{.5pt}\oversymb{\hspace{-.5pt}\gamma}^{d}_{t}dt}}\hspace{1pt},\hspace{4pt}
F = S_{0}\exp\bigg\{\int_{0}^{T}{\big(\hspace{.5pt}\oversymb{\hspace{-.5pt}\gamma}^{d}_{t}-\hspace{.5pt}\oversymb{\hspace{-.5pt}\gamma}^{f}_{t}\big)dt}\bigg\}\hspace{.5pt},
\end{equation}
as well as
\begin{equation}\label{eq4.22}
a = \frac{\varrho}{\sqrt{1-\varrho^{2}}} \hspace{4pt}\text{ and }\hspace{4pt}
b = \frac{\log(F/K) + (0.5+\varrho^{2})\tilde{\sigma}^{2}}{\sqrt{1-\varrho^{2}}\hspace{1pt}\tilde{\sigma}}\hspace{1pt}.
\end{equation}
We use \eqref{eq4.14}, the conditional option price formula in \eqref{eq2.11} and differentiate the variance of the mixed Monte Carlo estimator with respect to the parameter $\varrho$ defined in \eqref{eq4.21} to find, using lengthy integration by parts,
\begin{equation}\label{eq4.23}
\frac{\partial}{\partial\varrho}\hspace{1pt}\text{Var}_{\text{mixMC}} \approx \frac{2}{M}\hspace{1pt}\tilde{D}^{2}F^{2}\varrho\tilde{\sigma}^{2}e^{\varrho^{2}\tilde{\sigma}^{2}}\E\big[\Phi(aZ+b)^{2}\big],
\end{equation}
where $\Phi$ is the standard normal CDF and $Z\sim\mathcal{N}(0,1)$. Note that if $a_{11}<1$, i.e., if $\varrho>0$, the right-hand side of \eqref{eq4.23} is strictly positive. This implies that the variance of the mixed estimator decreases as $a_{11}$ increases, and attains its minimum when $a_{11}=1$. In fact, we know from Remark \ref{Rem4.4} that
\begin{equation}\label{eq4.24}
\text{Var}_{\text{mixMC}}(a_{11}=1) \approx 0.
\end{equation}
Moreover,
\begin{equation}\label{eq4.25}
\E\big[\Phi(aZ+b)^{2}\big] = \Phi\bigg(\frac{b}{\sqrt{1+a^{2}}}\bigg) - 2\mathrm{T}\bigg(\frac{b}{\sqrt{1+a^{2}}}\hspace{1pt},\hspace{1.5pt}\frac{1}{\sqrt{1+2a^{2}}}\bigg),
\end{equation}
where Owen's $\mathrm{T}$ function \citep{Owen:1980} is
\begin{equation}\label{eq4.26}
\mathrm{T}(\beta,\vartheta) = \phi(\beta)\int_{0}^{\vartheta}{\frac{\phi(\beta x)}{1+x^{2}}\hspace{1pt}dx}\hspace{.5pt}, \hspace{2pt}\text{ with }\hspace{2pt}
\mathrm{T}(\beta,1) = \frac{1}{2}\hspace{1pt}\Phi(\beta) - \frac{1}{2}\hspace{1pt}\Phi(\beta)^{2},
\end{equation}
and $\phi$ is the standard normal PDF. Note that
\begin{equation}\label{eq4.26.1}
\frac{b}{\sqrt{1+a^{2}}} = \underbrace{\frac{1}{\tilde{\sigma}}\log\frac{F}{K}}_{\equiv\beta_{1}} + \underbrace{\bigg(\frac{1}{2}+\varrho^{2}\bigg)\tilde{\sigma}}_{\equiv\beta_{2}(\varrho)} \hspace{4pt}\text{ and }\hspace{5pt}
\frac{1}{\sqrt{1+2a^{2}}} = \underbrace{\sqrt{\frac{1-\varrho^{2}}{1+\varrho^{2}}}}_{\equiv\vartheta(\varrho)}\hspace{1pt}.
\end{equation}
According to the data in Table \ref{table:1}, $\theta\ll1$ and so $\hspace{.5pt}\oversymb{\hspace{-.5pt}\gamma}\ll1$. Hence, for low maturities $T$ we have $\tilde{\sigma}\ll1$. From Remark \ref{Rem4.2}, using \eqref{eq4.24}, we know that the variance of the standard Monte Carlo estimator is obtained by integrating \eqref{eq4.23} over $[0,\hspace{-.5pt}1]$. Therefore,
\begin{equation}\label{eq4.26.2}
\Gamma_{\text{var}} \hspace{1pt}\approx\hspace{1pt} \frac{\int_{0}^{1}{\nu\big[\Phi\big(\beta_{1}+\beta_{2}(\nu)\big) - 2\mathrm{T}\big(\beta_{1}+\beta_{2}(\nu),\vartheta(\nu)\big)\big]d\nu}}{\int_{0}^{\varrho}{\nu\big[\Phi\big(\beta_{1}+\beta_{2}(\nu)\big) - 2\mathrm{T}\big(\beta_{1}+\beta_{2}(\nu),\vartheta(\nu)\big)\big]d\nu}}\hspace{1pt}.
\end{equation}
From \eqref{eq4.26} and \eqref{eq4.26.1}, we deduce that $\forall\hspace{1pt}\nu\in[0,\hspace{-.5pt}1]$,
\begin{equation}\label{eq4.26.3}
\Phi\big(\beta_{1}+0.5\hspace{.5pt}\tilde{\sigma}\big)^{2} \leq \Phi\big(\beta_{1}+\beta_{2}(\nu)\big) - 2\mathrm{T}\big(\beta_{1}+\beta_{2}(\nu),\vartheta(\nu)\big) \leq \Phi\big(\beta_{1}+1.5\hspace{.5pt}\tilde{\sigma}\big).
\end{equation}
Therefore,
\begin{equation}\label{eq4.26.4}
\frac{\Phi\big(\beta_{1}+0.5\hspace{.5pt}\tilde{\sigma}\big)^{2}}{\Phi\big(\beta_{1}+1.5\hspace{.5pt}\tilde{\sigma}\big)} \hspace{1pt}\cdot\hspace{1pt}\frac{1}{\varrho^{2}}
\hspace{1pt}\leq\hspace{1pt} \Gamma_{\text{var}} \hspace{1pt}\leq\hspace{1pt}
\frac{\Phi\big(\beta_{1}+1.5\hspace{.5pt}\tilde{\sigma}\big)}{\Phi\big(\beta_{1}+0.5\hspace{.5pt}\tilde{\sigma}\big)^{2}} \hspace{1pt}\cdot\hspace{1pt}\frac{1}{\varrho^{2}}\hspace{1pt}.
\end{equation}
The inequalities in \eqref{eq4.26.4} are approximate in the sense that the variance reduction factor is bounded from above and below by approximated quantities. Assuming that $\tilde{\sigma}\ll1$, \eqref{eq4.26.4} becomes
\begin{equation}\label{eq4.27}
\frac{\Phi\big(\beta_{1}\big)}{\varrho^{2}} \hspace{1pt}\leq\hspace{1pt} \Gamma_{\text{var}} \hspace{1pt}\leq\hspace{1pt} \frac{1}{\Phi\big(\beta_{1}\big)\varrho^{2}}\hspace{1pt}.
\end{equation}
Moreover, if we also assume that the option is deep in-the-money, then $\beta_{1}$ is relatively large (e.g., $\beta_{1}\geq1$) and so $\sqrt{\Phi(\beta_{1})}\approx1$. Hence, we conclude that
\begin{equation}\label{eq4.28}
\Gamma_{\text{dev}} \approx \big(1-a_{11}^{2}\big)^{-\frac{1}{2}},\hspace{3pt} \forall\hspace{.5pt} a_{11}>0.
\end{equation}
However, we know from Remark \ref{Rem4.2} that the standard deviation ratio is one when $a_{11}=0$. Hence, \eqref{eq4.26.4} -- \eqref{eq4.28} hold for all values of $a_{11}$. Interestingly, the deep in-the-money case is also the one where the payoff can be treated as smooth, and this also explains Figure \ref{fig:4}. Note that Remark \ref{Rem4.6} is consistent with Remarks \ref{Rem4.2} -- \ref{Rem4.4}.
\end{remark}

\section{Numerical results}\label{sec:numerics}

In this section, we test the efficiency of the mixed Monte Carlo/PDE method by comparison with alternative numerical schemes for two derivatives: a European call and an up-and-out put option. For the former, we use an analytical formula \eqref{eq2.11} to compute the conditional option price and benchmark against standard Monte Carlo with a log-Euler discretization and, under a simple correlation structure, against the semi-analytical formula of \citet{Ahlip:2013}. For the latter, we use either finite differences or the perturbative formula of \citet{Fatone:2008} for the conditional option price, and Monte Carlo -- with or without Brownian bridge -- as the reference method.

First, we demonstrate the convergence of the mixed Monte Carlo/PDE method and find empirical convergence rates. Then, we investigate the sensitivity of the variance reduction factor \eqref{eq4.8} to changes in the model parameters and link the numerical results to the analysis in Section \ref{sec:variance}. Our platform for the numerical implementation was MATLAB 2012a. The machine configuration on which all numerical tests were conducted is: Intel(R) Core(TM) i3 CPU, M370, 2.40 GHz, Memory (RAM): 8.00 GB, 64-bit Operating System running Windows 7 Professional.

Throughout this section, we assign the following values to the underlying model parameters as a base case, and vary a selection individually or jointly: $v_{0}=0.0275$, $r_{0}^{d}=0.0524$, $r_{0}^{f}=0.0291$, $k=1.70$, $k_{d}=0.20$, $k_{f}=0.32$, $\theta=0.0232$, $\theta_{d}=0.0475$, $\theta_{f}=0.0248$, $\xi=0.1500$, $\xi_{d}=0.0352$, $\xi_{f}=0.0317$, $\rho_{sv}=-0.10$, $\rho_{sd}=-0.15$, $\rho_{s\hspace{-.7pt}f}=-0.15$, $\rho_{v\hspace{.2pt}d}=0.12$, $\rho_{v\hspace{-.7pt}f}=0.05$ and $\rho_{d\hspace{.01pt}f}=0.25$. These values are consistent with empirical observations in FX markets and are close to the calibrated values in Tables \ref{table:1} and \ref{table:2}. Also, the values of the correlations $\rho_{sd}$, $\rho_{s\hspace{-.7pt}f}$ and $\rho_{d\hspace{.01pt}f}$ are borrowed from \citet{Piterbarg:2006}.

\subsection{European call option}\label{subsec:european}

Let the spot, the strike and the maturity be: $S_{0}=105$, $K=100$ and $T=1.5$. On a side note, FX option quotes are in terms of volatilities for a fixed delta and a fixed time to expiry, and not in terms of strikes. However, the strike price corresponding to a quoted volatility can easily be recovered using a conversion formula \citep{Ahlip:2013}. Recall that $\Theta$ from \eqref{eq4.3} is the true option price, whereas $\Theta_{\text{stdMC}}$ from \eqref{eq4.4} and $\Theta_{\text{mixMC}}$ from \eqref{eq4.5} are the standard and the mixed Monte Carlo estimators, respectively. In Table \ref{table3}, we report the common discretization bias, i.e., the time-discretization error, which is the same for both estimators by Remark \ref{Rem4.1}, and the two standard errors, i.e., the two standard deviations of the sample means, which we denote by StDev and estimate using $10\hspace{.5pt}000$ samples. However, according to \citet{Ahlip:2013}, a closed-form solution for the European option price is not available under a full correlation structure. Hence, we need to find an accurate reference estimate $\Theta^{*}$ in order to evaluate the bias. We employ the mixed algorithm with $M=2\scalebox{0.85}{$\times$}10^{9}$ simulations and $N=200$ time steps to find: $\Theta^{*}=12.11968$. Then, the bias for a specific number of time steps is estimated using the reference estimate $\Theta^{*}$, whose accuracy is discussed below, and a sufficiently large number of simulations, i.e., $M=2\scalebox{0.85}{$\times$}10^{9}$.

The time needed to obtain a call price estimate with $64\hspace{.5pt}000$ simulations and $8$ time steps is $0.25$ seconds for the standard Monte Carlo method and $0.22$ seconds for the mixed Monte Carlo method. Hence, the computational cost is $12\%$ lower with the latter. This, however, is to be expected since we only simulate three of the four underlying processes.
\begin{table}[htb]\renewcommand{\arraystretch}{1.10}\addtolength{\tabcolsep}{0pt}\small
\begin{center}
\caption{Simulation results for a European call option.}\label{table3}
\begin{tabular}{ c c c c c }
  \toprule[.1em]
  Time steps & Discretization bias & Simulations & StDev (stdMC) & StDev (mixMC) \\
  \hline
	\addlinespace[4pt]
  $1$ & $0.37231$ & $\hspace{10.2pt}1\hspace{.3pt}000$ & $0.50348$ & $0.09342$ \\
  $2$ & $0.08039$ & $\hspace{10pt}4\hspace{.5pt}000$ & $0.24126$ & $0.04194$ \\
	$4$ & $0.01775$ & $\hspace{5pt}16\hspace{.5pt}000$ & $0.12042$ & $0.02023$ \\
	$8$ & $0.00444$ & $\hspace{5pt}64\hspace{.5pt}000$ & $0.06071$ & $0.00994$ \\
	$\hspace{-5pt}16$ & $0.00160$ & $\hspace{0pt}256\hspace{.5pt}000$ & $0.02932$ & $0.00487$ \\
	$\hspace{-5pt}32$ & $0.00073$ & $\hspace{-5.2pt}1\hspace{.3pt}024\hspace{.5pt}000$ & $0.01507$ & $0.00262$ \\
	\bottomrule[.1em]
\end{tabular}
\end{center}
\end{table}

The data in Table \ref{table3} indicate that the bias is small even when only a few time steps are used. This phenomenon can be explained by the fact that the Feller condition is satisfied, i.e., $2k\theta>\xi^{2}$, $2k_{d}\theta_{d}>\xi_{d}^{2}$ and $2k_{f}\theta_{f}>\xi_{f}^{2}$, and hence the discretizations of the variance and interest rate processes rarely hit zero. The simulation results in Table \ref{table3} confirm the square-root convergence of the statistical error and the first-order convergence of the discretization error. Then, using extrapolation, we obtain an approximate root mean square error (RMSE) of the reference estimate:
\begin{equation*}
\text{Bias}\left(\Theta^{*}\right) \approx 1.168\scalebox{0.85}{$\times$}10^{-4},\hspace{4pt} \text{StDev}\left(\Theta^{*}\right) \approx 0.593\scalebox{0.85}{$\times$}10^{-4} \hspace{4pt}\Rightarrow\hspace{4pt} \text{RMSE}\left(\Theta^{*}\right) \approx 1.310\scalebox{0.85}{$\times$}10^{-4}.
\end{equation*}
This is equivalent to an RMSE of about $0.001\%$ of the actual option price, suggesting that the reference estimate $\Theta^{*}=12.11968$ is accurate to three decimal places.

Furthermore, the standard deviation is reduced by $83\%$, i.e., is approximately six times lower with the mixed method. However, the variance reduction is not consistent across the range of possible values of the model parameters, and is most sensitive to changes in the correlations between the exchange rate and the squared volatility or the interest rates, i.e., $\rho_{sv}$, $\rho_{sd}$ and $\rho_{s\hspace{-.7pt}f}$, the speed of mean reversion $k$ and the volatility of volatility $\xi$.

The values $\rho_{v\hspace{.2pt}d}=0.12$, $\rho_{v\hspace{-.7pt}f}=0.05$, $\rho_{d\hspace{.01pt}f}=0.25$ and $\rho_{sv}$, $\rho_{sd}$, $\rho_{s\hspace{-.7pt}f} \in [-0.5,0.5]$ lead to a valid (i.e., positive definite) correlation matrix. The data in Figure \ref{fig:1} suggest that the highest variance reduction is achieved when the absolute values of the correlations $\rho_{sv}$, $\rho_{sd}$ and $\rho_{s\hspace{-.7pt}f}$ are small or, more precisely, when $\rho_{sv}\approx-0.05$ and $\rho_{sd}\approx\rho_{s\hspace{-.7pt}f}\approx0$, in which case the standard deviation is reduced by a factor of $20$. Hence, the best performance of the mixed estimator coincides with the independence of $W^{s}$ from the Brownian motions $W^{v}$, $W^{d}$ and $W^{f}$, an observation which is consistent with Remarks \ref{Rem4.3} and \ref{Rem4.4}.
\begin{figure}[htb]
\begin{center}
\includegraphics[width=0.95\columnwidth,keepaspectratio]{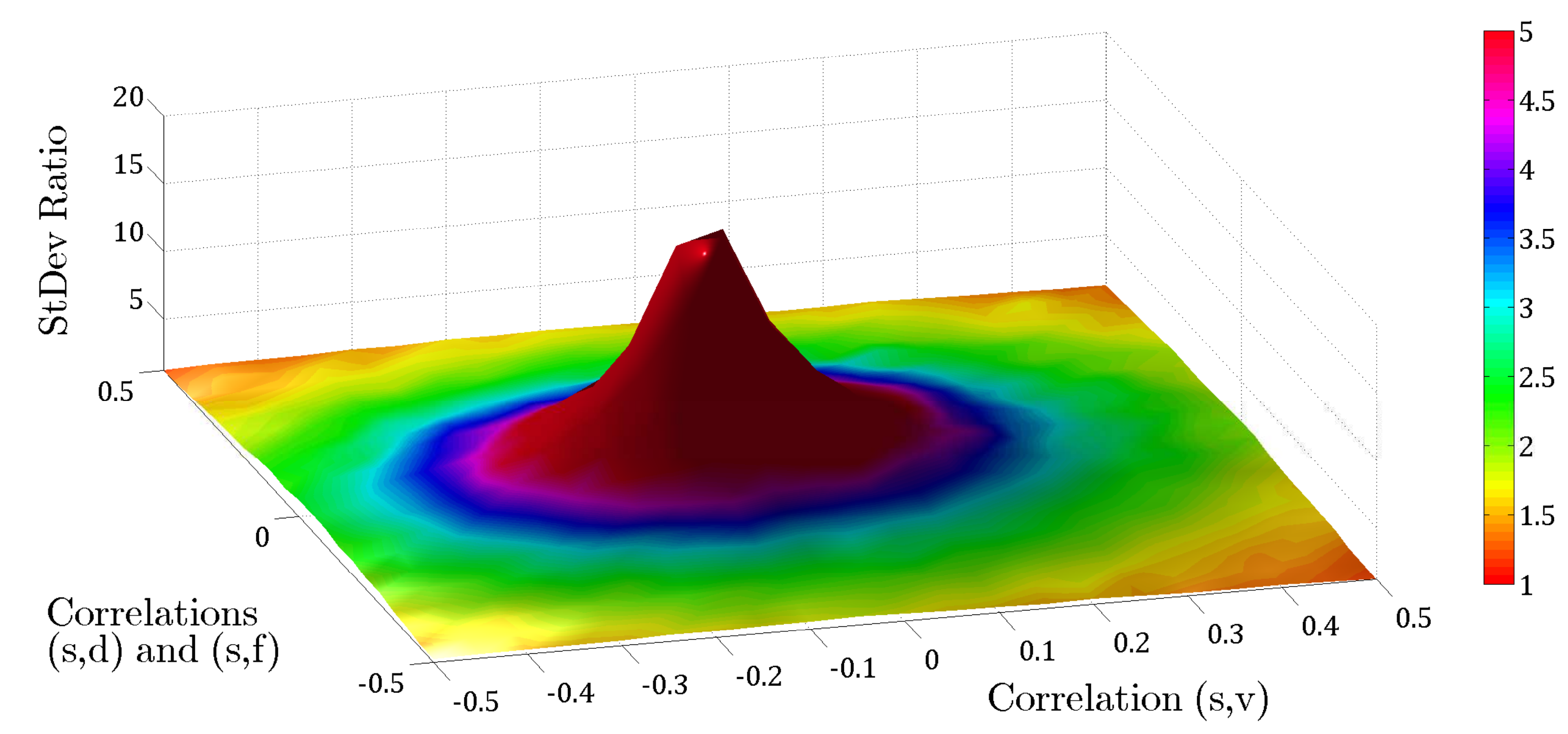}
\caption{The standard deviation ratio $\Gamma_{\text{dev}}$ from \eqref{eq4.8} with $4000$ simulations and $10$ time steps, plotted against the correlation coefficients $\rho_{sv}$, $\rho_{sd}$ and $\rho_{s\hspace{-.7pt}f}$ when the last two are equal.}
\label{fig:1}
\end{center}
\end{figure}

Furthermore, we infer from Figure \ref{fig:1} that, in addition to the lower computational cost, the mixed Monte Carlo/PDE method outperforms standard Monte Carlo in terms of accuracy for all $\rho_{sv}$, $\rho_{sd}$, $\rho_{s\hspace{-.7pt}f} \in [-0.5,0.5]$, and that the variance reduction factor decreases as the absolute values of the correlations between the exchange rate and the squared volatility or the interest rates increase, two facts that were recognized in Remark \ref{Rem4.2}.

Assuming that $\rho_{sd}=\rho_{s\hspace{-.7pt}f}$, since $\rho_{v\hspace{.2pt}d}\approx\rho_{v\hspace{-.7pt}f}\approx0$, we may approximate
\begin{equation}\label{eq5.1}
a_{11} \hspace{1pt}\approx\hspace{1pt} \sqrt{1-\rho_{sv}^{2}-\frac{2}{1+\rho_{d\hspace{.01pt}f}}\hspace{1pt}\rho_{sd}^{2}} \hspace{1pt}=\hspace{1pt} \sqrt{1-\rho_{sv}^{2}-\frac{8}{5}\hspace{1pt}\rho_{sd}^{2}}\hspace{1.5pt}.
\end{equation}
However,
\begin{equation}\label{eq5.2}
\Gamma_{\text{dev}} \hspace{1pt}\approx\hspace{1pt} \big(1-a_{11}^{2}\big)^{-\frac{1}{2}} \hspace{1pt}\approx\hspace{1pt} \big(\rho_{sv}^{2}+1.6\hspace{.5pt}\rho_{sd}^{2}\big)^{-\frac{1}{2}}
\end{equation}
from Remark \ref{Rem4.6}. Combined with the above observation on the location of the maximum variance reduction factor, we infer that the set of points $(\rho_{sv},\rho_{sd})$ corresponding to $\Gamma_{\text{dev}}=\mu$, for some $\mu\geq1$, takes approximately the form of an ellipse with the equation:
\begin{equation}\label{eq5.3}
\mu^{2}\big(\rho_{sv}+0.05\big)^{2}+1.6\hspace{1pt}\mu^{2}\rho_{sd}^{2} = 1.
\end{equation}
This confirms that the isolines in Figure \ref{fig:1} display an elliptic shape, a fact also illustrated in Table \ref{table4}, where the correlation pairs $(\rho_{sv},\rho_{sd})$ are chosen so that $\mu=2.812$. The estimated standard deviation ratio along the contour line is $3.148\pm0.239$, which is close to the theoretical value $\mu$. This attests to the accuracy of our approximations and numerical results.
\begin{table}[htb]\renewcommand{\arraystretch}{1.10}\addtolength{\tabcolsep}{0pt}\small
\begin{center}
\caption{The estimated standard deviation ratio for different correlations.}\label{table4}
\begin{tabular}{ c r r r r r r r r }
  \toprule[.1em]
	\addlinespace[4pt]
  $\rho_{sv}$ & $-0.40$ & $-0.40$ & $-0.30$ & $-0.30$ & $0.20$ & $0.20$ & $0.30$ & $0.30$ \\
	\addlinespace[1pt]
  $\rho_{sd}$ &	$-0.05$ & $0.05$ & $-0.20$ & $0.20$ & $-0.20$ & $0.20$ & $-0.05$ & $0.05$ \\
	\addlinespace[1pt]
  $\Gamma_{\text{dev}}$ &\hspace{4pt} $3.288$ &\hspace{4pt} $3.088$ &\hspace{4pt} $3.387$ &\hspace{4pt} $2.910$ &\hspace{4pt} $2.951$ &\hspace{4pt} $3.127$ &\hspace{4pt} $2.952$ &\hspace{4pt} $3.055$ \\
	\addlinespace[1pt]
	\bottomrule[.1em]
\end{tabular}
\end{center}
\end{table}

Next, we assume a partial correlation structure between the Brownian drivers such that the squared volatility dynamics are independent of the domestic and foreign interest rate dynamics, i.e., $\rho_{v\hspace{.2pt}d}=\rho_{v\hspace{-.7pt}f}=0$. Using Remark \ref{Rem4.6} and computing $a_{11}$ explicitly from \eqref{eq2.2},
\begin{equation}\label{eq5.4}
\Gamma_{\text{dev}} \hspace{1pt}\approx\hspace{1pt} \bigg[\rho_{sv}^{2}+\frac{1}{1-\rho_{d\hspace{.01pt}f}^{2}}\hspace{0pt}\big(\rho_{sd}^{2}+\rho_{s\hspace{-.7pt}f}^{2}-2\rho_{sd}\rho_{s\hspace{-.7pt}f}\rho_{d\hspace{.01pt}f}\big)\bigg]^{-\frac{1}{2}}.
\end{equation}
Before, we fixed $\rho_{d\hspace{.01pt}f}$ and analyzed the variance reduction with respect to $\rho_{sv}$, $\rho_{sd}$ and $\rho_{s\hspace{-.7pt}f}$, when the last two were equal. Now, we fix $\rho_{sv}$ instead and focus on the effect of varying~$\rho_{d\hspace{.01pt}f}$ on $\Gamma_{\text{dev}}$. To this end, one can easily show that
\begin{equation}\label{eq5.5}
\frac{\rho_{sd}^{2}+\rho_{s\hspace{-.7pt}f}^{2}-2\rho_{sd}\rho_{s\hspace{-.7pt}f}\rho_{d\hspace{.01pt}f}}{1-\rho_{d\hspace{.01pt}f}^{2}} \geq
\max\big\{\rho_{sd}^{2},\hspace{1pt} \rho_{s\hspace{-.7pt}f}^{2}\hspace{-.5pt}\big\}.
\end{equation}
Assuming that $\rho_{sd}$ and $\rho_{s\hspace{-.7pt}f}$ are not simultaneously zero, equality in \eqref{eq5.5} holds when
\begin{equation}\label{eq5.6}
\rho_{d\hspace{.01pt}f} = \rho_{d\hspace{.01pt}f}^{*} \equiv \frac{\rho_{s\hspace{-.7pt}f}}{\rho_{sd}}\Ind_{|\rho_{sd}|\geq|\rho_{s\hspace{-.7pt}f}|} \hspace{1pt}+\hspace{3pt} \frac{\rho_{sd}}{\rho_{s\hspace{-.7pt}f}}\Ind_{|\rho_{sd}|<|\rho_{s\hspace{-.7pt}f}|}.
\end{equation}
Upon its substitution into \eqref{eq5.4}, we find a theoretical standard deviation ratio $\mu$. Table \ref{table5} estimates $\Gamma_{\text{dev}}$ and compares it with $\mu$, when $\rho_{sv}=-0.10$, $\rho_{s\hspace{-.7pt}f}=-0.20$, $\rho_{sd}\in[-0.25,0.25]$ and $\rho_{d\hspace{.01pt}f}\in[-0.85,0.85]$. Note that these correlation values guarantee a symmetric positive definite correlation matrix. Hence, when $\rho_{d\hspace{.01pt}f}^{*}$ lies outside the interval of allowed values, we estimate $\Gamma_{\text{dev}}$ using $\rho_{d\hspace{.01pt}f}=\pm0.85$ instead. The data in Table \ref{table5} suggest that the approximation \eqref{eq5.4} to the standard deviation ratio is quite accurate, especially for a strongly negative correlation between the two interest rates.
\begin{table}[htb]\renewcommand{\arraystretch}{1.10}\addtolength{\tabcolsep}{-0.3pt}\small
\begin{center}
\caption{The empirical and theoretical standard deviation ratios for different correlations.}\label{table5}
\begin{tabular}{ c r r r r r r r r r r r }
  \toprule[.1em]
	\addlinespace[4pt]
  $\rho_{sd}$ & $-0.25$ & $-0.20$ & $-0.15$ & $-0.10$ & $-0.05$ & $0$ & $0.05$ & $0.10$ & $0.15$ & $0.20$ & $0.25$ \\
	\addlinespace[1pt]
  $\rho_{d\hspace{.01pt}f}^{*}$ &	$0.80$ & $1.00$ & $0.75$ & $0.50$ & $0.25$ & $0$ & $-0.25$ & $-0.50$ & $-0.75$ & $-1.00$ & $-0.80$ \\
	\addlinespace[1pt]
  $\Gamma_{\text{dev}}$ & $4.501$ & $5.301$ & $5.305$ & $5.224$ & $5.073$ & $5.010$ & $4.850$ & $4.867$ & $4.554$ & $4.467$ & $3.788$ \\
	\addlinespace[1pt]
  $\mu$ & $3.714$ & $4.472$ & $4.472$ & $4.472$ & $4.472$ & $4.472$ & $4.472$ & $4.472$ & $4.472$ & $4.472$ & $3.714$ \\
	\addlinespace[1pt]
	\bottomrule[.1em]
\end{tabular}
\end{center}
\end{table}

By close inspection of the data in Figure \ref{fig:2}, we infer that for each $\rho_{sd}\in[-0.25,0.25]$, the highest variance reduction is achieved for $\rho_{d\hspace{.01pt}f}^{*}$ as defined in \eqref{eq5.6}. Hence, we conclude that $\mu$ exhibits the qualitative behaviour of  $\Gamma_{\text{dev}}$, so that \eqref{eq5.4} provides a good approximation to the standard deviation ratio. In fact, our observations suggest that $\mu$ acts as a lower bound. We can extend these results to a full correlation structure between the Brownian drivers as long as $\rho_{v\hspace{.2pt}d}$ and $\rho_{v\hspace{-.7pt}f}$ are close to zero, as seen before.
\begin{figure}[hbt]
\begin{center}
\includegraphics[width=0.95\columnwidth,keepaspectratio]{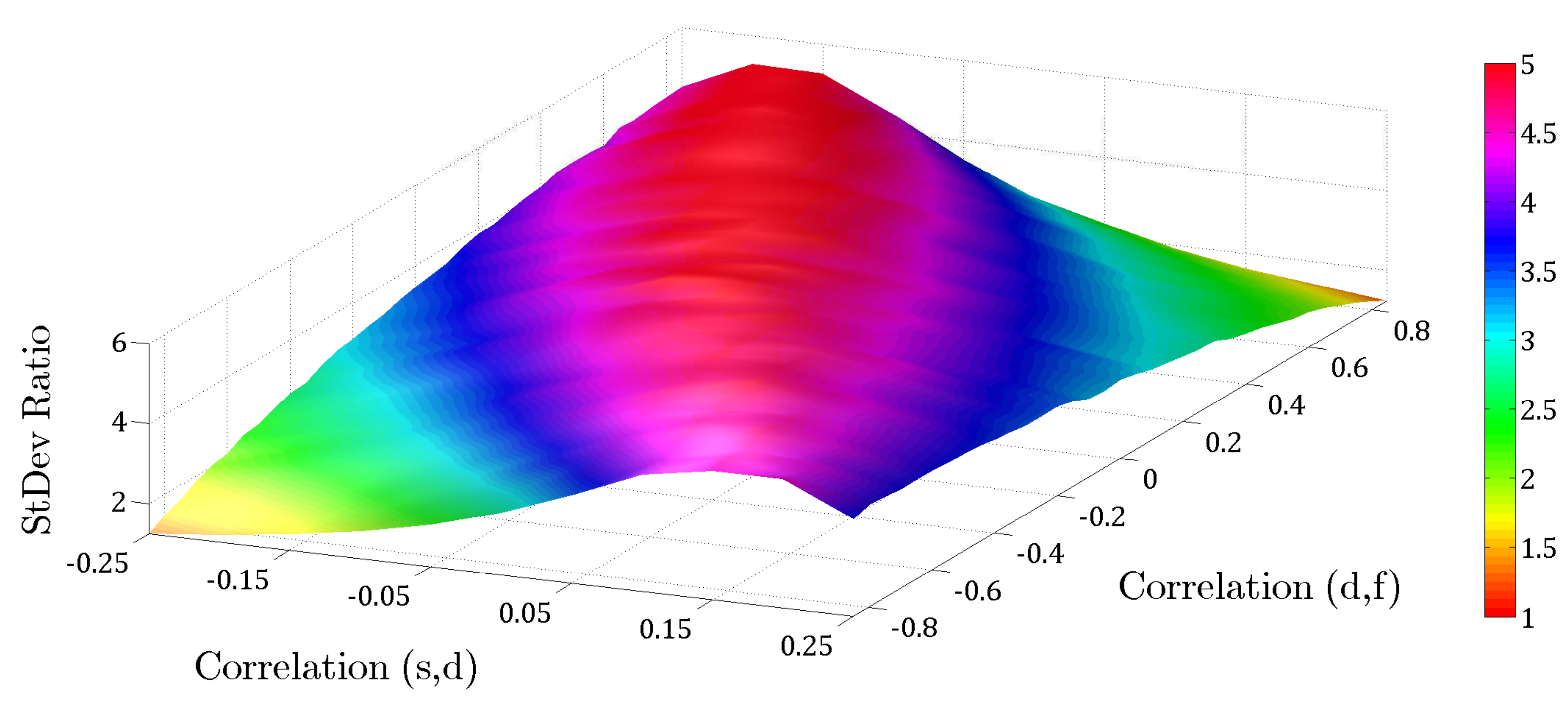}
\caption{The standard deviation ratio with $1000$ simulations and $4$ time steps, plotted against the correlation coefficients $\rho_{sd}$ and $\rho_{d\hspace{.01pt}f}$ when $\rho_{sv}=-0.10$, $\rho_{s\hspace{-.7pt}f}=-0.20$ and $\rho_{v\hspace{.2pt}d}=\rho_{v\hspace{-.7pt}f}=0$.}
\label{fig:2}
\end{center}
\end{figure}

Suppose that the exchange rate dynamics are independent of the interest rate dynamics, i.e., that $\rho_{sd}=\rho_{s\hspace{-.7pt}f}=0$, and define the optimal correlation $\rho_{sv}^{*}$ to be the value corresponding to the maximum standard deviation ratio, for a given volatility of volatility $\xi$. Then the data in Figure \ref{fig:3} suggest that the highest variance reduction is achieved when the absolute value of the correlation is small. In fact, we notice two things. First of all, that the optimal correlation approaches zero as the volatility of volatility decreases, i.e., that $\lim_{\xi\to0}\rho_{sv}^{*}=0$, which is to be expected from Remark \ref{Rem4.3}. And second, that the standard deviation ratio at $\rho_{sv}^{*}$ increases as the volatility of volatility decreases. In practice, $\xi_{d,f}\ll1$ (see Table \ref{table:2}), so the two interest rates have little impact on the variance of the mixed Monte Carlo/PDE estimator. Therefore, since $\rho_{sv}^{*}$ is close to zero, the variance comes mainly from $\sigma$ defined in \eqref{eq2.9}. On the other hand, smaller values of $\xi$ result in a smaller variance of $\sigma$, and hence of the mixed estimator as well, which leads to a higher standard deviation ratio.

We observe a similar behaviour when increasing the speed of mean reversion $k$ or the long-run variance $\theta$. Larger values of $k$ result in a smaller variance of $\sigma$ because of the mean-reverting property of the squared volatility, which ensures that the process returns to the long-run average quickly. On the other hand, larger values of $\theta$ produce higher volatilities, which then leads to an increase in the variance of the standard Monte Carlo estimator due to the larger diffusion term in the SDE driving the exchange rate process. We conclude the analysis by noting that values of $\rho_{sd}$ and $\rho_{s\hspace{-.7pt}f}$ close to zero produce similar results, to some extent. However, when the absolute values of the two correlations are not small, the impact of $\sigma$, and hence of $k$, $\xi$ and $\theta$, on the variance of the mixed estimator is reduced.
\begin{figure}[hbt]
\begin{center}
\includegraphics[width=0.95\columnwidth,keepaspectratio]{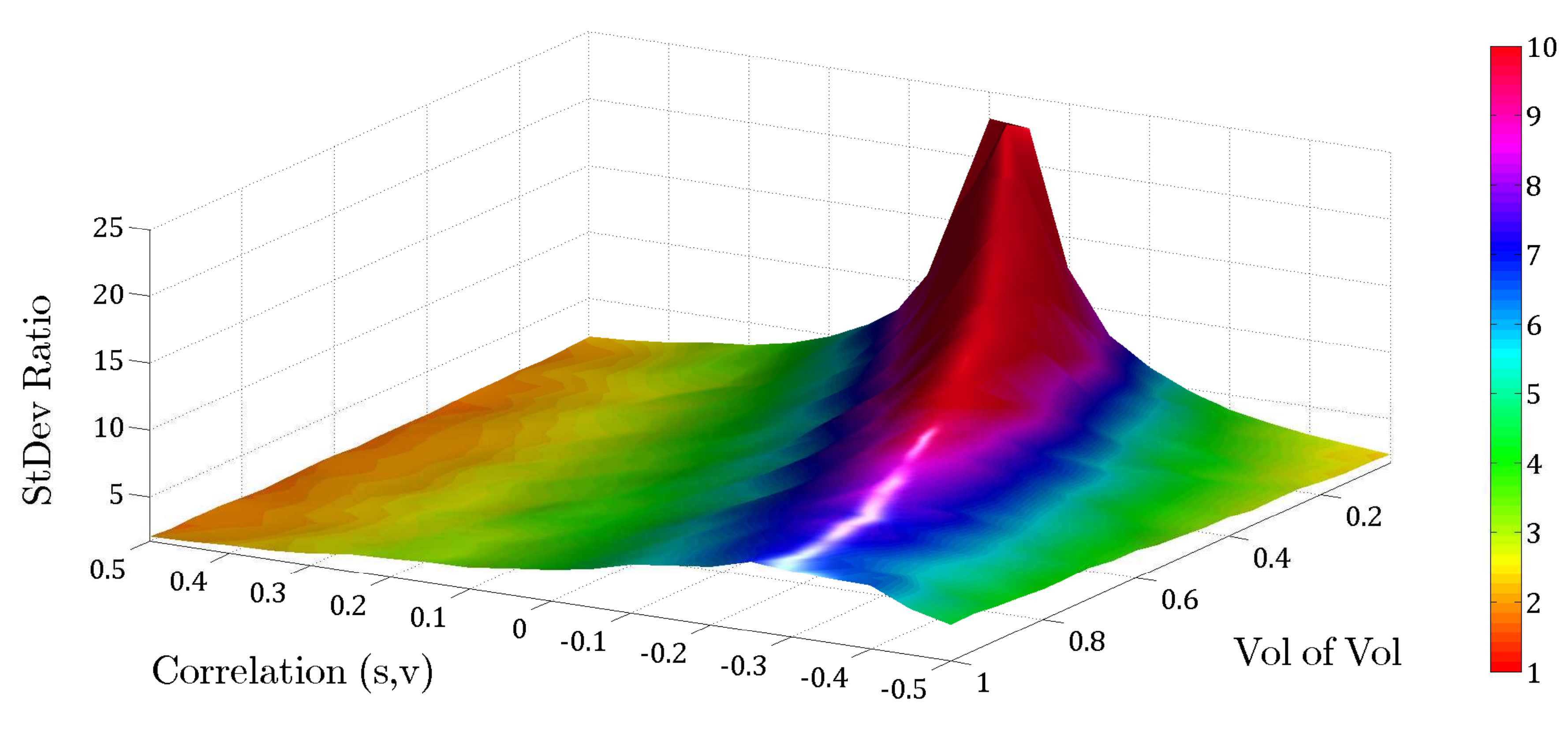}
\caption{The standard deviation ratio with $4000$ simulations and $10$ time steps, plotted against the correlation $\rho_{sv}$ and the volatility of volatility $\xi$ when $\rho_{sd}=\rho_{s\hspace{-.7pt}f}=0$.}
\label{fig:3}
\end{center}
\end{figure}

The maximum in Figure \ref{fig:3} is attained around $\rho_{sv}=0$ and $\xi=0.05$, where the standard deviation ratio is $\Gamma_{\text{dev}}=23$. Therefore, the same level of accuracy with the standard Monte Carlo method requires $529$ times more simulations.

Next, we examine the variance reduction for different spots and maturities. The data in Figure \ref{fig:4} suggest that unless the option is far out-of-the-money and the maturity is small, varying $S_{0}$ and $T$ has little impact on the standard deviation ratio, which is approximately $2$. Computing the option price with standard Monte Carlo means integrating the payoff function, which is not differentiable at the strike, whereas with the mixed Monte Carlo/PDE method we integrate the smooth conditional price. As the probability of a positive payoff decreases with $S_{0}$ and $T$, estimating it accurately with the former requires more simulations. Hence, the relative standard error of the standard Monte Carlo method increases as we go farther out-of-the-money or approach maturity, and the benefit of employing the mixed algorithm becomes clear. For example, for a 3-month call with a spot at 75\% of the strike, we observe a variance reduction factor of 5275.
\begin{figure}[hbt]
\begin{center}
\includegraphics[width=0.95\columnwidth,keepaspectratio]{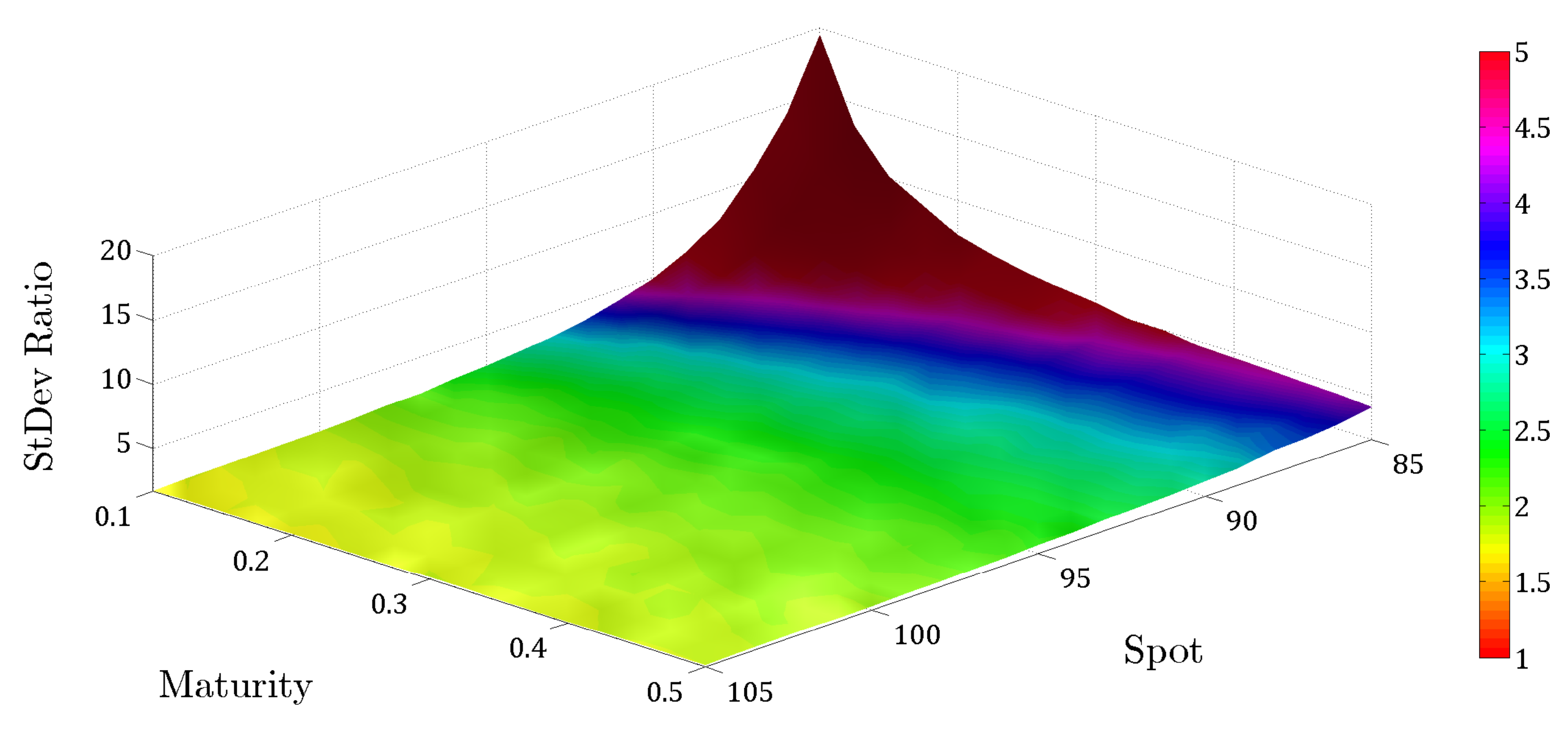}
\caption{The standard deviation ratio with $1000$ simulations and $4$ time steps, plotted against the spot $S_{0}$ and the maturity $T$ when $\rho_{sv}=-0.60$.}
\label{fig:4}
\end{center}
\end{figure}

Furthermore, Figure \ref{fig:4} suggests that the approximation for the standard deviation ratio derived in \eqref{eq4.28} does not hold when the option is far out-of-the-money. For instance, we compute $a_{11}=0.7893$, which gives a theoretical standard deviation ratio $\mu=1.629$. When $S_{0}=105$ and $T=0.5$, this is close to the estimated value $\Gamma_{\text{dev}}=1.832$. However, when $S_{0}=85$ and $T=0.1$, we observe a much higher standard deviation ratio $\Gamma_{\text{dev}}=19.863$.

Suppose that the domestic and the foreign short rate dynamics are independent of each other, and also independent of the dynamics of the exchange rate and its volatility, and let $\rho_{sv}=-0.10$, as before. Moreover, suppose that the other model parameters, as well as the spot, the strike and the maturity, take the values listed at the beginning of this section. We use the mixed Monte Carlo/PDE method with $M=8\scalebox{0.85}{$\times$}10^{7}$ simulations and $N=200$ time steps, for an RMSE of about $0.001\%$, to obtain a call price estimate: $\Theta'=12.13621$, which is about $0.14\%$ higher than the estimate corresponding to a full correlation of the factors, i.e., $\Theta^{*}=12.11968$. On the one hand, the postulated independence of the factors is critical from the point of view of analytical tractability \citep{Ahlip:2013}, but can result in fairly different option prices. On the other hand, a full correlation structure leads to a richer model and a better fit to the observed market data. Finally, we test the accuracy of the mixed method and employ the semi-analytical pricing formula of \citet{Ahlip:2013} to find the true option price, $\Theta=12.13603$, and thus a relative error of $\Theta'$ of about $0.0015\%$, which confirms that mixed Monte Carlo/PDE estimates are correct.

\subsection{Up-and-out put option}\label{subsec:barrier}

Let the spot, the strike, the barrier and the maturity be: $S_{0}=100$, $K=105$, $B=110$ and $T=0.25$, and consider a continuously monitored up-and-out put option. We will first value the contract using the mixed Monte Carlo/PDE method. Hence, for a specific realization of the variance and interest rates paths, we compute the conditional option price, i.e.,
\begin{equation}\label{eq5.7}
u(t,x) = \E\Big[e^{-\int_{t}^{T}{\oversymb{r}^{d}_{u}du}}\hspace{1pt}\big(K-\hspace{1.5pt}\oversymb{\hspace{-1.5pt}S}_{T}\big)^{+}\Ind_{\max_{t\leq u\leq T}\hspace{1.5pt}\oversymb{\hspace{-1.5pt}S}_{u}\hspace{.5pt}<\hspace{1pt}B}\hspace{0pt}\big|\hspace{2pt}\mathcal{G}_{T}^{f,d,v},\hspace{1pt} \hspace{1.5pt}\oversymb{\hspace{-1.5pt}S}_{t}=x\hspace{.5pt}\Big].
\end{equation}
We know that $u$ satisfies the following initial boundary value problem:
\begin{align}\label{eq5.8}
\partial_{t}u + \mu_{t}\hspace{1pt}x\hspace{1pt}\partial_{x}u + \frac{1}{2}\hspace{1pt}a_{11}^{2}\oversymb{V}_{\hspace{-2.5pt}t}\hspace{1pt}x^{2}\partial_{xx}u - \oversymb{r}^{d}_{t}\hspace{1pt}u = 0&,\hspace{3pt} \forall\hspace{1pt} 0<x<B,\hspace{2pt} t<T \\[2pt]
u(t,B)=0&,\hspace{3pt} \forall\hspace{.5pt} t\leq T \nonumber\\[4pt]
u(T,x)=(K-x)^{+}&,\hspace{3pt} \forall\hspace{1pt} 0\leq x<B. \nonumber
\end{align}
We solve the PDE backwards in time from the initial condition, on a rectangular domain with $t\in[0,T]$ and $x\in[0.7S_{0},B]$ that is discretized on a uniform grid with $N+1$ temporal nodes and $L+1$ spatial nodes. We selected this particular lower boundary of the spatial computational domain to reduce the number of spatial nodes, at the same time making sure that the truncation error arising from our choice of the domain is negligible. We employ a central difference scheme to approximate the spatial derivatives and a linearity boundary condition \citep{Tavella:2000} stating that $\partial_{xx}u=0$ at the lower boundary, where the option is deep-in-the-money and the price can be regarded as linear in $x$. For convenience, we have employed the same time grid used for the discretization of the squared volatility and the interest rates. The final estimate of the barrier option price is a Monte Carlo average over a sufficiently large number of discrete trajectories of the Brownian motions $W^{2}$, $W^{3}$ and $W^{4}$.

Alternatively, we can use the perturbative formula of \citet{Fatone:2008} to approximate the conditional option price, and then a simple Monte Carlo average to estimate the outer expectation. Hence, we call this numerical scheme the mixed Monte Carlo/Pert method. \citet{Fatone:2008} approximate the up-and-out put option price in the Black-Scholes model with time-dependent parameters via a series expansion and provide explicit formulae for the first three terms, which involve some elementary and nonelementary transcendental functions. However, we will focus only on the zeroth-order approximation because using a first-order correction term results in a hundredfold increase in computation time, and hence in a poor performance of the scheme as opposed to the mixed Monte Carlo/PDE method. Using $M=4\scalebox{0.85}{$\times$}10^{7}$ simulations and $N=200$ time steps to minimize the sampling and the discretization errors, respectively, we obtain a zeroth-order approximation: $\Theta^{0}=5.7700$.

Since a closed-form solution to the option pricing problem is not available, we need to find an accurate reference estimate $\Theta^{*}$ for the true option price $\Theta$ in order to compute the different errors of the numerical methods. Therefore, we use the mixed Monte Carlo/PDE algorithm with the Crank-Nicolson scheme, with $M=4\scalebox{0.85}{$\times$}10^{7}$ simulations, $N=200$ time steps and $L=20$ spatial steps, to find: $\Theta^{*}=5.7631$. Hence, the approximation error of the mixed Monte Carlo/Pert method is: $\Theta^{0}-\Theta^{*}=0.0069$, i.e., about $0.1\%$.

In Figure \ref{fig:5}, we report the time-discretization errors computed using the reference estimate $\Theta^{*}$ -- whose accuracy is discussed below -- or $\Theta^{0}$, and a sufficiently large number of simulations and spatial steps, i.e., $M=4\scalebox{0.85}{$\times$}10^{7}$ and $L=20$. For the standard Monte Carlo and the mixed Monte Carlo/Pert methods, the time-discretization error is defined as the bias, whereas for the mixed Monte Carlo/PDE algorithm, due to our choice of the finite difference grid, it contains the finite difference (FD) time-discretization error. Henceforth, the term ``discretization error'' stands for the time-discretization error.

On a side note, crossings of the barrier are monitored only at discrete times by standard Monte Carlo. This gives rise to a monitoring error, which is included in the discretization error. Moreover, due to the knock-out feature of the option, the true price is smaller than the Monte Carlo estimate, which explains the strong positive bias displayed in Figure \ref{fig:5}.
\begin{figure}[htb]
\begin{center}
\includegraphics[width=0.95\columnwidth,keepaspectratio]{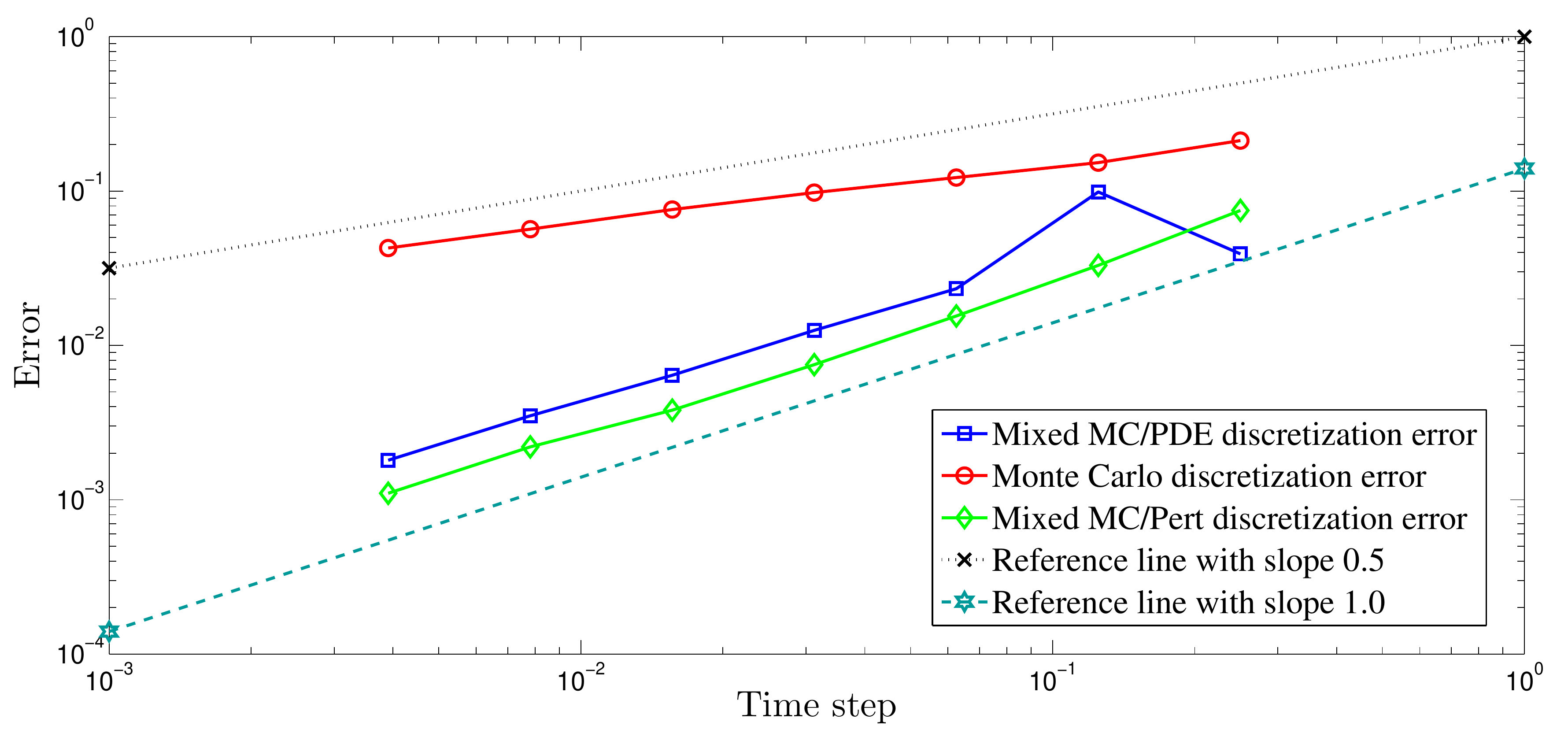}
\caption{The log-plot of the time-discretization errors of the mixed Monte Carlo/PDE, the mixed Monte Carlo/Pert and the standard Monte Carlo methods, against the time step.}
\label{fig:5}
\end{center}
\end{figure}

The data in Figure \ref{fig:5} suggest a square-root convergence of standard Monte Carlo and a first-order convergence of the mixed algorithms. On the other hand, we can use a Brownian bridge technique \citep[see][]{Glasserman:2003} to improve the first method and recover the first-order convergence. Indeed, the (red) Monte Carlo curve in Figure \ref{fig:5} would almost coincide with the (green) mixed Monte Carlo/Pert curve with the Brownian bridge correction. For instance, we calculated the discretization bias with $N=8$ time steps to be $0.0075$ for both Monte Carlo with Brownian bridge and mixed Monte Carlo/Pert.
\begin{figure}[hbt]
\begin{center}
\includegraphics[width=0.95\columnwidth,keepaspectratio]{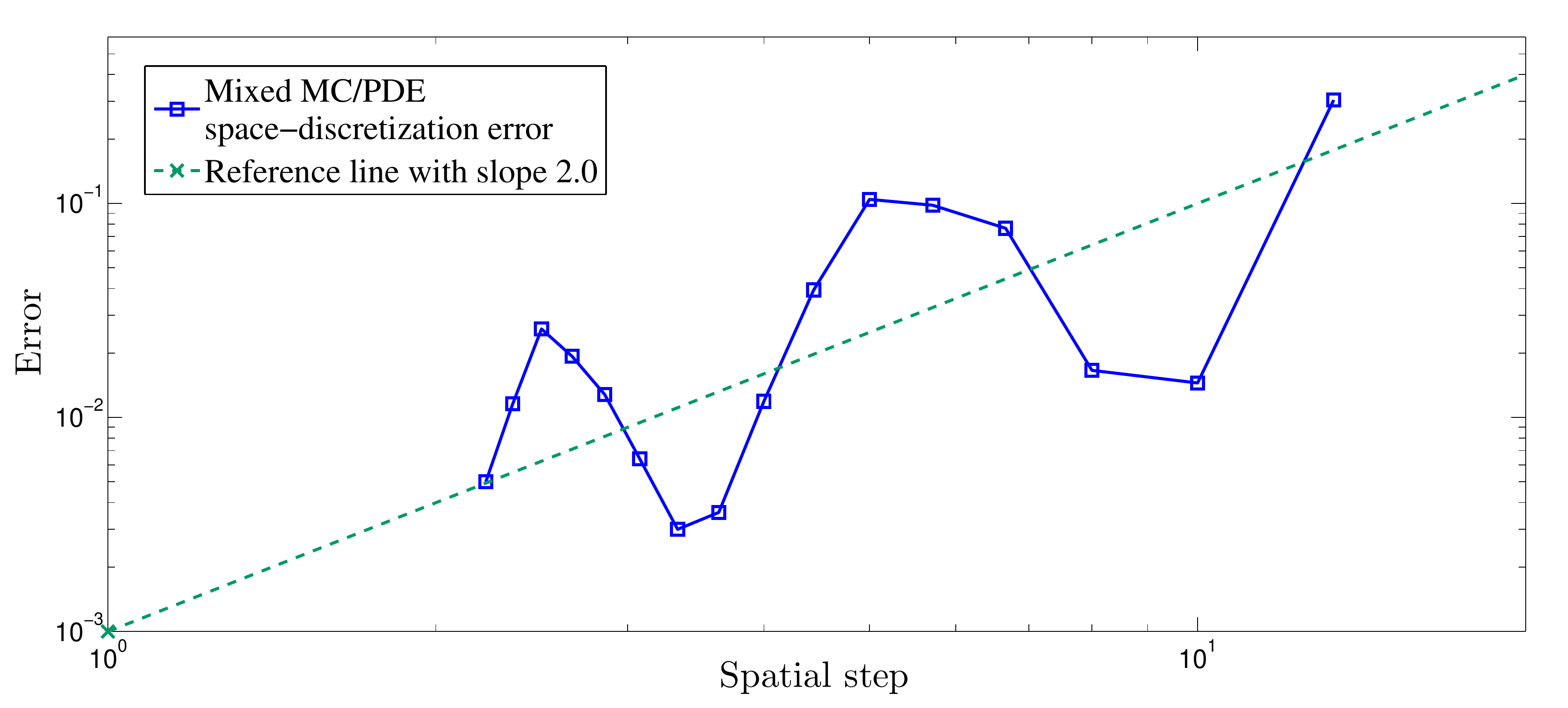}
\caption{The log-plot of the absolute space-discretization error of the mixed Monte Carlo/PDE method, against the spatial step.}
\label{fig:6}
\end{center}
\end{figure}

In Figure \ref{fig:6}, we report the space-discretization error computed using $\Theta^{*}$ and a sufficiently large number of simulations and time steps, i.e., $M=4\scalebox{0.85}{$\times$}10^{7}$ and $N=200$. The data suggest a second-order convergence as well as a local minimum discretization error when the strike price lies halfway between two adjacent nodes, a technique known as grid-shifting \citep{Tavella:2000}. Hence, considering the first-order convergence of the time-discretization error (T-Err) and the second-order convergence of the finite difference space-discretization error (S-Err) with the mixed Monte Carlo/PDE method, and using extrapolation, we obtain an approximate root mean square error (RMSE) of the reference estimate:
\begin{equation*}
\text{T-Err} \approx 5.76\scalebox{0.85}{$\times$}10^{-4},\hspace{3pt} \text{S-Err} \approx 10.80\scalebox{0.85}{$\times$}10^{-4},\hspace{3pt} \text{StDev}\approx2.09\scalebox{0.85}{$\times$}10^{-4} \hspace{3pt}\Rightarrow\hspace{3pt} \text{RMSE} \approx 1.67\scalebox{0.85}{$\times$}10^{-3}.
\end{equation*}
This is equivalent to an RMSE of about $0.03\%$ of the actual option price, suggesting that the reference estimate $\Theta^{*}=5.7631$ is accurate to two decimal places. Next, we compare the three numerical methods in terms of computation time for a given level of accuracy, in particular, when the RMSE is at most $0.30\%$ of the option price. First, using the empirical convergence rates determined above and extrapolation, we need $M=2.5\scalebox{0.85}{$\times$}10^{5}$ simulations and $N=800$ time steps, and hence a CPU time of $61.2$ seconds, with the standard Monte Carlo method. Second, we reach this level of accuracy with the mixed Monte Carlo/PDE method when $M=12\hspace{.5pt}000$, $N=10$ and $L=12$, in $2$ seconds. Third, we need to employ the mixed Monte Carlo/Pert method with a zeroth-order approximation, with $M=12\hspace{.5pt}000$ and $N=10$, which takes $3.1$ seconds.

Therefore, when the up-and-out put option price estimates need not be too accurate, e.g., when one decimal place of accuracy is sufficient, the two mixed algorithms are comparable in terms of CPU time and efficiency, and are considerably faster than standard Monte Carlo. However, a higher accuracy would require at least a first-order correction term in the mixed Monte Carlo/Pert approximation, making it highly time-consuming. We thus conclude that the mixed Monte Carlo/PDE method is the better of the three schemes.

We mentioned above that Monte Carlo with Brownian bridge recovers the observed first-order convergence of the discretization error and the level of the bias from the mixed Monte Carlo/Pert method. The time-discretization error with the mixed Monte Carlo/PDE method is approximately $1.6$ times larger, and includes the FD time-discretization error.

For a two-decimal-place accuracy, we fix $100$ time steps and $20$ spatial steps, such that the space and time-discretization errors are about $0.02\%$. The time required to obtain a barrier option price estimate with $40\hspace{.5pt}000$ simulations is then $26$ seconds for the mixed Monte Carlo/PDE method and $2.1$ seconds for Monte Carlo with Brownian bridge ($1.4$ seconds for standard Monte Carlo). Therefore, the computational cost is $92\%$ lower with the latter. In conclusion, due to the square-root convergence of the standard deviation, $\Gamma_{\text{dev}}$ needs to be above $4.5$ in order for the mixed Monte Carlo/PDE method to outperform Monte Carlo with Brownian bridge. Just as in the European call option case, the variance reduction is most sensitive to changes in the correlations between the exchange rate and the squared volatility or the interest rates.
\begin{figure}[htb]
\begin{center}
\includegraphics[width=0.95\columnwidth,keepaspectratio]{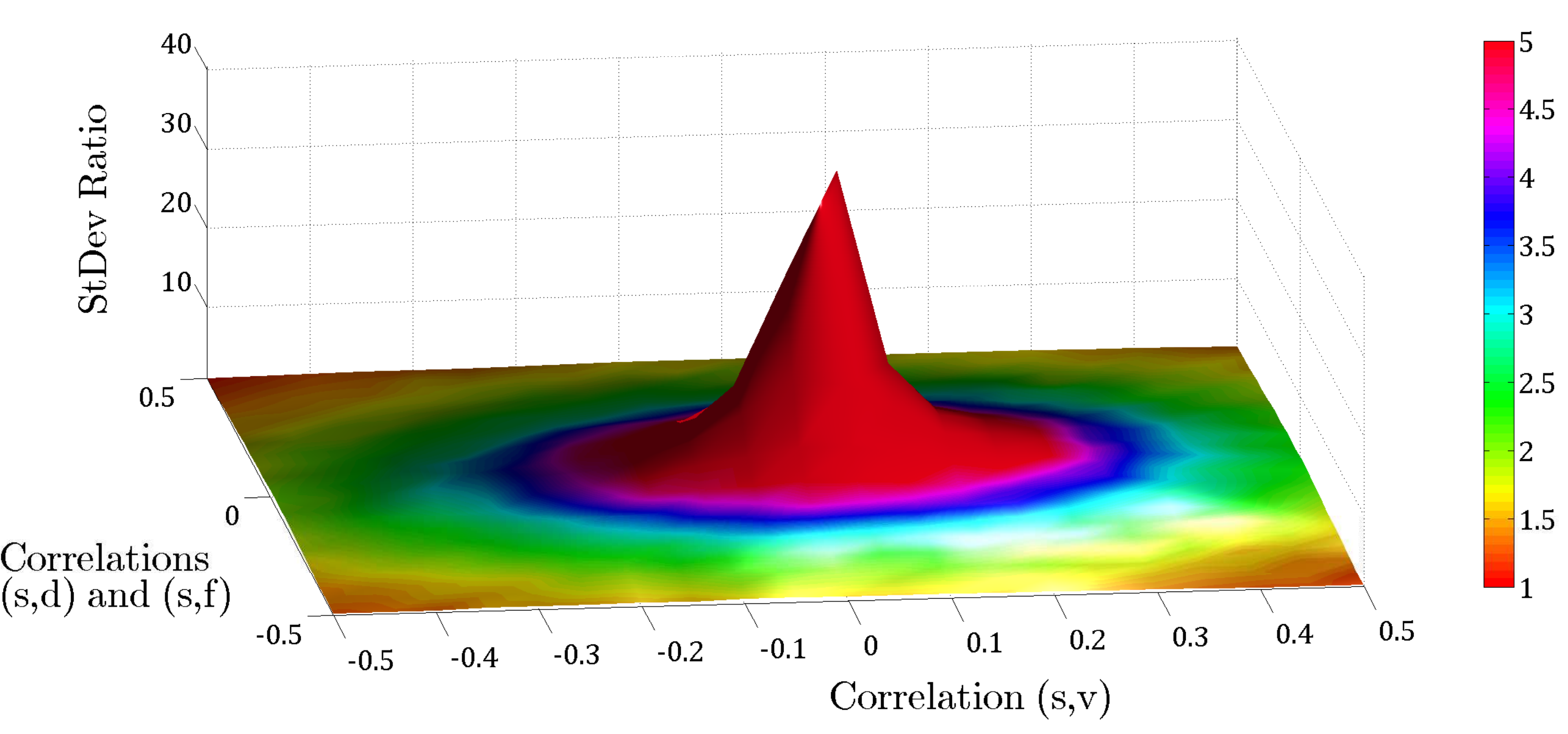}
\caption{The standard deviation ratio with $400$ simulations, $10$ time steps and $4$ spatial steps, plotted against the correlation coefficients $\rho_{sv}$, $\rho_{sd}$ and $\rho_{s\hspace{-.7pt}f}$ when the last two are equal.}
\label{fig:7}
\end{center}
\end{figure}

Figure \ref{fig:7} exhibits similar characteristics to those of Figure \ref{fig:1}. In particular, the highest variance reduction is achieved when $\rho_{sv}\approx0.05$ and $\rho_{sd}\approx\rho_{s\hspace{-.7pt}f}\approx0$, in which case $\Gamma_{\text{dev}}=40$. Based on a previous observation, this is equivalent to a $80$ times lower computational effort with the mixed algorithm. A careful inspection of the data in Figure \ref{fig:7} suggests that the set of points $(\rho_{sv},\rho_{sd})$ corresponding to $\Gamma_{\text{dev}}>4.5$ can approximately be described by the following inequality:
\begin{equation}\label{eq5.9}
\big(\rho_{sv}-0.05\big)^{2}+1.6\hspace{.5pt}\rho_{sd}^{2} < 4.2^{-2},
\end{equation}
i.e., the inside of an ellipse, this being the set of parameters where the benefit of variance reduction outweighs the additional complexity of solving the conditional PDE numerically in this instance.

The variance reduction achieved by the mixed method results in computational savings (in the number of samples) by a factor roughly independent of the desired accuracy. Conversely, higher accuracy of the finite difference method can only be achieved by a larger number of mesh points. Hence, it would appear that for high enough accuracy the mixed method can never win over the standard Monte Carlo method. An asymptotic complexity gain of the mixed method for small errors would require that the PDE can be solved with constant effort independent of the desired accuracy (and for this effort to be outweighed by the reduced number of samples required for a given statistical error). Multilevel Monte Carlo methods \citep{Giles:2012} recently developed for stochastic PDEs are designed precisely for this goal, and achieve it by concentrating the dominant number of samples on the coarsest meshes, while computing corrections on finer meshes with a vanishing number of paths. The application and numerical analysis in the present context is the subject of further research.

\section{Conclusions}\label{sec:conclusion}

The numerical experiments carried out in Section \ref{sec:numerics} suggest that the mixed method outperforms both standard Monte Carlo and finite difference methods under certain circumstances depending on the contract and the model parameters. When a closed-form solution for the conditional option price is available, we usually see a considerable improvement in terms of both accuracy and computation time. If not, the mixed algorithm provides better accuracy at the expense of added computation time and the set of parameter values where it outperforms the classical schemes is limited.

The analysis carried out in this paper is not restricted to the four-dimensional Heston-CIR model, but can be extended to higher-dimensional problems. For instance, one may consider multi-factor short rates as in \citet{Dang:2015}, with CIR dynamics and a term structure, in which case the convergence and variance reduction analysis applies with some slight modifications of the proofs. Stochastic volatility accounts for volatility clustering, dependence in the increments and long-term smiles and skews, but gives rise to unrealistic short-term patterns in the implied volatility. Hence, in order to improve the behaviour of the implied volatility for short maturities, one could extend the original model to a stochastic-local volatility model as in \citet{Cozma:2015}, which can easily be implemented for barrier option pricing with no extra computational effort, or add an independent jump component to the spot FX rate, in which case analytical formulae may be available for the conditional prices of European options. For instance, that is the case when the distribution of the jump size is normal \citep{Merton:1976} or double-exponential \citep{Kou:2002}.

However, several unsettled questions remain, like the strong convergence rate of the discretization scheme, or a finite difference scheme with an observed second-order convergence in time for pricing the barrier option. In addition, examining the hedging parameters is also relevant, and we intend to pursue all these topics in our future research.

\section*{Declaration of interest}\label{sec:declaration}

The authors report no conflicts of interest. The authors alone are responsible for the content and writing of the paper. The research of Andrei Cozma is funded by the EPSRC and the support is gratefully acknowledged.

\bibliographystyle{apa}
\bibliography{references}

\titleformat{\section}{\large\bfseries}{\appendixname~\thesection .}{0.5em}{}
\begin{appendices}

\section{Proof of Lemma \ref{Lem3.1}}\label{sec:aux3.1}

Let $\left\{\mathcal{G}_{t}^{y},\hspace{1.5pt} 0\hspace{-.5pt}\leq\hspace{-.5pt}t\hspace{-.5pt}\leq\hspace{-.5pt}T\right\}$ be the natural filtration generated by $W^{y}$ and employ the shorthand notation $\E_{t}^{y}\big[\hspace{.5pt}\cdot\hspace{.5pt}\big]=\E\big[\hspace{.5pt}\cdot\hspace{.5pt}|\mathcal{G}_{t}^{y}\big]$. If we assume that $t \in [t_{n}, t_{n+1}]$ and condition on $\mathcal{G}_{t_{n}}^{y}$, we get
\begin{align*}
\E_{t_{n}}^{y}\!\big[\hspace{1.5pt}\oversymb{\hspace{-1pt}\Theta\hspace{-1pt}}\hspace{1pt}_{t}\big] = \exp\bigg\{\lambda\int_{0}^{t_{n}}{\hspace{-.5pt}\oversymb{\hspace{.5pt}Y}_{\hspace{-2.5pt}u}\hspace{1pt}du} + \mu\int_{0}^{t_{n}}{\hspace{-2pt}\sqrt{\hspace{-.5pt}\oversymb{\hspace{.5pt}Y}_{\hspace{-2.5pt}u}}\,\frac{\delta W^{y}_{u}}{\delta t}\hspace{1.5pt}du}\bigg\}\exp\left\{\bigg[\lambda+\frac{t-t_{n}}{2\hspace{.5pt}\delta t}\hspace{1pt}\mu^{2}\bigg]\left(t-t_{n}\right)\hspace{-.5pt}\oversymb{\hspace{.5pt}Y}_{\hspace{-2.5pt}t_{n}}\right\}.
\end{align*}
Upon noticing the identity below,
\begin{equation*}
\sup_{x\in[0,1]}\hspace{1pt} \lambda\hspace{.5pt}x + \frac{1}{2}\hspace{1pt}\mu^{2}x^{2} \hspace{1.5pt}=\hspace{1.5pt} \Delta\Ind_{\Delta\hspace{.5pt}>\hspace{1pt}0}\hspace{1.5pt},
\end{equation*}
we deduce that $\Delta\hspace{-1pt}\leq\hspace{-1pt}0$ implies $\E_{t_{n}}^{y}\!\big[\hspace{1.5pt}\oversymb{\hspace{-1pt}\Theta\hspace{-1pt}}\hspace{1pt}_{t}\big] \leq \E_{t_{n}}^{y}\!\big[\hspace{1.5pt}\oversymb{\hspace{-1pt}\Theta\hspace{-1pt}}\hspace{1pt}_{t_{n}}\big]$ and $\Delta\hspace{-1pt}>\hspace{-1pt}0$ implies $\E_{t_{n}}^{y}\!\big[\hspace{1.5pt}\oversymb{\hspace{-1pt}\Theta\hspace{-1pt}}\hspace{1pt}_{t}\big] \leq \E_{t_{n}}^{y}\!\big[\hspace{1.5pt}\oversymb{\hspace{-1pt}\Theta\hspace{-1pt}}\hspace{1pt}_{t_{n+1}}\big]$. Moreover, since $\hspace{-.5pt}\oversymb{\hspace{.5pt}Y}$ is piecewise constant,
\begin{equation*}
\int_{0}^{t_{n}}{\hspace{-2pt}\sqrt{\hspace{-.5pt}\oversymb{\hspace{.5pt}Y}_{\hspace{-2.5pt}u}}\,\frac{\delta W^{y}_{u}}{\delta t}\hspace{1.5pt}du} \hspace{2pt}=\hspace{1pt} \int_{0}^{t_{n}}{\hspace{-2pt}\sqrt{\hspace{-.5pt}\oversymb{\hspace{.5pt}Y}_{\hspace{-2.5pt}u}}\,dW^{y}_{u}}\hspace{.5pt},\hspace{5pt} \forall\hspace{1pt}0\leq n\leq N.
\end{equation*}
Henceforth, we follow the argument of Proposition 3.6 in \citet{Cozma:2015.2}. \qed

\section{Proof of Proposition \ref{Prop3.2}}\label{sec:aux3.2}

We find it convenient to define a new stochastic process $L$ by
\begin{equation}\label{eq3.7}
L_{t} \equiv S_{t}\exp\left\{\int_{0}^{t}{r^{f}_{u}\hspace{1pt}du}\right\} = S_{0}\exp\left\{\int_{0}^{t}{\Big(r^{d}_{u}-\frac{1}{2}\hspace{1pt}v_{u}\Big)du} + \int_{0}^{t}{\sqrt{v_{u}}\hspace{1.5pt}dW_{u}^{s}}\right\}.
\end{equation}
Since $S_{t}\leq L_{t}$, $\forall\hspace{1pt}t \in [0,T]$, it suffices to prove the finiteness of the supremum over $t$ of
\begin{equation}\label{eq3.8}
\E\big[L_{t}^{\omega}\big] = S_{0}^{\omega}\E\bigg[\exp\bigg\{\omega\int_{0}^{t}{r^{d}_{u}\hspace{1pt}du} - \frac{\omega}{2}\int_{0}^{t}{v_{u}\hspace{1pt}du} + \omega\sum_{j=1}^{4}a_{1\hspace{-0.5pt}j}\int_{0}^{t}{\sqrt{v_{u}}\hspace{1.5pt}dW_{u}^{j}}\bigg\}\bigg].
\end{equation}
Let $\big\{\mathcal{G}_{t}^{d,v},\hspace{1.5pt} 0\hspace{-.5pt}\leq\hspace{-.5pt}t\hspace{-.5pt}\leq\hspace{-.5pt}T\big\}$ be the natural filtration generated by the Brownian drivers $W^{3}$ and $W^{4}$, i.e., generated by the processes $r^{d}$ and $v$ as observed until time $T$, and $\big\{\mathcal{G}_{t}^{v},\hspace{1.5pt} 0\hspace{-.5pt}\leq\hspace{-.5pt}t\hspace{-.5pt}\leq\hspace{-.5pt}T\big\}$ be the filtration generated by $W^{4}$. Conditioning the expectation on the right-hand side of \eqref{eq3.8} on the $\sigma$-algebra $\mathcal{G}^{d, v}_{t}$ and taking into account that $W^{1}$ and $W^{2}$ are independent, we can compute the inner expectation using moment generating functions (MGFs):
\begin{align}\label{eq3.11}
&\E\left[\exp\bigg\{\omega a_{11}\int_{0}^{t}{\sqrt{v_{u}}\hspace{1.5pt}dW_{u}^{1}} + \omega a_{12}\int_{0}^{t}{\sqrt{v_{u}}\hspace{1.5pt}dW_{u}^{2}}\bigg\}\Big|\,\mathcal{G}^{d, v}_{t}\right] \nonumber\\[3pt]
&\qquad = \E\left[\exp\bigg\{\omega a_{11}\int_{0}^{t}{\sqrt{v_{u}}\hspace{1.5pt}dW_{u}^{1}}\bigg\}\Big|\,\mathcal{G}^{d, v}_{t}\right]\E\left[\exp\bigg\{\omega a_{12}\int_{0}^{t}{\sqrt{v_{u}}\hspace{1.5pt}dW_{u}^{2}}\bigg\}\Big|\,\mathcal{G}^{d, v}_{t}\right] \nonumber\\[3pt]
&\qquad = \exp\left\{\frac{\omega^{2}}{2}\big(a_{11}^{2}+a_{12}^{2}\big)\int_{0}^{t}{v_{u}\hspace{1pt}du}\right\}.
\end{align}
Substituting back into \eqref{eq3.8} with \eqref{eq3.11} leads to
\begin{equation*}
\E\big[L_{t}^{\omega}\big] = S_{0}^{\omega}\E\bigg[\exp\bigg\{\omega\hspace{-1pt}\int_{0}^{t}{r^{d}_{u}\hspace{1pt}du} + \left[\frac{\omega^{2}}{2}\big(a_{11}^{2}\hspace{-1pt}+a_{12}^{2}\big) - \frac{\omega}{2}\right]\int_{0}^{t}{v_{u}\hspace{1pt}du} + \omega\sum_{j=3}^{4}a_{1\hspace{-0.5pt}j}\hspace{-1pt}\int_{0}^{t}{\hspace{-1pt}\sqrt{v_{u}}\hspace{1.5pt}dW_{u}^{j}}\bigg\}\bigg].
\end{equation*}
Next, we employ H\"older's inequality with the pair $(p,q)$, where $p,q>1$ and $q=p/(p-1)$, in order to force the term $\int_{0}^{t}{r_{u}^{d}\hspace{1pt}du}$ outside the expectation, and then condition the second expectation on the $\sigma$-algebra $\mathcal{G}_{t}^{v}$ to arrive at
\begin{align}\label{eq3.13}
\E\big[L_{t}^{\omega}\big] &\leq S_{0}^{\omega}\E\bigg[\exp\bigg\{p\omega\int_{0}^{t}{r^{d}_{u}\hspace{1pt}du}\bigg\}\bigg]^{\frac{1}{p}}\E\bigg[\exp\bigg\{q\left[\frac{\omega^{2}}{2}\big(a_{11}^{2}+a_{12}^{2}+qa_{13}^{2}\big) - \frac{\omega}{2}\right]\int_{0}^{t}{v_{u}\hspace{1pt}du} \nonumber\\[2pt]
&+ q\omega a_{14}\int_{0}^{t}{\sqrt{v_{u}}\hspace{1.5pt}dW_{u}^{4}}\bigg\}\bigg]^{\frac{1}{q}}.
\end{align}
All that is left to do is to show that the supremum over $t$ of each of the two expectations on the right-hand side of \eqref{eq3.13} is finite. However, Proposition 3.2 in \citet{Cozma:2015.2} provides the following sufficient conditions:
\begin{equation}\label{eq3.18}
k_{d} \geq \sqrt{2p\omega}\hspace{1pt}\xi_{d}\hspace{1pt},
\end{equation}
as well as
\begin{equation}\label{eq3.17}
k\geq q\omega\rho_{sv}\xi
\end{equation}
and
\begin{equation}\label{eq3.16}
\omega^{2}\xi^{2}a_{13}^{2}q^{2} + \left[2\omega\rho_{sv}\xi k + \omega^{2}\xi^{2}\left(a_{11}^{2}+a_{12}^{2}\right) - \omega\xi^{2}\right]q - k^{2} \leq 0,
\end{equation}
for all $T>0$. The first assumption in \eqref{eq3.5} ensures that $q_{0}(\alpha)\hspace{-1pt}>\hspace{-1pt}1$, and hence that $q_{1}(\alpha)\hspace{-1pt}>\hspace{-1pt}1$. This implies that $q_{1}(\alpha)/(q_{1}(\alpha)-1)\hspace{-1pt}>\hspace{-1pt}1$. Due to the second assumption  in \eqref{eq3.5}, we can find $p\hspace{-1pt}>\hspace{-1pt}1$ so that
\begin{equation}\label{eq3.19}
\frac{k_{d}^{2}}{2\xi_{d}^{2}} > \alpha p > \frac{\alpha\hspace{1.5pt}q_{1}(\alpha)}{q_{1}(\alpha)-1} \hspace{3pt}\Rightarrow\hspace{3pt} k_{d} > \sqrt{2p\alpha}\hspace{1pt}\xi_{d}\hspace{1pt}.
\end{equation}
The quadratic equation in $x$ below has roots of different signs, with positive root $q_{0}(\alpha)$:
\begin{align*}
\alpha^{2}\xi^{2}a_{13}^{2}x^{2} + \big[2\alpha\rho_{sv}\xi k + \alpha^{2}\xi^{2}\left(a_{11}^{2}+a_{12}^{2}\right) - \alpha\xi^{2}\big]x - k^{2} = 0.
\end{align*}
However, the H\"older pair satisfies $p=q/(q-1)$, so $q<q_{1}(\alpha)\leq q_{0}(\alpha)$. Therefore, $q$ lies in between the two roots of the quadratic, which implies that
\begin{align}\label{eq3.20}
\alpha^{2}\xi^{2}a_{13}^{2}q^{2} + \big[2\alpha\rho_{sv}\xi k + \alpha^{2}\xi^{2}\left(a_{11}^{2}+a_{12}^{2}\right) - \alpha\xi^{2}\big]q - k^{2} < 0.
\end{align}
From \eqref{eq3.4}, if $\rho_{sv}>0$,
\begin{equation}\label{eq3.21}
q<q_{1}(\alpha)\leq \frac{k}{\alpha\rho_{sv}\xi} \hspace{3pt}\Rightarrow\hspace{3pt} k>q\alpha\rho_{sv}\xi,
\end{equation}
and this clearly holds when the correlation coefficient is non-positive. Consider the three continuous maps below, which are strictly positive when $\omega = \alpha$,
\begin{align*}
  \begin{dcases}
	\omega \mapsto k - q\omega\rho_{sv}\xi\hspace{.5pt};\hspace{5pt} \omega \mapsto k_{d} - \sqrt{2p\omega}\hspace{1pt}\xi_{d}\hspace{1pt}; \\[4pt]
	\omega \mapsto k^{2} - \omega^{2}\xi^{2}a_{13}^{2}q^{2} - \big[2\omega\rho_{sv}\xi k + \omega^{2}\xi^{2}\left(a_{11}^{2}+a_{12}^{2}\right) - \omega\xi^{2}\big]q.
	\end{dcases}
\end{align*}
Then we can find $\alpha_{1}>\alpha$ such that all three functions are positive on $[\alpha,\alpha_{1})$. We have thus proved that conditions \eqref{eq3.18} -- \eqref{eq3.16} are satisfied and the conclusion follows. The extension to the interval $[1,\alpha_{1})$ follows immediately from Jensen's inequality. \qed

\vspace{1em}
In the special case that $a_{13}=0$, i.e., $\rho_{sd}=\rho_{sv}\rho_{v\hspace{.2pt}d}$, the argument is the same and the only difference appears in condition \eqref{eq3.16}, which becomes
\begin{align*}
\big[2\omega\rho_{sv}\xi k + \omega^{2}\xi^{2}\left(1-\rho_{sv}^{2}\right) - \omega\xi^{2}\big]q - k^{2} \leq 0.
\end{align*}
Henceforth, one can easily show that Proposition \ref{Prop3.2} still holds in this case as long as
\begin{equation*}
k > \alpha\rho_{sv}\xi + \sqrt{\alpha(\alpha-1)}\hspace{1.5pt}\xi\hspace{.5pt},\hspace{.6em} \frac{k_{d}^{2}}{2\xi_{d}^{2}} > \alpha\hspace{1pt}\max\left\{1\hspace{.5pt},\hspace{2.5pt} \frac{k}{k-\alpha\rho_{sv}\xi}\hspace{1pt},\hspace{3pt} \frac{k^{2}}{(k-\alpha\rho_{sv}\xi)^{2}-\alpha(\alpha-1)\xi^{2}}\hspace{.5pt}\right\}\hspace{-.5pt}.
\end{equation*}

\section{Proof of Proposition \ref{Prop3.4}}\label{sec:aux3.4}

For convenience, define a new stochastic process $\hspace{1.5pt}\oversymb{\hspace{-1.5pt}L}$ by
\begin{align}\label{eq3.25}
\hspace{1.5pt}\oversymb{\hspace{-1.5pt}L}_{t} \equiv \hspace{1.5pt}\oversymb{\hspace{-1.5pt}S}_{t}\exp\left\{\int_{0}^{t}{\oversymb{r}^{f}_{u}\hspace{1pt}du}\right\}
&= S_{0}\exp\bigg\{\int_{0}^{t}{\Big(\oversymb{r}^{d}_{u}-\frac{1}{2}\hspace{1pt}\oversymb{V}_{\hspace{-2.5pt}u}\Big)du} + a_{11}\int_{0}^{t}{\sqrt{\oversymb{V}_{\hspace{-2.5pt}u}}\,dW^{1}_{u}} \nonumber\\[0pt]
&+ \sum_{j=2}^{4}{a_{1\hspace{-0.5pt}j}\int_{0}^{t}{\sqrt{\oversymb{V}_{\hspace{-2.5pt}u}}\,\frac{\delta W^{j}_{u}}{\delta t}\,du}}\bigg\}.
\end{align}
As $\hspace{1.5pt}\oversymb{\hspace{-1.5pt}S}_{t}\leq\hspace{1.5pt}\oversymb{\hspace{-1.5pt}L}_{t}$, $\forall\hspace{1pt}t \in [0,T]$, it suffices to prove the finiteness of the supremum over $t$ and $\delta t$ of
\begin{align}\label{eq3.26}
\E\big[\hspace{1.5pt}\oversymb{\hspace{-1.5pt}L}_{t}^{\omega}\big]
&= S_{0}^{\omega}\E\bigg[\exp\bigg\{\omega\int_{0}^{t}{\oversymb{r}^{d}_{u}\hspace{1pt}du} - \frac{\omega}{2}\int_{0}^{t}{\oversymb{V}_{\hspace{-2.5pt}u}\hspace{1pt}du} + \omega a_{11}\int_{0}^{t}{\sqrt{\oversymb{V}_{\hspace{-2.5pt}u}}\,dW^{1}_{u}} \nonumber\\[0pt]
&+ \omega\sum_{j=2}^{4}{a_{1\hspace{-0.5pt}j}\int_{0}^{t}{\sqrt{\oversymb{V}_{\hspace{-2.5pt}u}}\,\frac{\delta W^{j}_{u}}{\delta t}\,du}}\bigg\}\bigg].
\end{align}
Conditioning the expectation on the right-hand side on $\mathcal{G}^{d,v}_{T}$ and bearing in mind that~$W^{1}\hspace{-3pt}\independent\hspace{-3pt}W^{2}$, we can split the inner expectation into two parts, which we compute using MGFs. First,
\begin{equation}\label{eq3.27}
\E\bigg[\exp\bigg\{\omega a_{11}\int_{0}^{t}{\sqrt{\oversymb{V}_{\hspace{-2.5pt}u}}\,dW_{u}^{1}}\bigg\}\Big|\,\mathcal{G}^{d,v}_{T}\bigg] = \exp\left\{\frac{\omega^{2}}{2}\hspace{1pt}a_{11}^{2}\int_{0}^{t}{\oversymb{V}_{\hspace{-2.5pt}u}\hspace{1pt}du}\right\}.
\end{equation}
Second, let $t\in[t_{n},t_{n+1})$. As $\oversymb{V}$ is piecewise constant and $W^{2}$ has independent increments,
\begin{align}\label{eq3.28}
&\E\bigg[\exp\bigg\{\omega a_{12}\int_{0}^{t}{\sqrt{\oversymb{V}_{\hspace{-2.5pt}u}}\,\frac{\delta W^{2}_{u}}{\delta t}\,du}\bigg\}\Big|\,\mathcal{G}^{d,v}_{T}\bigg] \nonumber\\[2pt]
&\qquad = \E\bigg[\exp\bigg\{\omega a_{12}\int_{0}^{t_{n}}{\hspace{-2pt}\sqrt{\oversymb{V}_{\hspace{-2.5pt}u}}\,dW_{u}^{2}} + \omega a_{12}\hspace{1pt}\frac{t-t_{n}}{\delta t}\int_{t_{n}}^{t_{n+1}}{\hspace{-2pt}\sqrt{\oversymb{V}_{\hspace{-2.5pt}u}}\,dW_{u}^{2}}\bigg\}\Big|\,\mathcal{G}^{d,v}_{T}\bigg] \nonumber\\[2pt]
&\qquad = \exp\left\{\frac{\omega^{2}}{2}\hspace{1pt}a_{12}^{2}\int_{0}^{t_{n}}{\oversymb{V}_{\hspace{-2.5pt}u}\hspace{1pt}du}\right\}\exp\left\{\frac{\omega^{2}}{2}\hspace{1pt}a_{12}^{2}\hspace{1pt}\frac{(t-t_{n})^{2}}{(\delta t)^{2}}\int_{t_{n}}^{t_{n+1}}{\oversymb{V}_{\hspace{-2.5pt}u}\hspace{1pt}du}\right\} \nonumber\\[2pt]
&\qquad \leq \exp\left\{\frac{\omega^{2}}{2}\hspace{1pt}a_{12}^{2}\int_{0}^{t}{\oversymb{V}_{\hspace{-2.5pt}u}\hspace{1pt}du}\right\}.
\end{align}
Substituting back into \eqref{eq3.26} with \eqref{eq3.27} and \eqref{eq3.28} leads to an upper bound,
\begin{align}\label{eq3.29}
\E\big[\hspace{1.5pt}\oversymb{\hspace{-1.5pt}L}_{t}^{\omega}\big] &\leq S_{0}^{\omega}\E\bigg[\exp\bigg\{\omega\int_{0}^{t}{\oversymb{r}^{d}_{u}\hspace{1pt}du} + \omega\sum_{j=3}^{4}{a_{1\hspace{-0.5pt}j}\int_{0}^{t}{\sqrt{\oversymb{V}_{\hspace{-2.5pt}u}}\,\frac{\delta W^{j}_{u}}{\delta t}\,du}} \nonumber\\[2pt]
&+ \left[\frac{\omega^{2}}{2}\big(a_{11}^{2}+a_{12}^{2}\big) - \frac{\omega}{2}\right]\int_{0}^{t}{\oversymb{V}_{\hspace{-2.5pt}u}\hspace{1pt}du}\bigg\}\bigg].
\end{align}
Next, we employ H\"older's inequality with the pair $(p,q)$, where $p,q>1$ and $q=p/(p-1)$, to force the term $\int_{0}^{t}{\oversymb{r}^{d}_{u}\hspace{1pt}du}$ outside the expectation. Then, we condition the second expectation on the $\sigma$-algebra $\mathcal{G}_{T}^{v}$ and proceed as in \eqref{eq3.28} to arrive at
\begin{align}\label{eq3.30}
\E\big[\hspace{1.5pt}\oversymb{\hspace{-1.5pt}L}_{t}^{\omega}\big] &\leq S_{0}^{\omega}\E\bigg[\exp\bigg\{p\omega\int_{0}^{t}{\oversymb{r}^{d}_{u}\hspace{1pt}du}\bigg\}\bigg]^{\frac{1}{p}}\E\bigg[\exp\bigg\{q\omega a_{14}\int_{0}^{t}{\sqrt{\oversymb{V}_{\hspace{-2.5pt}u}}\,\frac{\delta W^{4}_{u}}{\delta t}\,du} \nonumber\\[1pt]
&+ q\left[\frac{\omega^{2}}{2}\big(a_{11}^{2}+a_{12}^{2}+q a_{13}^{2}\big) - \frac{\omega}{2}\right]\int_{0}^{t}{\oversymb{V}_{\hspace{-2.5pt}u}\hspace{1pt}du}\bigg\}\bigg]^{\frac{1}{q}}.
\end{align}
All that we have left to do is to show that the supremum over $t$ and $\delta t$ of each of the two expectations on the right-hand side of \eqref{eq3.30} is finite. However, one can easily deduce from Lemma \ref{Lem3.1} the following sufficient conditions:
\begin{equation}\label{eq3.31}
k_{d}\geq\frac{1}{2}\hspace{1pt}p\omega\hspace{1pt}T\xi_{d}^{2}\hspace{1pt},\hspace{5pt} k\geq q\omega\rho_{sv}\xi+\frac{1}{2}\hspace{1pt}\Delta T\xi^{2},
\end{equation}
where
\begin{equation}\label{eq3.32}
\Delta = q\left[\frac{1}{2}\hspace{1pt}\omega(\omega - 1) + \frac{1}{2}\hspace{1pt}\omega^{2}(q-1)\big(a_{13}^{2}+a_{14}^{2}\big)\right].
\end{equation}
Note that we used $\sum_{j=1}^{4}{a_{1\hspace{-0.5pt}j}^{2}}=1$ in \eqref{eq3.32}. On the other hand, the first assumption in \eqref{eq3.23} ensures that $q_{2}(\alpha)>1$. This, in turn, implies that $q_{2}(\alpha)/(q_{2}(\alpha)-1)>1$. Due to the second assumption in \eqref{eq3.23}, we can find $p>1$ so that
\begin{equation}\label{eq3.33}
\frac{2k_{d}}{T\xi_{d}^{2}} > \alpha p > \frac{\alpha\hspace{1.5pt}q_{2}(\alpha)}{q_{2}(\alpha)-1} \hspace{3pt}\Rightarrow\hspace{3pt} k_{d}>\frac{1}{2}\hspace{1pt}p\alpha\hspace{1pt}T\xi_{d}^{2}\hspace{1pt}.
\end{equation}
The quadratic equation in $x$ below has roots of different signs, with positive root $q_{2}(\alpha)$:
\begin{equation*}
\frac{1}{4}\hspace{1pt}T\alpha^{2}\xi^{2}\big(a_{13}^{2}+a_{14}^{2}\big)x^{2} + \left[\alpha\rho_{sv}\xi + \frac{1}{4}\hspace{1pt}T\alpha^{2}\xi^{2}\big(a_{11}^{2}+a_{12}^{2}\big)-\frac{1}{4}\hspace{1pt}T\alpha\xi^{2}\right]x - k = 0.
\end{equation*}
However, the H\"older pair satisfies $p=q/(q-1)$, so $q<q_{2}(\alpha)$. Hence, $q$ lies in between the two roots of the quadratic, which implies that
\begin{align}\label{eq3.34}
\frac{1}{4}\hspace{1pt}T\alpha^{2}\xi^{2}\big(a_{13}^{2}+a_{14}^{2}\big)q^{2} + \left[\alpha\rho_{sv}\xi + \frac{1}{4}\hspace{1pt}T\alpha^{2}\xi^{2}\big(a_{11}^{2}+a_{12}^{2}\big)-\frac{1}{4}\hspace{1pt}T\alpha\xi^{2}\right]q - k < 0.
\end{align}
Rearranging terms in the above inequality, we obtain
\begin{align}\label{eq3.35}
k > q\alpha\rho_{sv}\xi + \frac{1}{2}\hspace{1pt}T\xi^{2}q\left[\frac{1}{2}\hspace{1pt}\alpha(\alpha - 1) + \frac{1}{2}\hspace{1pt}\alpha^{2}(q-1)\big(a_{13}^{2}+a_{14}^{2}\big)\right].
\end{align}
From \eqref{eq3.33} and \eqref{eq3.35}, employing a continuity argument similar to that used in the proof of Proposition \ref{Prop3.2}, we deduce that the two conditions in \eqref{eq3.31} hold on an interval $[\alpha,\alpha_{2})$, for some $\alpha_{2}>\alpha$, which concludes the proof. The extension to the interval $[1,\alpha_{2})$ follows from Jensen's inequality, with $\eta_{\omega}=\eta_{\alpha},\hspace{1pt}\forall\hspace{.5pt}\omega\in[1,\alpha]$. \qed

\vspace{1em}
In the event that $a_{13}$ and $a_{14}$ are simultaneously zero, i.e., $\rho_{sv}=\rho_{sd}=0$, one can easily show that Proposition \ref{Prop3.4} still holds as long as
\begin{equation*}
k > \frac{1}{4}\hspace{1pt}\alpha(\alpha-1)T\xi^{2},\hspace{.6em} \frac{k_{d}}{T\xi_{d}^{2}} > \frac{2\alpha k}{4k-\alpha(\alpha-1)T\xi^{2}}\hspace{1.5pt}.
\end{equation*}

\section{Proof of Proposition \ref{Prop3.6}}\label{sec:aux3.6}

We follow the argument of Proposition \ref{Prop3.2} closely and condition on the $\sigma$-algebra $\mathcal{G}^{v}_{t}$ instead to deduce that
\begin{align}\label{eqD.1}
\E\big[R_{t}^{\omega}\big] &\leq S_{0}^{\omega}\E\bigg[\exp\bigg\{\left[\frac{1}{2}\hspace{1pt}\omega^{2}\big(1-\rho_{sv}^{2}\big) - \frac{1}{2}\hspace{1pt}\omega\right]\int_{0}^{t}{v_{u}\hspace{1pt}du} + \omega \rho_{sv}\int_{0}^{t}{\sqrt{v_{u}}\hspace{1.5pt}dW_{u}^{4}}\bigg\}\bigg].
\end{align}
First of all, suppose that $\alpha=1$ and $T\geq0$. If $k<\rho_{sv}\xi$, then
\begin{equation*}
\lim_{\omega\hspace{1pt}\downarrow\hspace{1pt}1^{+}}\hspace{1pt}\frac{1}{\nu(\omega)}\log\left(\frac{\omega\rho_{sv}\xi-k+\nu(\omega)}{\omega\rho_{sv}\xi-k-\nu(\omega)}\right) = \infty.
\end{equation*}
Hence, by a continuity argument, we can find $\alpha_{1}>1$ such that for all $\omega\in(1,\alpha_{1})$,
\begin{equation}\label{eqD.2}
k<\omega\rho_{sv}\xi-\sqrt{\omega(\omega-1)}\hspace{1.5pt}\xi \hspace{5pt}\text{ and }\hspace{5pt}
T<\frac{1}{\nu(\omega)}\log\left(\frac{\omega\rho_{sv}\xi-k+\nu(\omega)}{\omega\rho_{sv}\xi-k-\nu(\omega)}\right).
\end{equation}
If $k=\rho_{sv}\xi$, then $\rho_{sv}\in(0,1]$ and
\begin{equation*}
\lim_{\omega\hspace{1pt}\downarrow\hspace{1pt}1^{+}}\hspace{1pt}\frac{2}{\hat{\nu}(\omega)}\left[\frac{\pi}{2}-\arctan\left(\frac{\omega\rho_{sv}\xi-k}{\hat{\nu}(\omega)}\right)\right] =
\lim_{\omega\hspace{1pt}\downarrow\hspace{1pt}1^{+}}\hspace{1pt}\frac{2}{\hat{\nu}(\omega)}\hspace{1pt}\arctan\left(\frac{\sqrt{\omega-(\omega-1)\rho_{sv}^{2}}}{\sqrt{\omega-1}\hspace{1pt}\rho_{sv}}\right) = \infty.
\end{equation*}
Furthermore, note that for all $\omega>1$,
\begin{equation*}
\omega\rho_{sv}\xi-\sqrt{\omega(\omega-1)}\hspace{1.5pt}\xi \leq \omega k-\sqrt{\omega(\omega-1)}\hspace{1.5pt}k < k.
\end{equation*}
Hence, we can find $\alpha_{1}>1$ such that for all $\omega\in(1,\alpha_{1})$,
\begin{equation}\label{eqD.3}
\omega\rho_{sv}\xi-\sqrt{\omega(\omega-1)}\hspace{1.5pt}\xi<k<\omega\rho_{sv}\xi+\sqrt{\omega(\omega-1)}\hspace{1.5pt}\xi
\end{equation}
and
\begin{equation}\label{eqD.4}
T<\frac{2}{\hat{\nu}(\omega)}\left[\frac{\pi}{2}-\arctan\left(\frac{\omega\rho_{sv}\xi-k}{\hat{\nu}(\omega)}\right)\right].
\end{equation}
If $k>\rho_{sv}\xi$, then we can find $\alpha_{1}>1$ such that for all $\omega\in(1,\alpha_{1})$,
\begin{equation}\label{eqD.5}
k>\omega\rho_{sv}\xi+\sqrt{\omega(\omega-1)}\hspace{1.5pt}\xi.
\end{equation}
The conclusion follows from Proposition 3.2 in \citet{Cozma:2015.2} and \eqref{eqD.2} -- \eqref{eqD.5}. Next, suppose that $\alpha>1$ and $T<T^{*}$, with $T^{*}$ defined in \eqref{eq3.38.1} -- \eqref{eq3.38.4}.

If $k<\alpha\rho_{sv}\xi-\sqrt{\alpha(\alpha-1)}\hspace{1.5pt}\xi$, by a continuity argument, we can find $\alpha_{1}>\alpha$ such that for all $\omega\in(\alpha,\alpha_{1})$,
\begin{equation}\label{eqD.6}
k<\omega\rho_{sv}\xi-\sqrt{\omega(\omega-1)}\hspace{1.5pt}\xi \hspace{5pt}\text{ and }\hspace{5pt}
T<\frac{1}{\nu(\omega)}\log\left(\frac{\omega\rho_{sv}\xi-k+\nu(\omega)}{\omega\rho_{sv}\xi-k-\nu(\omega)}\right).
\end{equation}
If $k=\alpha\rho_{sv}\xi-\sqrt{\alpha(\alpha-1)}\hspace{1.5pt}\xi$, then $\rho_{sv}\in(0,1]$ and for all $\omega>\alpha$,
\begin{equation*}
\omega-\alpha < \sqrt{\omega(\omega-1)} - \sqrt{\alpha(\alpha-1)} \hspace{3pt}\Rightarrow \hspace{3pt} \omega\rho_{sv}\xi-\sqrt{\omega(\omega-1)}\hspace{1.5pt}\xi < \alpha\rho_{sv}\xi-\sqrt{\alpha(\alpha-1)}\hspace{1.5pt}\xi.
\end{equation*}
Furthermore, note that
\begin{equation*}
\lim_{\omega\hspace{1pt}\downarrow\hspace{1pt}\alpha^{+}}\hspace{1pt}\frac{2}{\hat{\nu}(\omega)}\left[\frac{\pi}{2}-\arctan\left(\frac{\omega\rho_{sv}\xi-k}{\hat{\nu}(\omega)}\right)\right] =
\lim_{\omega\hspace{1pt}\downarrow\hspace{1pt}\alpha^{+}}\hspace{1pt}\frac{2}{\hat{\nu}(\omega)}\hspace{1pt}\arctan\left(\frac{\hat{\nu}(\omega)}{\alpha\rho_{sv}\xi-k}\right) = \frac{2}{\alpha\rho_{sv}\xi-k}\hspace{1pt}.
\end{equation*}
Hence, we can find $\alpha_{1}>\alpha$ such that for all $\omega\in(\alpha,\alpha_{1})$,
\begin{equation}\label{eqD.7}
\omega\rho_{sv}\xi-\sqrt{\omega(\omega-1)}\hspace{1.5pt}\xi<k<\omega\rho_{sv}\xi+\sqrt{\omega(\omega-1)}\hspace{1.5pt}\xi
\end{equation}
and
\begin{equation}\label{eqD.8}
T<\frac{2}{\hat{\nu}(\omega)}\left[\frac{\pi}{2}-\arctan\left(\frac{\omega\rho_{sv}\xi-k}{\hat{\nu}(\omega)}\right)\right].
\end{equation}
If $\alpha\rho_{sv}\xi-\sqrt{\alpha(\alpha-1)}\hspace{1.5pt}\xi<k<\alpha\rho_{sv}\xi+\sqrt{\alpha(\alpha-1)}\hspace{1.5pt}\xi$, we can clearly find $\alpha_{1}>\alpha$ so that both \eqref{eqD.7} and \eqref{eqD.8} hold for all $\omega\in(\alpha,\alpha_{1})$. If $k=\alpha\rho_{sv}\xi+\sqrt{\alpha(\alpha-1)}\hspace{1.5pt}\xi$, then for all $\omega>\alpha$,
\begin{equation*}
(\alpha-\omega)\rho_{sv} < \sqrt{\omega(\omega-1)} - \sqrt{\alpha(\alpha-1)} \hspace{3pt}\Rightarrow \hspace{3pt} \alpha\rho_{sv}\xi+\sqrt{\alpha(\alpha-1)}\hspace{1.5pt}\xi < \omega\rho_{sv}\xi+\sqrt{\omega(\omega-1)}\hspace{1.5pt}\xi.
\end{equation*}
Furthermore, note that
\begin{equation*}
\lim_{\omega\hspace{1pt}\downarrow\hspace{1pt}\alpha^{+}}\hspace{1pt}\frac{2}{\hat{\nu}(\omega)}\left[\frac{\pi}{2}-\arctan\left(\frac{\omega\rho_{sv}\xi-k}{\hat{\nu}(\omega)}\right)\right] = \lim_{\omega\hspace{1pt}\downarrow\hspace{1pt}\alpha^{+}}\hspace{1pt}\frac{2\pi}{\hat{\nu}(\omega)} = \infty.
\end{equation*}
Hence, we can find $\alpha_{1}>\alpha$ such that both \eqref{eqD.7} and \eqref{eqD.8} hold for all $\omega\in(\alpha,\alpha_{1})$. Finally, if $k>\alpha\rho_{sv}\xi+\sqrt{\alpha(\alpha-1)}\hspace{1.5pt}\xi$, then we can find $\alpha_{1}>\alpha$ such that for all $\omega\in(\alpha,\alpha_{1})$,
\begin{equation}\label{eqD.9}
k>\omega\rho_{sv}\xi+\sqrt{\omega(\omega-1)}\hspace{1.5pt}\xi.
\end{equation}
The conclusion follows from Proposition 3.2 in \citet{Cozma:2015.2} and \eqref{eqD.6} -- \eqref{eqD.9}. The extension to the interval $[1,\alpha_{1})$ follows from Jensen's inequality. \qed

\section{Proof of Proposition \ref{Prop3.7}}\label{sec:aux3.7}

We follow the argument of Proposition \ref{Prop3.4} closely and condition on the $\sigma$-algebra $\mathcal{G}^{v}_{T}$ instead to deduce that
\begin{align}\label{eqE.1}
\E\big[\hspace{1.5pt}\oversymb{\hspace{-1.5pt}R}_{\hspace{-.5pt}t}^{\omega}\big] \leq S_{0}^{\omega}\E\bigg[\exp\bigg\{\left[\frac{1}{2}\hspace{1pt}\omega^{2}\big(1-\rho_{sv}^{2}\big) - \frac{1}{2}\hspace{1pt}\omega\right]\int_{0}^{t}{\oversymb{V}_{\hspace{-2.5pt}u}\hspace{1pt}du} + \omega\rho_{sv}\int_{0}^{t}{\sqrt{\oversymb{V}_{\hspace{-2.5pt}u}}\,\frac{\delta W^{4}_{u}}{\delta t}\,du}\bigg\}\bigg].
\end{align}
Suppose that $T<T^{*}$, with $T^{*}$ from \eqref{eq3.40.1} -- \eqref{eq3.40.2}. If $k<\alpha\rho_{sv}\xi+\frac{1}{2}\hspace{1pt}\sqrt{\alpha(\alpha-1)}\hspace{1.5pt}\xi$, then by a continuity argument, we can find $\alpha_{2}>\alpha$ such that for all $\omega\in(\alpha,\alpha_{2})$,
\begin{equation}\label{eqE.2}
k<\omega\rho_{sv}\xi+\frac{1}{2}\hspace{1pt}\sqrt{\omega(\omega-1)}\hspace{1.5pt}\xi \hspace{5pt}\text{ and }\hspace{5pt}
T<\frac{1}{\omega\rho_{sv}\xi+\sqrt{\omega(\omega-1)}\hspace{1.5pt}\xi-k}\hspace{1pt}.
\end{equation}
If $k=\alpha\rho_{sv}\xi+\frac{1}{2}\hspace{1pt}\sqrt{\alpha(\alpha-1)}\hspace{1.5pt}\xi$, since $k$, $\xi>0$ and for all $\omega>\alpha$, we have
\begin{equation*}
\rho_{sv}>-\hspace{1pt}\frac{1}{2}\hspace{1pt}\sqrt{1-\frac{1}{\alpha}}>-\hspace{1pt}\frac{1}{2}\hspace{1pt} \hspace{3pt}\Rightarrow \hspace{3pt} \alpha\rho_{sv}\xi+\frac{1}{2}\hspace{1pt}\sqrt{\alpha(\alpha-1)}\hspace{1.5pt}\xi < \omega\rho_{sv}\xi+\frac{1}{2}\hspace{1pt}\sqrt{\omega(\omega-1)}\hspace{1.5pt}\xi.
\end{equation*}
Furthermore, note that
\begin{equation*}
\frac{4(k-\alpha\rho_{sv}\xi)}{\alpha(\alpha-1)\xi^{2}} = \frac{1}{\alpha\rho_{sv}\xi+\sqrt{\alpha(\alpha-1)}\hspace{1.5pt}\xi-k} \hspace{5pt}\text{ and }\hspace{5pt}
\lim_{\omega\hspace{1pt}\downarrow\hspace{1pt}\alpha^{+}}\hspace{1pt}\frac{1}{\omega\rho_{sv}\xi+\sqrt{\omega(\omega-1)}\hspace{1.5pt}\xi-k} = T^{*},
\end{equation*}
with $T^{*}$ from \eqref{eq3.40.2}. Hence, we can find $\alpha_{2}>\alpha$ such that \eqref{eqE.2} holds for all $\omega\in(\alpha,\alpha_{2})$. Finally, if $k>\alpha\rho_{sv}\xi+\frac{1}{2}\hspace{1pt}\sqrt{\alpha(\alpha-1)}\hspace{1.5pt}\xi$, since
\begin{equation*}
\lim_{\omega\hspace{1pt}\downarrow\hspace{1pt}\alpha^{+}}\hspace{1pt}\frac{4(k-\omega\rho_{sv}\xi)}{\omega(\omega-1)\xi^{2}} = T^{*},
\end{equation*}
with $T^{*}$ from \eqref{eq3.40.2}, we can find $\alpha_{2}>\alpha$ such that for all $\omega\in(\alpha,\alpha_{2})$,
\begin{equation}\label{eqE.3}
k>\omega\rho_{sv}\xi+\frac{1}{2}\hspace{1pt}\sqrt{\omega(\omega-1)}\hspace{1.5pt}\xi \hspace{5pt}\text{ and }\hspace{5pt}
T<\frac{4(k-\omega\rho_{sv}\xi)}{\omega(\omega-1)\xi^{2}}\hspace{1pt}.
\end{equation}
The conclusion follows from Lemma \ref{Lem3.1} and \eqref{eqE.2} -- \eqref{eqE.3}. The extension to the interval $[1,\alpha_{2})$ follows from Jensen's inequality, with $\eta_{\omega}=\eta_{\alpha},\hspace{1pt}\forall\hspace{.5pt}\omega\in[1,\alpha]$. \qed

\section{Proof of Proposition \ref{Prop3.8.0}}\label{sec:aux3.8.0}

The following auxiliary result proves the almost sure positivity of the foreign interest rate.

\begin{lemma}\label{LemF.1}
Let $\kappa>(r_{0}^{f})^{-1}$ and define the stopping time
\begin{equation}\label{eqF.1}
\tau_{\kappa} = \inf\big\{t\geq 0 :\hspace{2pt} r^{f}_{t}\leq\kappa^{-1} \big\}.
\end{equation}
If $2k_{f}\theta_{f}>\xi_{f}^{2}\hspace{.5pt}$, then
\begin{equation}\label{eqF.2}
\plim_{\kappa\to\infty}\Prob\big(\tau_{\kappa}\leq T\big) = 0.
\end{equation}
\end{lemma}
\begin{proof}
Define the function $U:(0,\infty)\mapsto\RR$ by
\begin{equation}\label{eqF.3}
U(x) = x^{-\alpha},\hspace{1em} \alpha = \frac{1}{2\xi_{f}^{2}}\big(2k_{f}\theta_{f}-\xi_{f}^{2}\big).
\end{equation}
By It\^o's formula, we have
\begin{align}\label{eqF.4}
\E\Big[U\big(r_{T\wedge\hspace{1pt}\tau_{\kappa}}^{f}\big)\Big] &= U\big(r_{0}^{f}\big) -
\E\int_{0}^{T\wedge\hspace{1pt}\tau_{\kappa}}{\hspace{-1pt}\alpha\big(r_{s}^{f}\big)^{-(1+\alpha)}\big(k_{f}\theta_{f}-k_{f}r^{f}_{s}-\rho_{s\hspace{-.7pt}f}\xi_{f}\sqrt{v_{s}r^{f}_{s}}\hspace{1pt}\big)ds} \nonumber\\[2pt]
&\hspace{-1em}+ \frac{1}{2}\hspace{.5pt}\E\int_{0}^{T\wedge\hspace{1pt}\tau_{\kappa}}{\hspace{-1pt}\alpha(1+\alpha)\xi_{f}^{2}\big(r_{s}^{f}\big)^{-(1+\alpha)}ds}
- \E\int_{0}^{T\wedge\hspace{1pt}\tau_{\kappa}}{\hspace{-1pt}\alpha\xi_{f}\big(r_{s}^{f}\big)^{-(0.5+\alpha)}dW^{f}_{s}}.
\end{align}
However,
\begin{equation*}
\E\int_{0}^{T}{\alpha^{2}\xi_{f}^{2}\big(r_{s}^{f}\big)^{-(1+2\alpha)}\Ind_{s<\hspace{.5pt}\tau_{\kappa}}ds} \leq
\alpha^{2}\xi_{f}^{2}\kappa^{1+2\alpha}\hspace{.5pt}T < \infty,
\end{equation*}
so the stochastic integral on the right-hand side of \eqref{eqF.4} is a true martingale. Hence,
\begin{equation}\label{eqF.5}
\E\Big[U\big(r_{T\wedge\hspace{1pt}\tau_{\kappa}}^{f}\big)\Big] \leq U\big(r_{0}^{f}\big) -
\E\int_{0}^{T\wedge\hspace{1pt}\tau_{\kappa}}{\hspace{-2pt}\Big(a\big(r_{s}^{f}\big)^{-(1+\alpha)}-b\big(r_{s}^{f}\big)^{-\alpha}-c\hspace{.5pt}v_{s}^{0.5}\big(r_{s}^{f}\big)^{-(0.5+\alpha)}\Big)ds}\hspace{.5pt},
\end{equation}
where
\begin{equation}\label{eqF.6}
a = \frac{1}{8\xi_{f}^{2}}\big(2k_{f}\theta_{f}-\xi_{f}^{2}\big)^{2},\hspace{3pt}
b = \frac{k_{f}}{2\xi_{f}^{2}}\big(2k_{f}\theta_{f}-\xi_{f}^{2}\big),\hspace{3pt}
c = \frac{|\rho_{s\hspace{-.7pt}f}|}{2\xi_{f}}\big(2k_{f}\theta_{f}-\xi_{f}^{2}\big).
\end{equation}
Employing Fubini's theorem and H\"older's inequality in \eqref{eqF.5}, we get
\begin{align}\label{eqF.7}
\E\Big[U\big(r_{T\wedge\hspace{1pt}\tau_{\kappa}}^{f}\big)\Big] &\leq U\big(r_{0}^{f}\big) -
\int_{0}^{T}\bigg(a\E\Big[\big(r_{s}^{f}\big)^{-(1+\alpha)}\Ind_{s<\hspace{.5pt}\tau_{\kappa}}\Big]-
b\E\Big[\big(r_{s}^{f}\big)^{-(1+\alpha)}\Ind_{s<\hspace{.5pt}\tau_{\kappa}}\Big]^{\frac{\alpha}{1+\alpha}} \nonumber\\[3pt]
&- c\hspace{.5pt}\sup_{u\in[0,T]}\E\Big[v_{u}^{1+\alpha}\Big]^{\frac{1}{2(1+\alpha)}}\E\Big[\big(r_{s}^{f}\big)^{-(1+\alpha)}\Ind_{s<\hspace{.5pt}\tau_{\kappa}}\Big]^{\frac{1+2\alpha}{2(1+\alpha)}}\bigg)ds\hspace{.5pt}.
\end{align}
The moments of the square root process are uniformly bounded \citep{Dereich:2012} and the function $f:[0,\infty)\mapsto\RR$ defined by
\begin{equation}\label{eqF.8}
f(x) = ax - bx^{\frac{\alpha}{1+\alpha}} - c\hspace{.5pt}\sup_{u\in[0,T]}\E\big[v_{u}^{1+\alpha}\big]^{\frac{1}{2(1+\alpha)}}x^{\frac{1+2\alpha}{2(1+\alpha)}}
\end{equation}
is clearly bounded from below. Hence, we can find a constant $C$ independent of $\kappa$ such that
\begin{equation}\label{eqF.9}
\E\Big[U\big(r_{T\wedge\hspace{1pt}\tau_{\kappa}}^{f}\big)\Big] \leq C.
\end{equation}
Since $r^{f}$ has continuous paths, we have $r_{\tau_{\kappa}}^{f}=\kappa^{-1}$ and $U(r_{\tau_{\kappa}}^{f})=\kappa^{\alpha}$. Therefore, using \eqref{eqF.9} and the fact that $U$ is positive, we deduce that
\begin{equation}\label{eqF.10}
\kappa^{\alpha}\Prob\big(\tau_{\kappa}\leq T\big) = \E\Big[U\big(r_{\tau_{\kappa}}^{f}\big)\Ind_{\tau_{\kappa}\leq\hspace{1pt}T}\Big] = \E\Big[U\big(r_{T\wedge\hspace{1pt}\tau_{\kappa}}^{f}\big)\Ind_{\tau_{\kappa}\leq\hspace{1pt}T}\Big] \leq
\E\Big[U\big(r_{T\wedge\hspace{1pt}\tau_{\kappa}}^{f}\big)\Big] \leq C.
\end{equation}
Letting $\kappa\to\infty$ in \eqref{eqF.10} yields the conclusion.
\end{proof}

The next two lemmas give moment bounds for the original and the discretized foreign interest rate processes.

\begin{lemma}\label{LemF.2}
The process $r^{f}$ has uniformly bounded moments, i.e.,
\begin{equation}\label{eqF.11}
\E\bigg[\sup_{t\in[0,T]}\big(r^{f}_{t}\big)^{p}\bigg] < \infty,\hspace{5pt} \forall\hspace{.5pt}p\geq1.
\end{equation}
\end{lemma}
\begin{proof}
Fix any $p\geq1$. From \eqref{eq2.1},
\begin{equation}\label{eqF.12}
r^{f}_{t} = r^{f}_{0} + k_{f}\theta_{f}t - k_{f}\int_{0}^{t}{r^{f}_{u}du} - \rho_{s\hspace{-.7pt}f}\xi_{f}\int_{0}^{t}{\sqrt{v_{u}r^{f}_{u}}\hspace{1pt}du} + \xi_{f}\int_{0}^{t}{\sqrt{r^{f}_{u}}\hspace{1pt}dW^{f}_{u}}.
\end{equation}
Using the fact that $2\sqrt{|ab|}\leq |a|+|b|$ and H\"older's inequality, we deduce that
\begin{align}\label{eqF.13}
\big(r^{f}_{t}\big)^{p} &\leq 2^{2(p-1)}\big(r^{f}_{0} + k_{f}\theta_{f}t\big)^{p} + 2^{p-2}|\rho_{s\hspace{-.7pt}f}|^{p}\xi_{f}^{p}\bigg(\int_{0}^{t}{v_{u}\hspace{.5pt}du}\bigg)^{p} \nonumber\\[2pt]
&+ 2^{p-2}\big(2k_{f}+|\rho_{s\hspace{-.7pt}f}|\xi_{f}\big)^{p}\bigg(\int_{0}^{t}{r^{f}_{u}du}\bigg)^{p} + 2^{2(p-1)}\xi_{f}^{p}\bigg|\int_{0}^{t}{\sqrt{r^{f}_{u}}\hspace{1pt}dW^{f}_{u}}\bigg|^{p}.
\end{align}
Fix $t\in[0,T]$. Using H\"older's inequality, we get
\begin{align}\label{eqF.14}
\sup_{s\in[0,t]}\big(r^{f}_{s}\big)^{p} &\leq 2^{2(p-1)}\big(r^{f}_{0} + k_{f}\theta_{f}T\big)^{p} + 2^{p-2}|\rho_{s\hspace{-.7pt}f}|^{p}\xi_{f}^{p}T^{p-1}\int_{0}^{T}{v_{u}^{p}\hspace{.5pt}du} \nonumber\\[1pt]
&\hspace{-1.5em}+ 2^{p-2}\big(2k_{f}+|\rho_{s\hspace{-.7pt}f}|\xi_{f}\big)^{p}T^{p-1}\int_{0}^{t}{\big(r^{f}_{u}\big)^{p}du} + 2^{2(p-1)}\xi_{f}^{p}\sup_{s\in[0,t]}\bigg|\int_{0}^{s}{\sqrt{r^{f}_{u}}\hspace{1pt}dW^{f}_{u}}\bigg|^{p}.
\end{align}
From the Burkholder-Davis-Gundy inequality, we know that there exists a constant $C_{p}>0$ such that
\begin{equation*}
\E\bigg[\sup_{s\in[0,t]}\bigg|\int_{0}^{s}{\sqrt{r^{f}_{u}}\hspace{1pt}dW^{f}_{u}}\bigg|^{p}\hspace{1pt}\bigg] \leq C_{p}\E\bigg[\bigg(\int_{0}^{t}{r^{f}_{u}du}\bigg)^{p/2}\hspace{1pt}\bigg] \leq \frac{1}{2}\hspace{1pt}C_{p} + \frac{1}{2}\hspace{1pt}C_{p}T^{p-1}\E\bigg[\int_{0}^{t}{\big(r^{f}_{u}\big)^{p}du}\bigg].
\end{equation*}
Taking expectations and employing Fubini's theorem in \eqref{eqF.14},
\begin{align*}
\E\bigg[\sup_{s\in[0,t]}\big(r^{f}_{s}\big)^{p}\bigg] &\leq 2^{2(p-1)}\big(r^{f}_{0} + k_{f}\theta_{f}T\big)^{p} + 2^{2p-3}\xi_{f}^{p}C_{p} + 2^{p-2}|\rho_{s\hspace{-.7pt}f}|^{p}\xi_{f}^{p}T^{p}\sup_{u\in[0,T]}\E\big[v_{u}^{p}\big] \\[0pt]
&\hspace{-1.5em}+ \Big(2^{p-2}\big(2k_{f}+|\rho_{s\hspace{-.7pt}f}|\xi_{f}\big)^{p}T^{p-1} + 2^{2p-3}\xi_{f}^{p}C_{p}T^{p-1}\Big)\int_{0}^{t}{\E\bigg[\sup_{s\in[0,u]}\big(r^{f}_{s}\big)^{p}\bigg]du}.
\end{align*}
Applying Gronwall's inequality, we get
\begin{align}\label{eqF.15}
\E\bigg[\sup_{t\in[0,T]}\big(r^{f}_{t}\big)^{p}\bigg] &\leq \Big(2^{2(p-1)}\big(r^{f}_{0} + k_{f}\theta_{f}T\big)^{p} + 2^{2p-3}\xi_{f}^{p}C_{p} + 2^{p-2}|\rho_{s\hspace{-.7pt}f}|^{p}\xi_{f}^{p}T^{p}\sup_{u\in[0,T]}\E\big[v_{u}^{p}\big]\Big) \nonumber\\[0pt]
&\times \exp\left\{2^{p-2}\big(2k_{f}+|\rho_{s\hspace{-.7pt}f}|\xi_{f}\big)^{p}T^{p} + 2^{2p-3}\xi_{f}^{p}C_{p}T^{p}\right\}.
\end{align}
The conclusion follows from the boundedness of moments of $v$.
\end{proof}

\begin{lemma}\label{LemF.3}
The process $\hat{r}^{f}$ from \eqref{eq2.7.2} has uniformly bounded moments, i.e.,
\begin{equation}\label{eqF.16}
\sup_{\delta t\in(0,\eta)}\E\bigg[\sup_{t\in[0,T]}\big(\hat{r}^{f}_{t}\big)^p\bigg] < \infty,\hspace{5pt} \forall\hspace{.5pt}p\geq1,\hspace{5pt} \forall\hspace{.5pt}\eta>0.
\end{equation}
\end{lemma}
\begin{proof}
Fix any $p\geq1$ and $\eta>0$. From \eqref{eq2.7.1},
\begin{equation}\label{eqF.17}
\tilde{r}^{f}_{t} = r^{f}_{0} + k_{f}\theta_{f}t - k_{f}\int_{0}^{t}{\oversymb{r}^{f}_{u}du} - \rho_{s\hspace{-.7pt}f}\xi_{f}\int_{0}^{t}{\sqrt{\oversymb{V}_{\hspace{-2.5pt}u}\oversymb{r}^{f}_{u}}\hspace{1pt}du} + \xi_{f}\int_{0}^{t}{\sqrt{\oversymb{r}^{f}_{u}}\hspace{1pt}dW^{f}_{u}}.
\end{equation}
Since $\hat{r}^{f}_{t}\leq|\tilde{r}^{f}_{t}|$, following the argument of Lemma \ref{LemF.2}, we deduce that
\begin{align}\label{eqF.18}
\E\bigg[\sup_{s\in[0,t]}\big(\hat{r}^{f}_{s}\big)^{p}\bigg] &\leq 2^{2(p-1)}\big(r^{f}_{0} + k_{f}\theta_{f}T\big)^{p} + 2^{2p-3}\xi_{f}^{p}C_{p} + 2^{p-2}|\rho_{s\hspace{-.7pt}f}|^{p}\xi_{f}^{p}T^{p}\sup_{u\in[0,T]}\E\big[\oversymb{V}_{\hspace{-2.5pt}u}^{p}\big] \nonumber\\[0pt]
&\hspace{-1em}+ \Big(2^{p-2}\big(2k_{f}+|\rho_{s\hspace{-.7pt}f}|\xi_{f}\big)^{p}T^{p-1} + 2^{2p-3}\xi_{f}^{p}C_{p}T^{p-1}\Big)\int_{0}^{t}{\E\Big[\big(\oversymb{r}^{f}_{u}\big)^{p}\Big]du}.
\end{align}
Since $\sup_{u\in[0,T]}\E\big[\oversymb{V}_{\hspace{-2.5pt}u}^{p}\big]\leq\sup_{u\in[0,T]}\E\big[V_{u}^{p}\big]$, with $V$ defined as in \eqref{eq2.6}, and $\oversymb{r}^{f}_{u}\leq\sup_{s\in[0,u]}\hat{r}^{f}_{s}$, applying Gronwall's inequality, we get
\begin{align}\label{eqF.19}
\E\bigg[\sup_{t\in[0,T]}\big(\hat{r}^{f}_{t}\big)^{p}\bigg] &\leq \Big(2^{2(p-1)}\big(r^{f}_{0} + k_{f}\theta_{f}T\big)^{p} + 2^{2p-3}\xi_{f}^{p}C_{p} + 2^{p-2}|\rho_{s\hspace{-.7pt}f}|^{p}\xi_{f}^{p}T^{p}\sup_{u\in[0,T]}\E\big[V_{u}^{p}\big]\Big) \nonumber\\[0pt]
&\times \exp\left\{2^{p-2}\big(2k_{f}+|\rho_{s\hspace{-.7pt}f}|\xi_{f}\big)^{p}T^{p} + 2^{2p-3}\xi_{f}^{p}C_{p}T^{p}\right\}.
\end{align}
The conclusion follows from Proposition 3.4 in \citet{Cozma:2015}.
\end{proof}

Next, we use Lemma \ref{LemF.3} to prove the convergence of the $L^{2}$ difference between the two time continuous discretizations.

\begin{lemma}\label{LemF.4}
The $L^{2}$ difference between $\hat{r}^{f}$ and $\oversymb{r}^{f}$ converges to zero with $\delta t$, i.e.,
\begin{equation}\label{eqF.20}
\plim_{\delta t\to0}\hspace{1pt}\sup_{t\in[0,T]}\E\Big[\big(\hat{r}^{f}_{t}-\oversymb{r}^{f}_{t}\big)^{2}\Big] = 0.
\end{equation}
\end{lemma}
\begin{proof}
Suppose that $t\in[t_{n},t_{n+1})$. Since $|\hat{r}^{f}_{t}-\oversymb{r}^{f}_{t}|\leq|\tilde{r}^{f}_{t}-\tilde{r}^{f}_{t_{n}}|$ and from \eqref{eq2.7.1}, we can bound the squared absolute difference from above as follows:
\begin{align*}
\big(\hat{r}^{f}_{t}-\oversymb{r}^{f}_{t}\big)^{2} &\leq
\Big(k_{f}\theta_{f}\delta t + 0.5|\rho_{s\hspace{-.7pt}f}|\xi_{f}\delta t\hspace{1pt}V_{t_{n}} + \big(k_{f}+ 0.5|\rho_{s\hspace{-.7pt}f}|\xi_{f}\big)\delta t\hspace{1pt}\hat{r}_{t_{n}}^{f} +\xi_{f}\sqrt{\hat{r}_{t_{n}}^{f}}\hspace{1pt}\big|W^{f}_{t}-W^{f}_{t_{n}}\big|\Big)^{\hspace{-1pt}2} \\[3pt]
&\hspace{-1.55em}\leq 4k_{f}^{2}\theta_{f}^{2}(\delta t)^{2} + |\rho_{s\hspace{-.7pt}f}|^{2}\xi_{f}^{2}(\delta t)^{2}V_{t_{n}}^{2} + \big(2k_{f}+ |\rho_{s\hspace{-.7pt}f}|\xi_{f}\big)^{2}(\delta t)^{2}\big(\hat{r}_{t_{n}}^{f}\big)^{2} + 4\xi_{f}^{2}\hat{r}_{t_{n}}^{f}\big(W^{f}_{t}-W^{f}_{t_{n}}\big)^{2}.
\end{align*}
Therefore,
\begin{align}\label{eqF.21}
\sup_{t\in[0,T]}\E\Big[\big(\hat{r}^{f}_{t}-\oversymb{r}^{f}_{t}\big)^{2}\Big] &\leq
4k_{f}^{2}\theta_{f}^{2}(\delta t)^{2} + |\rho_{s\hspace{-.7pt}f}|^{2}\xi_{f}^{2}(\delta t)^{2}\sup_{t\in[0,T]}\E\big[V_{t}^{2}\big] \nonumber\\[1pt]
&\hspace{-1.5em}+ \big(2k_{f}+ |\rho_{s\hspace{-.7pt}f}|\xi_{f}\big)^{2}(\delta t)^{2}\sup_{0\leq n\leq N}\E\Big[\big(\hat{r}_{t_{n}}^{f}\big)^{2}\Big] + 4\xi_{f}^{2}\delta t\sup_{0\leq n\leq N}\E\big[\hat{r}_{t_{n}}^{f}\big].
\end{align}
Using Lemma \ref{LemF.3} as well as Proposition 3.4 in \citet{Cozma:2015} concludes the proof.
\end{proof}

The following lemma derives the strong mean square convergence of the stopped process.

\begin{lemma}\label{LemF.5}
Let $l>v_{0}$ and define the stopping times
\begin{equation}\label{eqF.22}
\tau_{l} = \inf\big\{t\geq 0 :\hspace{2pt} v_{t}\geq l \big\} \hspace{5pt}\text{ and }\hspace{5pt} \tau = \tau_{\kappa}\hspace{-1pt}\wedge\tau_{l}\hspace{.5pt},
\end{equation}
with $\tau_{\kappa}$ defined in \eqref{eqF.1}. Then the stopped process converges uniformly in $L^{2}$, i.e.,
\begin{equation}\label{eqF.23}
\plim_{\delta t\to0}\E\bigg[\sup_{t \in [0,T]}\big(r^{f}_{t\wedge\tau}-\hat{r}^{f}_{t\wedge\tau}\big)^{2}\bigg] = 0.
\end{equation}
\end{lemma}
\begin{proof}
From \eqref{eqF.12} and \eqref{eqF.17}, since $|r^{f}_{t}-\hat{r}^{f}_{t}|\leq|r^{f}_{t}-\tilde{r}^{f}_{t}|$, we have
\begin{align}\label{eqF.24}
\big|r^{f}_{t\wedge\tau}-\hat{r}^{f}_{t\wedge\tau}\big| &\leq
\bigg|- k_{f}\int_{0}^{t\wedge\tau}{\hspace{-3.25pt}\big(r^{f}_{u}-\hat{r}^{f}_{u}\big)du}
- k_{f}\int_{0}^{t\wedge\tau}{\hspace{-3.25pt}\big(\hat{r}^{f}_{u}-\oversymb{r}^{f}_{u}\big)du}
+ \xi_{f}\int_{0}^{t\wedge\tau}{\hspace{-3pt}\Big(\sqrt{r^{f}_{u}}-\sqrt{\hat{r}^{f}_{u}}\hspace{1.5pt}\Big)dW^{f}_{u}} \nonumber\\[2pt]
&\hspace{-1em}+ \xi_{f}\int_{0}^{t\wedge\tau}{\hspace{-3pt}\Big(\sqrt{\hat{r}^{f}_{u}}-\sqrt{\oversymb{r}^{f}_{u}}\hspace{1.5pt}\Big)dW^{f}_{u}}
- \rho_{s\hspace{-.7pt}f}\xi_{f}\int_{0}^{t\wedge\tau}{\hspace{-3pt}\sqrt{v_{u}}\hspace{1pt}\Big(\sqrt{r^{f}_{u}}-\sqrt{\hat{r}^{f}_{u}}\hspace{1.5pt}\Big)du} \nonumber\\[2pt]
&\hspace{-1em}- \rho_{s\hspace{-.7pt}f}\xi_{f}\int_{0}^{t\wedge\tau}{\hspace{-3pt}\sqrt{v_{u}}\hspace{1pt}\Big(\sqrt{\hat{r}^{f}_{u}}-\sqrt{\oversymb{r}^{f}_{u}}\hspace{1.5pt}\Big)du}
- \rho_{s\hspace{-.7pt}f}\xi_{f}\int_{0}^{t\wedge\tau}{\hspace{-3pt}\sqrt{\oversymb{r}^{f}_{u}}\hspace{1pt}\Big(\sqrt{v_{u}}-\sqrt{\oversymb{V}_{\hspace{-2.5pt}u}}\hspace{1.5pt}\Big)du}\hspace{1pt}\bigg|.
\end{align}
Fix $t\in[0,T]$. Squaring both sides and using Cauchy's inequality, then taking expectations and using Doob's martingale inequality and Fubini's theorem, we get
\begin{align}\label{eqF.25}
\E\bigg[\sup_{s\in[0,t]}\big(r^{f}_{s\wedge\tau}-\hat{r}^{f}_{s\wedge\tau}\big)^{2}\bigg] &\leq
7k_{f}^{2}T\int_{0}^{t}{\E\Big[\big(r^{f}_{u}-\hat{r}^{f}_{u}\big)^{2}\Ind_{u<\hspace{1pt}\tau}\hspace{-1pt}\Big]du}
+ 7k_{f}^{2}T\int_{0}^{T}{\E\Big[\big(\hat{r}^{f}_{u}-\oversymb{r}^{f}_{u}\big)^{2}\Big]du} \nonumber\\[1pt]
&\hspace{-5.2em}+ 28\xi_{f}^{2}\int_{0}^{t}{\E\bigg[\Big(\sqrt{r^{f}_{u}}-\sqrt{\hat{r}^{f}_{u}}\hspace{1.5pt}\Big)^{\hspace{-1pt}2}\Ind_{u<\hspace{1pt}\tau}\hspace{-1pt}\bigg]du}
+ 28\xi_{f}^{2}\int_{0}^{T}{\E\Big[\big|\hat{r}^{f}_{u}-\oversymb{r}^{f}_{u}\big|\Big]du} \nonumber\\[2pt]
&\hspace{-5.2em}+ 7\rho_{s\hspace{-.7pt}f}^{2}\xi_{f}^{2}T\int_{0}^{t}{\E\bigg[v_{u}\Big(\sqrt{r^{f}_{u}}-\sqrt{\hat{r}^{f}_{u}}\hspace{1.5pt}\Big)^{\hspace{-1pt}2}\Ind_{u<\hspace{1pt}\tau}\hspace{-1pt}\bigg]du}
+ 7\rho_{s\hspace{-.7pt}f}^{2}\xi_{f}^{2}T\int_{0}^{T}{\E\Big[\oversymb{r}^{f}_{u}\big|v_{u}-\oversymb{V}_{\hspace{-2.5pt}u}\big|\Big]du} \nonumber\\[3pt]
&\hspace{-5.2em}+ 7\rho_{s\hspace{-.7pt}f}^{2}\xi_{f}^{2}T\int_{0}^{T}{\E\Big[v_{u}\big|\hat{r}^{f}_{u}-\oversymb{r}^{f}_{u}\big|\Ind_{u<\hspace{1pt}\tau}\hspace{-1pt}\Big]du}.
\end{align}
On the other hand, we know that $|r^{f}_{u}-\hat{r}^{f}_{u}|\Ind_{u<\hspace{.5pt}\tau} \leq |r^{f}_{u\wedge\tau}-\hat{r}^{f}_{u\wedge\tau}|$ and
\begin{equation}\label{eqF.26}
\Big(\sqrt{r^{f}_{u}}-\sqrt{\hat{r}^{f}_{u}}\hspace{1.5pt}\Big)^{\hspace{-1pt}2}\Ind_{u<\hspace{1pt}\tau} \leq
\Big(\sqrt{r^{f}_{u\wedge\tau}}-\sqrt{\hat{r}^{f}_{u\wedge\tau}}\hspace{1.5pt}\Big)^{\hspace{-1pt}2} \leq
\kappa\big(r^{f}_{u\wedge\tau}-\hat{r}^{f}_{u\wedge\tau}\big)^{2}.
\end{equation}
Substituting back into \eqref{eqF.25} with \eqref{eqF.26}, we arrive at the following inequality:
\begin{align}\label{eqF.27}
\E\bigg[\sup_{s\in[0,t]}\big(r^{f}_{s\wedge\tau}-\hat{r}^{f}_{s\wedge\tau}\big)^{2}\bigg] &\leq
\big(7k_{f}^{2}T+28\xi_{f}^{2}\kappa+7\rho_{s\hspace{-.7pt}f}^{2}\xi_{f}^{2}Tl\kappa\big)\int_{0}^{t}{\E\bigg[\sup_{s\in[0,u]}\big(r^{f}_{s\wedge\tau}-\hat{r}^{f}_{s\wedge\tau}\big)^{2}\bigg]du} \nonumber\\[3pt]
&\hspace{-3.5em}+ 7k_{f}^{2}T^{2}\sup_{u\in[0,T]}\E\Big[\big(\hat{r}^{f}_{u}-\oversymb{r}^{f}_{u}\big)^{2}\Big]
+ \big(28\xi_{f}^{2}T+7\rho_{s\hspace{-.7pt}f}^{2}\xi_{f}^{2}T^{2}l\big)\sup_{u\in[0,T]}\E\Big[\big|\hat{r}^{f}_{u}-\oversymb{r}^{f}_{u}\big|\Big] \nonumber\\[0pt]
&\hspace{-3.5em}+ 7\rho_{s\hspace{-.7pt}f}^{2}\xi_{f}^{2}T^{2}\sup_{u\in[0,T]}\E\Big[\big(\oversymb{r}^{f}_{u}\big)^{2}\Big]^{\frac{1}{2}}\sup_{u\in[0,T]}\E\Big[\big|v_{u}-\oversymb{V}_{\hspace{-2.5pt}u}\big|^{2}\Big]^{\frac{1}{2}}.
\end{align}
The convergence to zero of the last three terms on the right-hand side of \eqref{eqF.27} follows from Lemmas \ref{LemF.3} and \ref{LemF.4}, and Proposition 3.5 in \citet{Cozma:2015}. The conclusion follows from a simple application of Gronwall's inequality.
\end{proof}

From Lemma \ref{LemF.4}, we know that in order to establish the strong mean square convergence of $\oversymb{r}^{f}$, it suffices to prove this for $\hat{r}^{f}$, since
\begin{equation}\label{eqF.28}
\big|r^{f}_{t}-\oversymb{r}^{f}_{t}\big|^{2} \leq 2\big|r^{f}_{t}-\hat{r}^{f}_{t}\big|^{2} + 2\big|\hat{r}^{f}_{t}-\oversymb{r}^{f}_{t}\big|^{2},\hspace{5pt} \forall\hspace{.5pt}t\in[0,T].
\end{equation}

\begin{lemma}\label{LemF.6} If $2k_{f}\theta_{f}>\xi_{f}^{2}$, then the process $\hat{r}^{f}$ converges strongly in $L^{2}$, i.e.,
\begin{equation}\label{eqF.29}
\plim_{\delta t\to0}\hspace{1pt}\sup_{t\in[0,T]}\E\Big[\big|r^{f}_{t}-\hat{r}^{f}_{t}\big|^{2}\Big] = 0.
\end{equation}
\end{lemma}
\begin{proof}
Fix $\kappa>(r^{f}_{0})^{-1}$, $l>v_{0}$, and recall the definition of the stopping time $\tau$ from \eqref{eqF.22}. Since $r^{f}$ and $\hat{r}^{f}$ are non-negative,
\begin{align*}
\sup_{t\in[0,T]}\E\Big[\big|r^{f}_{t}-\hat{r}^{f}_{t}\big|^{2}\Big] &\leq \sup_{t\in[0,T]}\E\Big[\big|r^{f}_{t}-\hat{r}^{f}_{t}\big|^{2}\Ind_{\tau\leq\hspace{1pt}t}\hspace{-1pt}\Big] + \sup_{t\in[0,T]}\E\Big[\big|r^{f}_{t}-\hat{r}^{f}_{t}\big|^{2}\Ind_{t<\hspace{1pt}\tau}\hspace{-1pt}\Big] \\[1pt]
&\hspace{-1.75em}\leq \sup_{t\in[0,T]}\E\Big[\big(r^{f}_{t}\big)^{2}\Ind_{\tau\leq\hspace{1pt}T}\hspace{-1pt}\Big] + \sup_{t\in[0,T]}\E\Big[\big(\hat{r}^{f}_{t}\big)^{2}\Ind_{\tau\leq\hspace{1pt}T}\hspace{-1pt}\Big] + \sup_{t\in[0,T]}\E\Big[\big|r^{f}_{t\wedge\tau}-\hat{r}^{f}_{t\wedge\tau}\big|^{2}\Big].
\end{align*}
Since $\Ind_{\tau\leq\hspace{1pt}T}\leq\Ind_{\tau_{\kappa}\leq\hspace{1pt}T}+\Ind_{\tau_{l}\leq\hspace{1pt}T}$ and applying Cauchy's inequality, we get
\begin{align}\label{eqF.30}
\sup_{t\in[0,T]}\E\Big[\big|r^{f}_{t}-\hat{r}^{f}_{t}\big|^{2}\Big] &\leq
\bigg\{\sup_{t\in[0,T]}\E\Big[\big(r^{f}_{t}\big)^{4}\Big]^{\frac{1}{2}} + \sup_{t\in[0,T]}\E\Big[\big(\hat{r}^{f}_{t}\big)^{4}\Big]^{\frac{1}{2}}\bigg\}\Prob\big(\tau_{\kappa}\leq T\big)^{\frac{1}{2}} \nonumber\\[0pt]
&+ \bigg\{\sup_{t\in[0,T]}\E\Big[\big(r^{f}_{t}\big)^{4}\Big]^{\frac{1}{2}} + \sup_{t\in[0,T]}\E\Big[\big(\hat{r}^{f}_{t}\big)^{4}\Big]^{\frac{1}{2}}\bigg\}\Prob\big(\tau_{l}\leq T\big)^{\frac{1}{2}} \nonumber\\[3pt]
&+ \sup_{t\in[0,T]}\E\Big[\big|r^{f}_{t\wedge\tau}-\hat{r}^{f}_{t\wedge\tau}\big|^{2}\Big].
\end{align}
On the other hand, using Markov's inequality, we obtain an upper bound
\begin{equation}\label{eqF.31}
\Prob\big(\tau_{l}\leq T\big) \leq \Prob\bigg(\sup_{t\in[0,T]}v_{t}\geq l\bigg) \leq \frac{1}{l}\E\bigg[\sup_{t\in[0,T]}v_{t}\bigg].
\end{equation}
However, the expectation on the right-hand side is clearly finite by the Burkholder-Davis-Gundy inequality. Taking the limit as $\delta t \to 0$ in \eqref{eqF.30} and employing Lemmas \ref{LemF.1} to \ref{LemF.3} and \ref{LemF.5}, since $\kappa$ and $l$ can be made arbitrarily large, leads to the conclusion.
\end{proof}

\section{Proof of Proposition \ref{Prop3.8}}\label{sec:aux3.8}

Note that $\forall\hspace{.5pt} t\in[t_{n},t_{n+1})$ and $\forall\hspace{0pt} j\in\{2,3,4\}$, since $\oversymb{V}$ is piecewise constant,
\begin{align}\label{eq3.45}
\int_{0}^{t}{\sqrt{\oversymb{V}_{\hspace{-2.5pt}u}}\,\frac{\delta W^{j}_{u}}{\delta t}\,du}
&= \sum_{i=0}^{n-1}{\sqrt{\oversymb{V}_{\hspace{-2.5pt}t_{i}}}\,\frac{W^{j}_{t_{i+1}}-W^{j}_{t_{i}}}{\delta t}\hspace{1pt}\delta t + \sqrt{\oversymb{V}_{\hspace{-2.5pt}t_{n}}}\,\frac{W^{j}_{t_{n+1}}-W^{j}_{t_{n}}}{\delta t}\hspace{.5pt}(t-t_{n})} \nonumber\\[2pt]
&\hspace{-2.5em}= \int_{0}^{t}{\sqrt{\oversymb{V}_{\hspace{-2.5pt}u}}\,dW_{u}^{j}} + \sqrt{\oversymb{V}_{\hspace{-2.5pt}t}}\left[\frac{t-t_{n}}{\delta t}\big(W_{t_{n+1}}^{j}-W_{t}^{j}\big) - \frac{t_{n+1}-t}{\delta t}\big(W_{t}^{j}-W_{t_{n}}^{j}\big)\right]\hspace{-.5pt}.
\end{align}
For convenience, $\forall\hspace{.5pt} t\in[t_{n},t_{n+1})$ and $\forall\hspace{0pt} j\in\{2,3,4\}$, we define
\begin{equation}\label{eq3.46}
Z^{j}_{t} = \frac{t-t_{n}}{\delta t}\big(W_{t_{n+1}}^{j}-W_{t}^{j}\big) - \frac{t_{n+1}-t}{\delta t}\big(W_{t}^{j}-W_{t_{n}}^{j}\big).
\end{equation}
Substituting back into \eqref{eq3.43} with \eqref{eq3.45} and \eqref{eq3.46}, we obtain
\begin{equation}\label{eq3.47}
X_{t} = x_{0} + \int_{0}^{t}{\Big(\oversymb{r}^{d}_{u}-\oversymb{r}^{f}_{u}-\frac{1}{2}\hspace{1pt}\oversymb{V}_{\hspace{-2.5pt}u}\Big)du} + \int_{0}^{t}{\sqrt{\oversymb{V}_{\hspace{-2.5pt}u}}\,dW_{u}^{s}} + \sum_{j=2}^{4}{a_{1\hspace{-0.5pt}j}\hspace{.5pt}\sqrt{\oversymb{V}_{\hspace{-2.5pt}t}}\,Z^{j}_{t}}.
\end{equation}
The absolute difference between the original and the discretized log-processes is thus
\begin{align}\label{eq3.48}
\big|x_{t}-X_{t}\big| &= \bigg|\int_{0}^{t}{\big(r^{d}_{u}-\oversymb{r}^{d}_{u}\big)du}-\int_{0}^{t}{\big(r^{f}_{u}-\oversymb{r}^{f}_{u}\big)du}-\frac{1}{2}\int_{0}^{t}{\big(v_{u}-\oversymb{V}_{\hspace{-2.5pt}u}\big)du} \nonumber\\[2pt]
&+ \int_{0}^{t}{\Big(\sqrt{v_{u}}-\sqrt{\oversymb{V}_{\hspace{-2.5pt}u}}\Big)\hspace{.5pt}dW_{u}^{s}} - \sum_{j=2}^{4}{a_{1\hspace{-0.5pt}j}\hspace{.5pt}\sqrt{\oversymb{V}_{\hspace{-2.5pt}t}}\,Z^{j}_{t}}\,\bigg|\,.
\end{align}
Squaring both sides of \eqref{eq3.48}, applying the Cauchy-Schwarz inequality, taking the supremum over all $t\in[0,T]$, and then using Cauchy's inequality for all Riemann integrals leads to
\begin{align}\label{eq3.49}
\sup_{t \in [0,T]}\big|x_{t}-X_{t}\big|^{2} &\leq 7T\int_{0}^{T}{\big(r^{d}_{u}-\oversymb{r}^{d}_{u}\big)^{2}du} + 7T\int_{0}^{T}{\big(r^{f}_{u}-\oversymb{r}^{f}_{u}\big)^{2}du} + \frac{7}{4}\hspace{1pt}T\int_{0}^{T}{\big(v_{u}-\oversymb{V}_{\hspace{-2.5pt}u}\big)^{2}du} \nonumber\\[0pt]
&+ 7\sup_{t \in [0,T]}\left|\int_{0}^{t}{\Big(\sqrt{v_{u}}-\sqrt{\oversymb{V}_{\hspace{-2.5pt}u}}\Big)\hspace{.5pt}dW_{u}^{s}}\right|^{2} \hspace{-1pt}+ 7\sum_{j=2}^{4}{a_{1\hspace{-0.5pt}j}^{2}\sup_{t \in [0,T]}\hspace{-1pt}V_{t}\hspace{1pt}\sup_{t \in [0,T]}\hspace{-1pt}\big|Z_{t}^{j}\big|^{2}}.
\end{align}
We used the fact that $\sup_{t\in[0,T]}\oversymb{V}_{\hspace{-2.5pt}t}\leq\sup_{t\in[0,T]}V_{t}$, with $V$ defined as in \eqref{eq2.6}. Taking expectations and employing Fubini's theorem, H\"older's inequality, Doob's martingale inequality and the It\^o isometry, we derive
\begin{align}\label{eq3.50}
\E\bigg[\sup_{t \in [0,T]}\big|x_{t}-X_{t}\big|^{2}\bigg] &\leq 7T^{2}\sup_{t \in \left[0,T\right]}\E\left[\big|r^{d}_{t}-\oversymb{r}^{d}_{t}\big|^{2}\right] + 7T^{2}\sup_{t \in \left[0,T\right]}\E\left[\big|r^{f}_{t}-\oversymb{r}^{f}_{t}\big|^{2}\right] \nonumber\\[2pt]
&+ \frac{7}{4}\hspace{1pt}T^{2}\sup_{t \in [0,T]}\E\left[\big|v_{t}-\oversymb{V}_{\hspace{-2.5pt}t}\big|^{2}\right] + 28\hspace{.5pt}T\sup_{t \in [0,T]}\E\Big[\big|v_{t}-\oversymb{V}_{\hspace{-2.5pt}t}\big|\Big] \nonumber\\
&+ 7\sum_{j=2}^{4}{a_{1\hspace{-0.5pt}j}^{2}\hspace{1pt}\E\bigg[\sup_{t \in [0,T]}V_{t}^{2}\bigg]^{1/2}\E\bigg[\sup_{t \in [0,T]}\big|Z_{t}^{j}\big|^{4}\bigg]^{1/2}}.
\end{align}
The convergence as $\delta t\to0$ of the first four terms on the right-hand side of \eqref{eq3.50} follows from Proposition \ref{Prop3.8.0}, and Proposition 3.5 in \citet{Cozma:2015}. Next, integrating the time continuous auxiliary variance process defined in \eqref{eq2.5} leads to
\begin{equation}\label{eq3.51}
\tilde{v}_{t} = v_{0} + k\int_{0}^{t}{\big(\theta-\oversymb{V}_{\hspace{-2.5pt}u}\big)du} + \xi\int_{0}^{t}{\sqrt{\oversymb{V}_{\hspace{-2.5pt}u}}\hspace{1.5pt}dW_{u}^{4}}.
\end{equation}
However, $V=\max\{0,\tilde{v}\}\leq\left|\tilde{v}\right|$ and, using Cauchy's inequality, Fubini's theorem and Doob's inequality, we find an upper bound
\begin{equation}\label{eq3.52}
\E\bigg[\sup_{t \in [0,T]}V_{t}^{2}\bigg] \leq 3\big(v_{0}+k\theta T\big)^{2} + 3k^{2}T^{2}\sup_{t \in [0,T]}\E\big[V_{t}^{2}\big] + 12\xi^{2}T\sup_{t \in [0,T]}\E\big[V_{t}\big].
\end{equation}
The uniform boundedness of the second moment of the FTE discretization for the variance as $\delta t\to0$ follows from Proposition 3.4 in \citet{Cozma:2015}. Finally, employing the definition in \eqref{eq3.46}, we bound the term inside the last expectation in \eqref{eq3.50} from above.
\begin{align}\label{eq3.53}
\sup_{t\in[0,T]}\big|Z_{t}^{j}\big|^{4}
&= \sup_{0\leq n<N}\hspace{1pt}\sup_{t_{n}\leq t<t_{n+1}}\bigg|\frac{t-t_{n}}{\delta t}\big(W_{t_{n+1}}^{j}-W_{t_{n}}^{j}\big) - \big(W_{t}^{j}-W_{t_{n}}^{j}\big)\bigg|^{4} \nonumber\\[1pt]
&\leq 8\sup_{0\leq n<N}\hspace{1pt}\sup_{t_{n}\leq t<t_{n+1}}\bigg[\left(\frac{t-t_{n}}{\delta t}\right)^{\hspace{-2pt}4}\big|W_{t_{n+1}}^{j}-W_{t_{n}}^{j}\big|^{4} + \big|W_{t}^{j}-W_{t_{n}}^{j}\big|^{4}\bigg] \nonumber\\[6pt]
&\leq 16\sup_{0\leq n<N}\hspace{1pt}\sup_{t_{n}\leq t<t_{n+1}}\big|W_{t}^{j}-W_{t_{n}}^{j}\big|^{4} \nonumber\\[6pt]
&\leq 16\sup_{t\in[0,T]}\big|W_{t}^{j}-W_{\delta t\floor{t/\delta t}}^{j}\big|^{4}.
\end{align}
However, moments of the Euler modulus of continuity of a Brownian motion converge to $0$ as $\delta t\to0$ \citep{Fischer:2009}, which concludes the proof. \qed

\section{Proof of Proposition \ref{Prop3.9}}\label{sec:aux3.9}

For fixed, positive numbers $\epsilon$ and $\gamma$ such that $\log\left(1+\epsilon\right)>\gamma$, define the set
\begin{equation}\label{eq3.55}
B_{\epsilon,\gamma} = \big\{x \in \mathbb{R} \,\big|\, \exists\hspace{1pt} y \in \mathbb{R}\hspace{-1pt}:\, \left|x-y\right|<\gamma \,\text{ and }\, \left|e^{x}-e^{y}\right|\geq\epsilon \big\}.
\end{equation}
However, since the exponential function is strictly increasing,
\begin{equation*}
x \in B_{\epsilon,\gamma} \hspace{2pt}\Leftrightarrow\hspace{2pt}
\exists\hspace{1pt} y \in (x-\gamma, x+\gamma)\hspace{-1pt}:\, e^{\min\{x,y\}}\big(e^{|x-y|}-1\big)\geq\epsilon \hspace{2pt}\Leftrightarrow\hspace{2pt}
e^{x}\big(e^{\gamma}-1\big)>\epsilon\hspace{.5pt}.
\end{equation*}
Hence,
\begin{equation}\label{eq3.56}
B_{\epsilon,\gamma} = \big(\hspace{.5pt}a(\epsilon,\gamma),\hspace{.5pt}+\infty\hspace{.5pt}\big),\hspace{1pt}\text{ where }\hspace{1pt} a(\epsilon,\gamma)=\log\Big(\frac{\epsilon}{e^{\gamma}-1}\Big)>0.
\end{equation}
We have the following string of inclusions of events,
\begin{align}\label{eq3.57}
\bigg\{\hspace{-1pt}\sup_{t \in [0,T]}\hspace{-1pt}\big|S_{t}\hspace{-1pt}-\hspace{-1pt}\hspace{1.5pt}\oversymb{\hspace{-1.5pt}S}_{t}\big|\hspace{-1pt}>\hspace{-1pt}\epsilon\bigg\}
&\subseteq \bigg\{\hspace{-1pt}\sup_{t \in [0,T]}\hspace{-1pt}\big|x_{t}\hspace{-1pt}-\hspace{-1pt}X_{t}\big|\hspace{-1pt}\geq\hspace{-1pt}\gamma\bigg\} \cup \bigg\{\hspace{-1pt}\sup_{t \in [0,T]}\hspace{-1pt}\big|x_{t}\hspace{-1pt}-\hspace{-1pt}X_{t}\big|\hspace{-1pt}<\hspace{-1pt}\gamma, \sup_{t \in [0,T]}\hspace{-1pt}\big|e^{x_{t}}\hspace{-1.5pt}-\hspace{-1pt}e^{X_{t}}\big|\hspace{-1pt}>\hspace{-1pt}\epsilon\bigg\} \nonumber\\[1pt]
&\subseteq \bigg\{\hspace{-1pt}\sup_{t \in [0,T]}\hspace{-1pt}\big|x_{t}\hspace{-1pt}-\hspace{-1pt}X_{t}\big|\hspace{-1pt}\geq\hspace{-1pt}\gamma\bigg\} \cup \bigg\{\exists\hspace{1pt} t \in [0,T]\hspace{-1pt}:\, x_{t} \in B_{\epsilon,\gamma}\bigg\} \nonumber\\[1pt]
&\subseteq \bigg\{\hspace{-1pt}\sup_{t \in [0,T]}\hspace{-1pt}\big|x_{t}\hspace{-1pt}-\hspace{-1pt}X_{t}\big|\hspace{-1pt}\geq\hspace{-1pt}\gamma\bigg\} \cup \bigg\{\sup_{t \in [0,T]}x_{t}>a(\epsilon,\gamma)\bigg\}.
\end{align}
In terms of probabilities of events, the previous inclusion becomes:
\begin{equation}\label{eq3.58}
\Prob\bigg(\sup_{t \in [0,T]}\big|S_{t}-\hspace{1.5pt}\oversymb{\hspace{-1.5pt}S}_{t}\big|>\epsilon\bigg) \leq \Prob\bigg(\sup_{t \in [0,T]}\big|x_{t}-X_{t}\big|\geq\gamma\bigg) + \Prob\bigg(\sup_{t \in [0,T]}x_{t}>a(\epsilon,\gamma)\bigg).
\end{equation}
The convergence in probability of the log-process is a consequence of Proposition \ref{Prop3.8} and Markov's inequality. Therefore, all that we have left to prove is that the second probability on the right-hand side of \eqref{eq3.58} can be made arbitrarily small. However, if we fix $\epsilon>0$ and vary $\gamma$, then $\lim_{\gamma \to 0}a(\epsilon,\gamma) = \infty$. A simple application of Markov's inequality leads to
\begin{equation}\label{eq3.59}
\Prob\bigg(\sup_{t \in [0,T]}x_{t}>a(\epsilon,\gamma)\bigg) \leq \Prob\bigg(\sup_{t \in [0,T]}\left|x_{t}\right|>a(\epsilon,\gamma)\bigg) \leq \frac{1}{a(\epsilon,\gamma)}\E\bigg[\sup_{t \in [0,T]}\left|x_{t}\right|\bigg].
\end{equation}
On the other hand, using Jensen's inequality and Doob's martingale inequality,
\begin{align*}
\E\hspace{-.5pt}\bigg[\sup_{t \in [0,T]}\left|x_{t}\right|\bigg] &\leq \left|x_{0}\right| + T\hspace{-1pt}\sup_{t \in [0,T]}\hspace{-.5pt}\E\hspace{-.5pt}\big[r^{d}_{t}\big] \hspace{-.5pt}+ T\hspace{-1pt}\sup_{t \in [0,T]}\hspace{-.5pt}\E\hspace{-.5pt}\big[r^{f}_{t}\big] \hspace{-.5pt}+ \frac{T}{2}\sup_{t \in [0,T]}\hspace{-.5pt}\E\hspace{-.5pt}\big[v_{t}\big] \hspace{-.5pt}+ 2\sqrt{T}\hspace{-1pt}\sup_{t \in [0,T]}\hspace{-.5pt}\E\hspace{-.5pt}\big[v_{t}\big]^{\hspace{-.5pt}\frac{1}{2}}.
\end{align*}
However, the right-hand side is finite because the moments of the square root process are bounded and from Proposition \ref{LemF.2}, which concludes the proof. \qed

\section{Proof of Theorem \ref{Thm3.10}}\label{sec:aux3.10}

Fix $\epsilon > 0$ and define the event $A = \Big\{\big|S_{t}-\hspace{1.5pt}\oversymb{\hspace{-1.5pt}S}_{t}\big| > \epsilon\Big\}$. Since $S$ and $\hspace{1.5pt}\oversymb{\hspace{-1.5pt}S}$ are non-negative,
\begin{align}\label{eq3.62}
\sup_{t \in [0,T]}\E\Big[\big|S_{t}-\hspace{1.5pt}\oversymb{\hspace{-1.5pt}S}_{t}\big|^{\alpha}\Big]
&\leq \sup_{t \in [0,T]}\E\Big[\big|S_{t}-\hspace{1.5pt}\oversymb{\hspace{-1.5pt}S}_{t}\big|^{\alpha}\Ind_{A^{c}}\Big]
+ \sup_{t \in [0,T]}\E\Big[\big|S_{t}-\hspace{1.5pt}\oversymb{\hspace{-1.5pt}S}_{t}\big|^{\alpha}\Ind_{A}\Big] \nonumber\\[3pt]
&\leq \epsilon^{\alpha} + \sup_{t \in [0,T]}\E\big[S_{t}^{\alpha}\Ind_{A}\big] + \sup_{t \in [0,T]}\E\big[\hspace{.5pt}\hspace{1.5pt}\oversymb{\hspace{-1.5pt}S}_{\hspace{-.5pt}t}^{\alpha}\Ind_{A}\big].
\end{align}
Let $\alpha\hspace{-1pt}<\hspace{-1pt}\omega\hspace{-1pt}<\hspace{-1pt}\min\left\{\alpha_{1},\hspace{.5pt}\alpha_{2}\right\}$ and apply H\"older's inequality to the two expectations on the right-hand side of \eqref{eq3.62} with the pair $(p,q) = \big(\frac{\omega}{\alpha},\frac{\omega}{\omega-\alpha}\big)$. Hence,
\begin{align*}
\sup_{t \in [0,T]}\E\Big[\big|S_{t}-\hspace{1.5pt}\oversymb{\hspace{-1.5pt}S}_{t}\big|^{\alpha}\Big]
&\leq \epsilon^{\alpha} + \bigg\{\sup_{t \in [0,T]}\E\big[S_{t}^{\omega}\big]^{\hspace{-.5pt}\frac{\alpha}{\omega}}
+ \sup_{t \in [0,T]}\E\big[\hspace{.5pt}\hspace{1.5pt}\oversymb{\hspace{-1.5pt}S}_{\hspace{-.5pt}t}^{\omega}\big]^{\hspace{-.5pt}\frac{\alpha}{\omega}}\bigg\}\sup_{t \in [0,T]}\Prob\Big(\big|S_{t}-\hspace{1.5pt}\oversymb{\hspace{-1.5pt}S}_{t}\big|>\epsilon\Big)^{\hspace{-.5pt}1-\frac{\alpha}{\omega}}.
\end{align*}
The convergence of $\hspace{1.5pt}\oversymb{\hspace{-1.5pt}S}$ in probability is a consequence of Proposition \ref{Prop3.9}. Finally, employing Propositions \ref{Prop3.2} and \ref{Prop3.4} and then taking $\epsilon$ sufficiently small concludes the proof. \qed

\end{appendices}

\end{document}